\documentclass[11pt]{article}
\pdfoutput=1
\usepackage[margin=1in]{geometry}
\usepackage{amsmath,amsthm,enumerate,graphicx,amsfonts}
\usepackage{algorithm}
\usepackage{algorithmic}
\usepackage[utf8]{inputenc}
\usepackage{color}

\theoremstyle{plain}

\newcommand{\A}{\mathcal{A}}
\newcommand{\F}{\mathcal{F}}
\newcommand{\G}{\mathcal{G}}
\renewcommand{\L}{\mathcal{L}}
\newcommand{\C}{\mathcal{C}}
\newcommand{\D}{\mathcal{D}}

\renewcommand{\S}{\mathcal{S}}
\newcommand{\I}{\mathcal{I}}
\newcommand{\J}{\mathcal{J}}
\newcommand{\T}{\mathcal{T}}

\newcommand{\U}{\mathcal{U}}
\newcommand{\Z}{\mathcal{Z}}

\newcommand{\eps}{\epsilon}

\newcommand{\poly}{\mathrm{poly}}

\newtheorem{definition}{Definition}[section]
\newtheorem{theorem}{Theorem}[section]
\newtheorem{lemma}[theorem]{Lemma}

\newtheorem{corollary}[theorem]{Corollary}
\newtheorem{claim}[theorem]{Claim}

\begin{document}

\title{An Improved Approximation for $k$-median, and Positive Correlation in Budgeted Optimization}

\author{
Jarosław Byrka\thanks{Institute of Computer Science, University of Wrocław, Poland. Supported by NCN 2015/18/E/ST6/00456 grant.}\and
Thomas Pensyl\thanks{Department of Computer Science, University of Maryland, College Park, MD 20742. Research supported in part by NSF Awards CNS 1010789 and CCF 1422569.} \and
Bartosz Rybicki\thanks{Institute of Computer Science, University of Wrocław, Poland. Supported by NCN 2012/07/N/ST6/03068 grant.} \and
Aravind Srinivasan\thanks{Department of Computer Science and Instute for Advanced Computer Studies, University of Maryland, College Park, MD 20742.
Research supported in part by NSF Awards CNS 1010789 and CCF 1422569, and by a gift from Adobe, Inc. Email: 
\texttt{srin@cs.umd.edu}.} 
\and Khoa Trinh\footnotemark[2]}

\date{}

\maketitle

\begin{abstract}
Dependent rounding is a useful technique for optimization problems with hard budget constraints. This framework naturally leads to \emph{negative correlation} properties. However, what if an application naturally calls for dependent rounding on the one hand, and desires \emph{positive} correlation on the other? More generally, we develop
algorithms that guarantee the known properties of dependent rounding, but also have nearly best-possible behavior -- near-independence, which generalizes positive correlation -- on ``small" subsets of the variables. The recent breakthrough of Li \& Svensson for the classical $k$-median problem has to handle positive correlation in certain dependent-rounding settings, and does so implicitly. We improve upon Li-Svensson's approximation ratio for $k$-median from $2.732 + \epsilon$ to $2.675 + \epsilon$ by developing an algorithm that improves upon various aspects of their work.  Our dependent-rounding approach helps us improve the dependence of the runtime on the parameter $\epsilon$ from Li-Svensson's $N^{O(1/\epsilon^2)}$ to $N^{O((1/\epsilon) \log(1/\epsilon))}$. 
\end{abstract}

\section{Introduction and High-Level Details}
We consider two notions in combinatorial optimization: a concrete problem (the classical $k$-median problem) and a formulation of new types of distributions (generalizations of dependent-rounding techniques); the breakthrough of Li \& Svensson on the former \cite{li_svensson} uses special cases of the latter. We improve the approximation ratio of \cite{li_svensson} for the former, and develop efficient samplers for the latter -- which, in particular, show that such distributions exist; we then combine the two to improve the run-time of our approximation algorithm for $k$-median. The ideas developed here also lead to optimal approximations for certain budgeted satisfiability problems, of which the classical budgeted set-cover problem is a special case. We discuss these contributions in further detail in Sections \ref{subsec:k-med}, \ref{subsec:dep-round}, and \ref{subsec:budgeted-sat}. 

\subsection{The $k$-median problem}
\label{subsec:k-med}
Metric $k$-median is a fundamental location problem in combinatorial optimization. 
Herein, we are given a set $\F$ of facilities, a set $\J$ of clients, a budget $k$, and a symmetric distance-metric $d$ on $\F \cup \J$. The goal is to open a subset of at most $k$ facilities in $\F$ such that the total distance (or connection cost) from each client to its closest opened facility is minimized. Note that the only decision is which subset of the facilities to open. 
This problem is known to be $NP$-hard, so there has been much work done on designing approximations with provable performance guarantees; indeed, virtually every major technique in approximation algorithms has been used and/or developed for this problem and its variants. 

Convex combinations of \emph{two} integral solutions, with the corresponding convex combination of the number of open facilities being $k$, will be particularly useful for us:

\begin{definition}
\label{defn:bi-point}
\textbf{(Bi-point solution)}  Given a $k$-median instance $\I$, a bi-point solution is a pair $\F_1, \F_2 \subseteq \F$ such that $|\F_1| \leq k \leq |\F_2|$,  along with reals $a, b \geq 0$ with $a + b = 1$ such that $a|\F_1| + b|\F_2| = k$. (That is, the convex combination of the two ``solutions" is feasible for the natural LP relaxation of $\I$.) The cost of this bi-point solution is defined as $aD_1 + bD_2$, where $D_1$ and $D_2$ are the connection costs of $\F_1$ and $\F_2$ respectively.
\label{def:bipoint}
\end{definition}

Charikar, Guha, Tardos, and Shmoys used LP-rounding to achieve the first constant factor approximation ratio of $6\frac{2}{3}$ \cite{charikar}. Then, Jain and Vazirani \cite{jain_vazirani} applied Lagrangian Relaxation to remove the hard constraint of opening at most $k$ facilities, effectively reducing the problem to an easier version known as the Uncapacitated Facility Location (UFL) problem. Using this technique together with primal-dual methods for UFL, they first find a bi-point solution, losing a factor of 3. They then round this bi-point solution to an integral feasible solution losing another multiplicative factor of 2, yielding a $6$-approximation. Later, Jain, Mahdian, and 
Saberi (JMS) improved the approximation ratio of constructing the bi-point solution to $2$, resulting in a $4$-approximation \cite{jms}. Following this, a local-search-based $(3+\eps)$-approximation algorithm was developed by Arya et al.\ \cite{arya}.

Recently, Li and Svensson's breakthrough work gave a $(1+\sqrt{3}+\eps)$-approximation algorithm for $k$-median \cite{li_svensson}. To accomplish this, they defined an $\alpha$-\emph{pseudo-approximation} algorithm to be one that is an $\alpha$-approximation which, however, opens $k+O(1)$ facilities, and showed -- very surprisingly -- how to use such an algorithm as a blackbox to construct a true $(\alpha+\eps)$-approximation algorithm. They then took advantage of this by giving a bi-point rounding algorithm which opens $k+O(1)$ facilities, but loses a factor of $\frac{1+\sqrt{3}}{2} +\eps$ instead of the previous 2. Together with the factor of 2 lost during the JMS bi-point construction algorithm, this yields the final approximation ratio. Letting $N$ denote the input-size, the runtime of \cite{li_svensson} is $N^{O(1/\epsilon^2)}$. 

In this paper\footnote{In the conferrence version \cite{ByrkaPRST15} we claimed that we can also improve the cost of the bi-point solution, which was not correct. More details you can find in Appendix \ref{sec:jms_scaling}.}, we improve the bi-point rounding step to give an algorithm for $k$-median with improved approximation ratio and runtime. Section \ref{sec:bipoint-round} presents an improved approximation for rounding bi-point solutions; we obtain $1.3371 + \epsilon$ instead of $\frac{1+\sqrt{3}}{2} +\eps \sim 1.366 + \eps$. Section \ref{sec:dep-round} develops our dependent rounding technique, a specific application of which reduces the dependence of the run-time on $\epsilon$ from $N^{O(1/\epsilon^2)}$ as in \cite{li_svensson}, to $N^{O((1/\epsilon) \log(1/\epsilon))}$. 

How can we round a bi-point solution to a feasible one? 
We present a rounding algorithm that obtains a factor of $1.3371+\eps$. This yields a $2 \times 1.3371 + \epsilon \sim (2.675+\eps)$-approximation algorithm for $k$-median, an improvement over Li and Svensson's $(2.733+\eps)$. 
We analyze the worst-case instances of Li and Svensson's approach, the structure of which leads us to the new algorithm. 

Let the given bi-point solution be parametrized by $\F_1, \F_2, a, b, D_1,$ and $D_2$ as in Definition~\ref{defn:bi-point}. As an initial goal (which we will relax shortly), suppose we are interested in an algorithm which rounds this bi-point solution to an integer solution of cost at most $\alpha(aD_1+bD_2)$, for some $\alpha$. As already mentioned, such an algorithm can be used to get a $(2 \times \alpha)$-approximation to $k$-median.
Suppose a client $j$ is closest to $i_1$ in solution $\F_1$, and $i_2$ in pseudo-solution\footnote{One that may open more than $k$ facilities.} $\F_2$. \emph{Ideally}, one would like to round the bi-point solution in such a way that $i_1$ is open with probability $a$, and $i_2$ is open with the remaining probability $b$. Then the expected connection cost of $j$ would be exactly its contribution to the bi-point cost, and we would get a bi-point rounding factor of $1$. The problem is that we cannot directly correlate this pair of events for every single client, while still opening only $k$ facilities. 
Jain and Vazirani's approach is to pair each $i_1\in\F_1$ with its closest neighbor in $\F_2$, and ensure that one of the two is open \cite{jain_vazirani}. This approach loses at most a factor of $2$, which is equal to the integrality gap of the $k$-median LP, and so is the best one might expect. However, Li and Svensson beat this factor by allowing their algorithm to open $k+c$ facilities. They then give a (surprising) method to convert such an algorithm to one that satisfies the budget-$k$ constraint, adding $\eps$ to the approximation constant, and a factor of $n^{O(c/\eps)}$ to the runtime. This method actually runs the algorithm on a polynomial number of sub-instances of the original problem, and thus is not limited by the integrality gap of the original LP. Thus we obtain our relaxed goal:

\begin{definition}
\label{defn:relaxed-goal}
\textbf{(Relaxed goal})
Given a bi-point solution parametrized by $\F_1, \F_2, a, b, D_1,$ and $D_2$ as in Definition~\ref{defn:bi-point}, round it to an integer pseudo-solution using at most $k + f(\epsilon)$ facilities, and of cost at most $\alpha(aD_1+bD_2)$ for $\alpha \sim 1.3371$. Here, $\epsilon > 0$ is an arbitrary constant. 
\end{definition}

We will mainly discuss how to achieve such an improved $\alpha$ now, and defer discussion of the function $f$ to Section \ref{subsubsec:depround-kmedian}. 

Li and Svensson's approach, which we will generalize, is to create a cluster for every facility in $\F_1$, and put each facility in $\F_2$ into its nearest cluster. We refer to each cluster as a \emph{star}. Each star has a single \emph{center} in $\F_1$, and zero or more \emph{leaves} in $\F_2$. It is useful to consider algorithms with the property that each star has either its center or all of its leaves opened (we will do better by relaxing this property). Then our client $j$ may always connect to either $i_2$, or the center of the star containing $i_2$, which cannot be too far away. The work of \cite{li_svensson} presents two such algorithms. The first is based on a knapsack-type LP, and yields a rounding factor of $1.53$ by opening two extra facilities. The second opens $a$ and $b$ fractions of $\F_1$ and $\F_2$ at random and obtains $\frac{1+\sqrt3+\eps}2\approx 1.366+\eps$ by opening $O(1/(ab\eps))$ extra facilities. (The former and the trivial solution $\F_1$ are used to handle the cases where 
$a$ 
or $b$ is close to zero.)

We improve this by first considering the stars with only one leaf (or \emph{1-stars}) separately from the stars with two or more leaves (or \emph{2-stars}). This allows us more freedom in shifting probability mass. For 1-stars, we do not need to worry about opening $a$ of the centers and $b$ of the leaves. We can instead consider either opening the leaf, or opening the center, and then choose the better of the two. For 2-stars, we still have the option of opening the center and leaves in proportion $a$ and $b$ respectively. However we also may shift the mass toward the centers, and open them in proportion $1$ and $b/2$ instead. This gives 4 possible combinations; we may try them all and take the best solution.
 When combined with the aforementioned knapsack-based algorithm, this idea improves the approximation from $1.53$ to $1.4$, with the same $n^{O(1/\eps)}$ runtime. However, it does not immediately improve the $(1+\sqrt{3})/2$ factor. It does, however, impose some useful structure on the distribution of clients in a worst-case instance.
Also, we may consider shifting mass from 2-stars to 1-stars, the extreme being that we open the centers of all the 2-stars and open the remaining facilities in 1-stars. This must beat the (tight) bound of $D_1$, which would mean it does better than $(1+\sqrt{3})/2$. However, this effect is negligible if the number of 1-stars dominates the number of 2-stars. Thus, a worst-case instance must have a very large fraction of 1-stars.

But if there are so many 1-stars, perhaps we could close entirely (center and leaf) a tiny fraction of them and shift the mass to other stars. In fact, one may show a bi-point solution where 1.366 is the best we can do while still preserving the center-or-leaves property for all stars, but once we are allowed to \emph{close the center and leaf} of some of the 1-stars, we get a factor of 1. This motivates the following strategy: For each 1-star, classify it as ``long" or ``short" based on its relative ``size'' (distance from the center to the leaf) compared to its distance to the nearest leaf of a 2-star. If many such stars are long, we may close both the center and the leaf, and still get a reasonable cost bound for any client connected to that star. This gives us the ability to consider shifting mass from these long stars to other stars. On the other hand, if many such stars are short, this actually means we may obtain a better cost bound for clients connected to centers of short 1-stars and leaves of 2-stars, by 
connecting them to the leaf of the short star in the worst case. (The worst-case structure mentioned above enforces that there are such clients.) In either case we obtain a factor strictly smaller than $(1+\sqrt{3})/2$.

To find the approximation constant of this new algorithm, we construct a factor-revealing non-linear program, the solution to which describes the new worst-case instance. The program is not convex, so local-search methods have no guarantee of finding the global optimum. However, there are 4 variables that, when fixed, render the system linear. Inspired by Zwick's use of interval arithmetic \cite{zwick}, we consider a relaxation of the LP over a small interval of those 4 variables. The relaxation itself is linear and may be solved efficiently, but still gives a valid upper bound on the value of the original program in that interval. By splitting the search space into sufficiently small intervals, we are able to systematically prove an upper bound over the entire space, leading rigorously to the value of 1.3371.

Our approach shows there is potential to improve the approximation by improving the bi-point rounding algorithm. On the other hand, we give a family of explicit instances and bi-point solutions whose optimal rounding loses a factor of $\frac{1+\sqrt2}{2}\approx1.207$, even when we allow $k+o(k)$ facilities to be opened.

\subsection{Dependent rounding with almost-independence on small subsets}
\label{subsec:dep-round} 
Dependent rounding has emerged as a useful technique for rounding-based algorithms. We start by discussing an ``unweighted" special case of it, which captures much of its essence. Section~\ref{subsubsec:depround-kmedian} then discusses the general weighted case and its applications to $k$-median:  specifically, our improvement of the function $f(\epsilon)$ of Definition~\ref{defn:relaxed-goal} from the $\Theta(1/\epsilon)$  of \cite{li_svensson} to $\Theta(\log(1/\epsilon))$. This leads to our improved run-time. 

Let us discuss the basic ``unweighted" setting of dependent rounding. 
Starting with the key deterministic ``pipage rounding" algorithm of Ageev \& Sviridenko \cite{AS}, dependent-rounding schemes have been interpreted probabilistically, and have found several applications and generalizations in combinatorial optimization (see, e.g., \cite{DBLP:journals/algorithmica/BansalGLMNR12,DBLP:journals/siamcomp/CalinescuCPV11,DBLP:conf/soda/ChekuriVZ11,DBLP:journals/jacm/GandhiKPS06,harvey-olver:pipage,srin:level-sets} for a small sample). 
These naturally induce certain types of \emph{negative correlation} \cite{srin:level-sets} and sometimes even-more powerful negative-association properties 
\cite{dubhashi-jonasson-ranjan:neg-depend,Kramer:2011:NDS:1972785.1972788}, which are useful in proving Chernoff-like concentration bounds on various monotone functions of such random variables. We consider settings where some form of \emph{positive} correlation is desirable in dependent rounding, and construct efficiently-sampleable distributions that offer much more than what regular dependent-rounding and limited positive correlation ask for. 

We now define our basic problem. 
Consider any $P = (p_1, p_2, \ldots, p_n) \in [0,1]^n$ such that $\sum_i p_i$ is an {\em integer} $\ell$, and 
let $[s]$ denote the set $\{1, 2, \ldots, s\}$ as usual. Building on \cite{AS}, the work of 
\cite{srin:level-sets} shows how to efficiently sample a random vector $(X_1, X_2, \ldots, X_n) \in \{0,1\}^n$ such that:
\begin{description}
\item[{\bf (A1)}] the ``right" marginals: $\forall i$, $\Pr[X_i = 1] = p_i$; 
\item[{\bf (A2)}] the sum is preserved: $\Pr[\sum_i X_i = \ell] = 1$, and
\item[{\bf (A3)}] negative correlation: $\forall S \subseteq [n] ~\forall b \in \{0,1\}, 
~\Pr[\bigwedge_{i \in S} (X_i = b)] \leq  \prod_{i \in S} \Pr[X_i = b]$. 
\end{description}
As we see below, some algorithms for applications including $k$-median and budgeted MAX-SAT ask for the above properties, but also desire some form of \emph{positive} correlation among selected (often ``small") subsets of the variables. Generalizing all of these, our basic problem is:

\smallskip \noindent \textbf{The basic problem.} In the above setting of $P$ and $\ell$, suppose we also have some parameter $\alpha$ (which could potentially be $o(1)$) such that $\alpha \leq p_i \leq 1 - \alpha$ for all $i$. Can one (efficiently) sample $(X_1, X_2, \ldots, X_n) \in \{0,1\}^n$ such that (A1), (A2) and (A3) hold, \emph{along with} the property that ``suitably small" subsets of the $X_i$ are \emph{nearly independent}? Informally, for some ``not too small" $t$, we want, for all $i_1 < \cdots < i_k$ with $k \leq t$ and for all $b_1, \ldots, b_k \in \{0,1\}^k$ that 
$\Pr[\bigwedge_{j=1}^k (X_{i_j} = b_j)]$ is ``close" to what it would have been if the $X_r$'s were all independent. Concretely, letting
\begin{equation}
\label{eqn:desired-target}
\lambda = \lambda(i_1, \ldots, i_k, b_1, \ldots, b_k) \doteq [\prod_{u:~b_u = 0} (1 - p_{i_u})] \cdot [\prod_{v:~b_v = 1} p_{i_v}], 
\end{equation}
we want for some $t$ and suitably-small $\beta_1, \beta_2$ that 
\begin{equation}
\label{eqn:desired-goal}
\textbf{(A4)}~ \text{small subsets are nearly independent:}~~(1 - \beta_1) \cdot \lambda \leq \Pr[\bigwedge_{j=1}^k (X_{i_j} = b_j)] \leq (1 + \beta_2) \cdot \lambda,
\end{equation}
again requiring $k \leq t$ and otherwise letting the indices $i_j$ and bits $b_j$ be arbitrary. That is, we want \emph{dependent-rounding schemes with near-independence 
on ``small" subsets}. (Note that such ``almost $t$-wise independence" is also a very important tool in the different context of derandomization 
\cite{DBLP:journals/siamcomp/NaorN93,DBLP:journals/rsa/EvenGLNV98,DBLP:journals/jcss/ChariRS00,DBLP:journals/combinatorica/Lu02}.)
Given the broad applicability of dependent rounding, it is easy to imagine that such schemes will have a range of applications; two concrete applications arise in the Li-Svensson work \cite{li_svensson} and in the budgeted MAX-SAT problem -- see Sections \ref{subsubsec:depround-kmedian} and \ref{subsec:budgeted-sat}. Note that some upper-bound on $t$ is necessary, as a function of $n$ and $\alpha$: indeed, if $k = t > \alpha n$, then since $\ell = (1 - \alpha)n$ is possible, (A2) implies that $\Pr[\bigwedge_{j=1}^k (X_{i_j} = 0)]$ can be zero regardless of the rounding scheme, violating (\ref{eqn:desired-goal}). This motivates our restriction to ``small" $t$; in fact, even the cases $t = O(1)$ and $t = \text{polylog}(n)$ are of interest.

\smallskip \noindent \textbf{Our results for the basic problem.} We present a randomized linear time algorithm to sample $(X_1, X_2, \ldots, X_n) \in \{0,1\}^n$ such that (A1), (A2) and (A3) hold, and with (A4) true with the following parameters:
\begin{equation}
\label{eqn:beta12}
\beta_1 = t^2 / (n \alpha^2) ~\text{and}~ \beta_2 = \left(1 + \frac{t}{n \alpha^2}\right)^{t-1} - 1.
\end{equation}
In particular, for $t = o(\alpha \sqrt{n})$, we get $\beta_1, \beta_2 = o(1)$ -- i.e., we have near-independence to within $(1 \pm o(1))$ factors.
This result is presented in Theorem~\ref{thm:limited-dep-general}.
We add an extra element of randomness to the approach of \cite{srin:level-sets}, in order to obtain our results. (For applications where $\alpha$ or its substitutes $\hat{\alpha}$ and $\hat{q}$ are extremely close to zero, one can often round the $p_j$'s that are very close to $0$ or $1$ separately -- e.g., by techniques such as those of Appendix~\ref{sec:budgeted-MAX-SAT}.) 

\subsubsection{Weighted dependent rounding and rounding bi-point solutions for $k$-median}
\label{subsubsec:depround-kmedian}
The work of \cite{li_svensson} actually has to deal with a more general ``weighted" version of dependent rounding: given weights $a_i \geq 0$, we want to preserve $\sum_i a_i p_i$, rather than $\ell = \sum_i p_i$ as in (A2). However, the setting of \cite{li_svensson} has some additional useful properties, such as $\alpha = \Theta(1)$, $t = 2$, all the $p_i$'s are the same, and that $\sum_i a_iX_i$ can exceed $\sum_i a_i p_i$ by an additive constant factor. Using these properties, the work of \cite{li_svensson} solves this dependent-rounding problem in a certain manner; in particular, this leads to $f(\epsilon) = \Theta(1/\epsilon)$ in the context of Definition~\ref{defn:relaxed-goal}. The overall run-time of \cite{li_svensson} is $N^{O(f(\epsilon)/\epsilon)}$, and is hence $N^{O(1/\epsilon^2)}$. 

In Section~\ref{sec:dep-round}, we develop a general solution to the weighted dependent-rounding problem, which in particular solves the unweighted problem ((A1) -- (A4)); this is done without assuming that $t = O(1)$, or that all the $p_i$'s are the same, etc. Since (A4) -- for the weighted and unweighted cases -- goes far beyond negative correlation (property (A3)) alone, we believe that our method could have a range of consequences, given the number of applications of dependent rounding seen over the last $15$ years. In any case, when specialized to the $k$-median application, this enables us to set 
$f(\epsilon) = \Theta(\log(1/\epsilon))$ in the context of Definition~\ref{defn:relaxed-goal}, and hence our overall run-time for $k$-median becomes $N^{O((1/\epsilon) \log(1/\epsilon))}$.

\subsection{Budgeted MAX-SAT}
\label{subsec:budgeted-sat}
The budgeted version of set-cover is well-understood: given a set-cover instance in which we can choose at most $k$ sets, we aim to make this choice in order to maximize the (weighted) number of ground-set elements covered \cite{DBLP:journals/ipl/KhullerMN99}. The work of \cite{DBLP:journals/ipl/KhullerMN99} presents an $(1 - 1/e)$-approximation for this problem, and shows the hardness of obtaining an $(1 - 1/e + \epsilon)$-approximation. This problem can be naturally generalized as follows. 
Consider an arbitrary CNF-SAT formula $\phi$ over Boolean variables $x_1, x_2, \ldots, x_n$, and with weight $w_i \geq 0$ for clause $i$. We aim to assign a True/False value to each $x_j$, in order to maximize the total weight of the satisfied clauses. However, we also have a hard budget-constraint as follows. The cost model is that there exist $a, b \geq 0$ such that for each $x_j$, we pay $a$ if we set $x_j$ to True, and $b$ if we set $x_j$ to False. We have a hard budget of $B$ on our total cost. (Such a budget-constraint is perhaps more well-motivated when we view each variable as having two possible general choices, rather than True/False.) Thus, letting $X_j$ be the indicator variable for $x_j$ being true, our budget constraint for the case $a \geq b$ is that $\sum_j X_j \leq k'$ (where $k' = \lfloor (B - nb)/(a - b) \rfloor$); if $b > a$, then by complementing all variables, we get a constraint of the same type. Thus, we may assume 
w.l.o.g.\ that $\sum_j X_j \leq k$ is our budget constraint, for some given integer $k$. Note that budgeted set-cover is the special case of this problem when the formula $\phi$ has no variables negated. The constraint ``$\sum_j X_j \leq k$" naturally suggests dependent rounding; on the other hand, clauses in $\phi$ such as ``$x_i \vee \overline{x_j}$" would benefit from positive correlation. For this budgeted MAX-SAT problem, though, some of our rounding approaches for $k$-median show that the structure of the problem allows for a simple rounding solution, leading to an essentially-best-possible $(1 - 1/e - \epsilon)$-approximation for any constant $\epsilon > 0$. Please see Appendix~\ref{sec:budgeted-MAX-SAT}.

\subsection{Perspective}
\label{sec:perspective}
Our primary contributions are two-fold. 

Our first contribution is an in-depth analysis of the bi-point rounding algorithm of \cite{li_svensson}, and observe that their rounding can be improved by a multi-pronged approach. In essence, these different ``prongs" have different worst-case scenarios, and hence a suitable combination of them can do better with any adversarial strategy for choosing the input instance to the problem. This leads to an improved approximation ratio for the fundamental $k$-median problem.

Our second major contribution is to (weighted) dependent rounding. This general technique has seen numerous applications over the last $15$ years or so, primarily since it can handle hard constraints such as (A2), and can also guarantee negative correlation (A3), which is useful for concentration bounds and other applications. However, the existing body of work largely does not address positive correlation, let alone near-independence, even for relatively-small subsets of the underlying variables $X_i$. We are able to show that we can guarantee all the existing properties of dependent rounding, and guarantee near-independence for not-too-large subsets, in the sense of (A4). We believe that this will yield further applications. Our dependent-rounding approach also helps improve the run-time of the $k$-median application.

%\section{Improved primal-dual algorithm}
%\label{sec:jmsprime}
%\input{sec2_primal_dual.tex}

\section{Dependent rounding with near-independence on small subsets}
\label{sec:dep-round}
\newcommand{\simplify}{\textsc{Simplify}}
\newcommand{\E}{\text{\bf E}}
\newcommand{\X}{\mathbf{X}}
\newcommand{\qavg}{\hat{q}}
\newcommand{\alphavg}{\hat{\alpha}}

Recall the setup described in Section~\ref{subsec:dep-round}, and properties (A1)-(A4) in particular. For the $k$-median application, we will actually need a weighted generalization of this, as mentioned briefly in Section~\ref{subsubsec:depround-kmedian}. The basic change is that we now have positive weights $a_1, a_2, \ldots, a_n$, and want to preserve the \emph{weighted} sum $\sum_i a_i p_i$, instead of $\sum_i p_i$ as in (A2). Such preservation may not be possible (no matter what the rounding), so we will leave at most one variable un-rounded at the end, as specified by property (A0') next. Let us describe our main problem, and then discuss the results we obtain. Our \textbf{main problem},  
given $P=(p_1,\dots,p_n)\in[0,1]^n$ and positive weights $A=(a_1\dots,a_n)\in\mathbb{R}_{>0}^n$, is to efficiently sample a vector $(X_1,X_2,\dots,X_n)$ from a distribution on $[0,1]^n$ which satisfies the following properties:
\begin{description}
\item[(A0')] ``almost-integrality": all but at most one of the $X_i$ lies in $\{0,1\}$ (with the remaining at most one element lying in $[0,1]$); 
\item[(A1')] $\forall i, \E[X_i]=p_i$;
\item[(A2')] $\Pr[\sum_i a_iX_i = \sum_i a_ip_i]=1$;
\item[(A3')] $\forall S \subseteq [n]$, $\E[\prod_{i \in S} (1 - X_i)] \leq \prod_{i \in S} (1 - p_i)$, and 
$\E[\prod_{i \in S} X_i] \leq \prod_{i \in S} p_i$;
\item[(A4': informal)] if the weights $a_i$ are ``not too far apart", there is near-independence for subsets of $\{X_1, X_2, \ldots, X_n\}$ that are of cardinality at most $t$, 
analogously to (A4). 
\end{description}

In this section we will give an $O(n)$-time algorithm (called {\sc DepRound}) for sampling from such a distribution; see Section~\ref{sec:depround}.
In particular, this shows that such distributions exist. This algorithm is a generalization of the unweighted version given in \cite{srin:level-sets}, with a specific (random) ordering of operations, leading to the added property (A4') of near-independence. 
Our main theorem is Theorem~\ref{thm:limited-dep-general}. 
It basically says, in the notation of (\ref{eqn:desired-goal}), that we can achieve $\beta_1, \beta_2 = O(t^2/(n \alpha^3))$ when 
$t \leq O(\sqrt{n \alpha^3})$. Thus, we obtain near-independence up to fairly large sizes $t$. This bound is further improved in 
Section~\ref{sec:dep-round-special}. Let us also remark about the possible (sole) index $i$ that is left unrounded, as in (A0'). Three simple ways to round this are to round down, round up, or round randomly; these can be chosen in a problem-specific manner. None of these three fits the $k$-median application perfectly; a different probabilistic handling of this index $i$ is done in Section~\ref{sec:bipoint-main-case}. 

Note that {\sc DepRound} could be described more simply as a form of pipage rounding (with an added, crucial, element of processing the variables in random order), with flow adjusted proportionally to the weights. However, in order to facilitate the analysis of (A4'), we give the following, less compact description of the algorithm.
Also note that we will apply this method to $k$-median in Section~\ref{sec:bipoint-main-case}, but here it is described as a general-purpose rounding procedure, independent of $k$-median; this is because we believe that this method is of independent interest since it goes much beyond negative correlation. Thus, the reader is asked to note that the notation defined in this section is largely separate from that of the other sections.

\subsection{The {\sc Simplify} subroutine}
\label{sec:simplify}
As in \cite{srin:level-sets}, our main subroutine is a procedure called $\simplify(a_1,a_2,\beta_1,\beta_2)$. It takes as input two fractional values $\beta_1,\beta_2\in(0,1)$, and corresponding positive weights $a_1,a_2\in\mathbb{R}_{>0}$. The subroutine outputs a random pair of values $(\gamma_1,\gamma_2)\in[0,1]^2$, with the following properties:

\begin{description}
\item[(B0)] $\gamma_1,\gamma_2\in[0,1]$, and at least one of the two variables is integral (0 or 1);
\item[(B1)] $\E[\gamma_1]=\beta_1$ and $\E[\gamma_2]=\beta_2$;
\item[(B2)] $\Pr\left[a_1\gamma_1+a_2\gamma_2=a_1\beta_1+a_2\beta_2 \right]=1$; and
\item[(B3)]
$\E[\gamma_1\gamma_2]\le\beta_1\beta_2$, and $\E[(1-\gamma_1)(1-\gamma_2)\le(1-\beta_1)(1-\beta_2)$.
\end{description}

Now define $\text{\sc simplify}(a_1,a_2,\beta_1,\beta_2)$ as follows. 
There are four cases:
\begin{enumerate}[\bf\text{Case} I:]
\item $0\le a_1\beta_1+a_2\beta_2 \le\min\{a_1,a_2\}$. With probability $a_2\beta_2/(a_1\beta_1+a_2\beta_2 )$ set $\gamma_1= 0$. With remaining probability set $\gamma_2= 0$.
\item $a_1<a_1\beta_1+a_2\beta_2 <a_2$. With probability $\beta_1$ set $\gamma_1=1$. With remaining probability set $\gamma_1=0$.
\item $a_2<a_1\beta_1+a_2\beta_2 <a_1$. With probability $\beta_2$ set $\gamma_2=1$. With remaining probability set $\gamma_2=0$.
\item $\max\{a_1,a_2\}\le a_1\beta_1+a_2\beta_2 \le a_1+a_2$. With probability $a_2(1-\beta_2)/(a_1(1-\beta_1)+a_2(1-\beta_2))$ set $\gamma_1=1$. With remaining probability set $\gamma_2=1$.
\end{enumerate}

If we set $\gamma_1=0$, then set $\gamma_2=\beta_2+\beta_1\frac{a_1}{a_2}$.

If we set $\gamma_1=1$, then set $\gamma_2=\beta_2-(1-\beta_1)\frac{a_1}{a_2}$.

If we set $\gamma_2=0$, then set $\gamma_1=\beta_1+\beta_2\frac{a_2}{a_1}$.

If we set $\gamma_2=1$, then set $\gamma_1=\beta_1-(1-\beta_2)\frac{a_2}{a_1}$.

\begin{lemma}\label{lemma:b-props}
$\simplify(a_1,a_2,\beta_1,\beta_2)$ outputs $(\gamma_1,\gamma_2)$ with properties (B0), (B1), (B2), and (B3).
\end{lemma}
This is straightforward to show; we provide a partial proof in Appendix \ref{apdx:depround}.

\subsection{Main algorithm: {\sc DepRound}} \label{sec:depround}
We now describe the full dependent rounding algorithm, which we denote {\sc DepRound}.

\begin{algorithm}[h]
\caption{$\textsc{DepRound}(P,A)$}
\begin{algorithmic}[1]
\STATE $X\gets P$
\STATE Let $\pi\in S_n$ be a random permutation.
\WHILE{$X$ contains at least two fractional elements} 
\STATE Let $X_i$ and $X_j$ be the two left-most fractional elements in $\pi(X)$.\label{step:pick_ij}
\STATE $(X_i, X_j)\gets \simplify(a_i,a_j,X_i,X_j)$
\ENDWHILE
\STATE \textbf{Return} $X^{(t)}$.
\end{algorithmic} 
\label{algo:A}
\end{algorithm}
Define $X^{(s)}=(X^{(s)}_1,X^{(s)}_2,\ldots,X^{(s)}_n)$ to be the value of $X$ after the step $s$ of the rounding process, with $X^{(0)}=P$, and $X^{(T)}=X$ if the algorithm halts after $T$ steps. 
We say $X_i$ is \emph{fixed} during step $s$ if $X_i^{(s-1)}$ is fractional but $X_i^{(s)}$ is integral, since this implies $X_i$ will not change in any future steps.

It is not hard to follow the proof of \cite{srin:level-sets} and validate properties (A0'), (A1'), (A2'), and (A3'); we present a proof sketch below. 
\begin{lemma}
{\sc DepRound} samples a vector in $O(n)$ time; this vector satisfies properties (A0'), (A1'), (A2'), and (A3').
\end{lemma}
\begin{proof}
%The proof is basically as in \cite{srin:level-sets}, but we sketch it here.
By definition, {\sc DepRound} ends only when there is at most one fractional variable remaining, and by definition of $\simplify$, the remaining variables are either 0 or 1, showing (A0'). Since at least one variable is fixed in each constant-time step, the algorithm takes at most $n-1$ steps, and thus runs in linear time. (The random permutation $\pi$ may also be generated in linear time.)

(B1) implies that for each $i$ and $s$, $\E[X_i^{(s)}]=X_i^{(s-1)}$. By induction, this implies (A1'). (B2) implies that $\sum_i a_i X^{(T)}_i =\sum_i a_i X^{(T-1)}_i =\ldots=\sum_i a_i X^{(0)}=\sum_i a_i p_i$, thus implying (A2'). Finally, let $X_i$ and $X_j$ be the variables chosen in line \ref{step:pick_ij} during step $t$. If $i,j\in S$, then (B3) implies $\E[\prod_{i\in S}X^{(s)}_i]\le \prod_{i\in S}X^{(s-1)}_i$ and $\E[\prod_{i\in S}(1-X^{(s)}_i)]\le\prod_{i\in S}(1-X^{(s-1)}_i)$ (all other terms in the product are constant). If only one or neither of $i,j$ are in $S$, then the same hold with equality. By induction, these imply (A3').
\end{proof}
Note that the above properties hold under any ordering $\pi$. The additional element of using $\pi$ to process the indices in random order is needed only for the new property (A4').

\subsection{Limited dependence} 
In this section we will prove the limited dependence property (A4').  Consider a subset $I\subseteq[n]$ of $t$ indices, with corresponding ``target'' bits $(b_i)_{i\in I}\in\{0,1\}^t$.  For each $i\in I$, define $Y_i:=X_i$ if $b_i=1$, or $Y_i:=1-X_i$ if $b_i=0$. We are interested in the value of $\E\left[\prod_{i\in I}Y_i\right]$, which is essentially equal to the joint probability $\bigwedge_{i\in I}(X_i=b_i)$. Note these are not exactly equivalent, because of the one fractional variable. This variable must be handled in a domain-specific way to ensure this property holds, as we will do when applying it to $k$-median.

From property (A1'), we have $\E[X_i]=p_i$. Similarly, $\forall i\in I$, define $q_i$ such that $\E[Y_i]=q_i$. That is, let $q_i$ be either $p_i$ or $1-p_i$ if $b_i$ is 1 or 0, respectively. If we independently rounded each variable, $\E[\prod_{i\in I} Y_i]$ would be exactly $\prod_{i\in I}q_i$ (but (A2') would be violated). We will show that when set $I$ is not too large, the product is still very close to $\prod_{i\in I}q_i$ in expectation (i.e., that the variables $\{X_i\}_{i\in I}$ are near-independent). 

During a run of {\sc DepRound}, we say that two variables $X_i$ and $X_j$ are \emph{co-rounded} if they are both changed in the same step. We will first show that if two variables are far apart in $\pi(X)$, then they are unlikely to be co-rounded. Then we will show that a group of variables is near-independent if the probability of any two of them being co-rounded is small. Finally, we will show that for a small enough set $I$ variables are very likely to be far apart, and thus very likely to remain independent.

\subsubsection{Distant variables are seldom co-rounded}
Recall that {\sc DepRound} always calls $\simplify$ on the two left-most fractional variables in $\pi(X)$, fixing at least one of them (to 0 or 1). Thus, once a variable \emph{survives} (i.e., remains fractional after) one step it will continue to be included in all subsequent $\simplify$ steps until it gets fixed (or becomes the last remaining fractional variable). We want to upper bound the probability that a variable survives for too long. We first show that in any two consecutive steps involving $X_i$, there is a minimum probability that $X_i$ gets fixed (if the weights are not too different).
\begin{lemma}\label{lemma:pair-prob}
Let $a_{min}:=\min_i\{a_i\}$ and $a_{max}:=\max_i\{a_i\}$ be the minimum and maximum weights. Suppose $\frac{a_{max}}{a_{min}}\le 2$. 
Suppose $X_i$ is co-rounded with variable $X_j$ in step $s$, and if it survives it will be co-rounded with variable $X_k$ in step $s+1$. Let $\beta_j:=X_j^{(s-1)}$ and $\beta_k:=X_k^{(s)}$. 
Then $X_i$ will be fixed in one of these two steps with probability at least $p=\min\{\beta_j,1-\beta_j\}\cdot\min\{\beta_k,1-\beta_k\}$.
\end{lemma}
\begin{proof}
What is the probability that $X_i$ is fixed in the first step? It depends on which of the four cases occur. In case I, it is $\frac{a_j\beta_j}{a_i\beta_i+a_j\beta_j}\ge \frac{a_j\beta_j}{a_j}=\beta_j$. In case II it is 1. In case IV it is $\frac{a_j(1-\beta_j)}{a_i(1-\beta_i)+a_j(1-\beta_j)}
= \frac{a_j(1-\beta_j)}{a_i+a_j - (a_i\beta_i+a_j\beta_j)}
\ge  \frac{a_j(1-\beta_j)}{a_i+a_j-a_i}=1-\beta_j$. In these three cases, $X_i$ is fixed with probability at least $\min\{\beta_j,1-\beta_j\}\ge p$. However, in the remaining case III, $X_i$ will be fixed with probability 0. 

Given that case III occurs in the first step, what is the probability that $X_i$ gets fixed in the second step? It is at least the probability that one of cases I, II or IV occur in the second step, times $\min\{\beta_k,1-\beta_k\}$ (by the same reasoning as above). When case III occurs in the first step, $X_i^{(s)}$ gets set randomly to one of two values: $\beta_i+\beta_j\frac{a_j}{a_i}$ or $\beta_i-(1-\beta_j)\frac{a_j}{a_i}$, with probability $1-\beta_j$ or $\beta_j$, respectively.  These values differ by exactly $\frac{a_j}{a_i} \ge\frac{a_{min}}{a_{max}}\ge\frac12$. 
Now, in the second step, case III only occurs if $X_i^{(s)}$ lies in the open interval $(\frac{a_k}{a_i}-\frac{a_k}{a_i}\beta_k, 1-\frac{a_k}{a_i}\beta_k)$. But the distance between any two numbers in this interval is strictly less than $1-\frac{a_k}{a_i}\le1-\frac{a_{min}} {a_{max}}\le\frac12$. Therefore, the two possible values of $X_i^{(s)}$ cannot both lie in the interval required for case III, so with probability at least $\min\{\beta_j,1-\beta_j\}$, the second step will be a case other than III. 
\end{proof} 

\smallskip \noindent \textbf{Important remark on notation:} In the following paragraph, by overloading notation, we fix the random permutation $\pi$ to be some arbitrary but fixed $\pi$. Several pieces of notation such as $\sigma, J_i$, and, most importantly, 
$\delta_k:=\prod_{i=0}^{\lfloor |J_k|/2\rfloor}(1-\alpha_{j_{k,2i}}\alpha_{j_{k,2i+1}})$, are functions of this $\pi$. All statements and proofs \emph{until the end of the proof of Lemma~\ref{lemma:sandwich1}}, are \emph{conditional on the random permutation equaling $\pi$} (this is sometimes stated explicitly, sometimes not). 

\smallskip
Lemma~\ref{lemma:pair-prob} implies that the probability of a variable's survival decays exponentially with the number of steps survived. Now, given a permutation $\pi\in S_n$, let $\sigma:[t]\to I$ be the bijection such that $\pi^{-1}(\sigma(1))<\pi^{-1}(\sigma(2))<\cdots<\pi^{-1}(\sigma(t))$. 
Let 
\[ J:=[n]\setminus I \] 
be the set of the indices not in $I$. Now partition $J$ into sequences $(J_0,J_1,\ldots,J_t)$, using elements of $I$ as dividers. Formally, for $k=0,\ldots,t$, let $J_k$ be the maximal sequence $(j_{k,1},j_{k,2},\ldots)$ satisfying $\pi^{-1}(\sigma(k))<\pi^{-1}(j_{k,1})<\pi^{-1}(j_{k,2})<\cdots<\pi^{-1}(\sigma(k+1))$, letting $\pi^{-1}(\sigma(0)):=0$ and $\pi^{-1}(\sigma(i+1)):=n+1$. Note that if $\sigma(k)$ and $\sigma(k+1)$ are directly adjacent in $\pi$, then $J_k$ will be empty, as seen in the example below. 
\[ I=\{2,3,8\}
\qquad\pi([n])=(\underbrace{12,11,5}_{J_0},\mathbf{3},\underbrace{4,1,9,6}_{J_1},\underset{J_2=\emptyset}{\mathbf{8},\mathbf{2}},\underbrace{7,10}_{J_3})
 \qquad  (\sigma(1),\sigma(2),\sigma(3))=(3,8,2)
\]

For $k\in[t-1]$, let $Z_k$ be the ``bad'' event that {\sc DepRound} co-rounds $X_{\sigma(k)}$ and $X_{\sigma(k+1)}$. For $Z_k$ to occur, it is necessary that $X_{\sigma(k)}$ be co-rounded with all variables inbetween as well. For example, in the above sequence, $Z_1$ means that $X_{\sigma(1)}=X_3$ must be co-rounded with $X_4, X_1, X_9, X_6$, (surviving each round), and finally $X_8$. The next lemma bounds $\Pr[Z_k]$ in terms of the set $J_k$.

\begin{lemma}\label{lemma:zbound}
Let $\alpha_i:=\min\{p_i,1-p_i\}$. If $\frac{a_{max}}{a_{min}}\le 2$, then $\forall k\in[t-1]$, we have $\Pr[Z_k]\le \delta_k$, where $\delta_k:=\prod_{i=1}^{\lfloor |J_k|/2\rfloor}(1-\alpha_{j_{k,2i-1}}\alpha_{j_{k,2i}})$.
\end{lemma}

\begin{proof}
Assume $|J_k|\ge2$ (else the lemma is trivially true).
Let $E_{0}$ be the event that $X_{\sigma(k)}$ is co-rounded with $X_{j_{k,1}}$ (i.e. is not fixed earlier in the algorithm).
For $\ell\in[|J_k|]$, let $E_\ell$ be the event that $X_{\sigma(k)}$ is corounded with $X_{j_{k,\ell}}$ \emph{and} survives. To apply Lemma 2.3, we first express the probability of $Z_k$ in terms of consecutive pairs of events.
\begin{align}
\Pr[Z_k]=\Pr[E_0\land E_1\land\ldots\land E_{|J_k|}]
&\le \Pr\left[E_0\land \left(\bigwedge_{i=1}^{\lfloor |J_k|/2\rfloor}(E_{2i-1}\land E_{2i})\right)\right] 
\nonumber\\&= Pr[E_0]\prod_{i=1}^{\lfloor |J_k|/2\rfloor}\Pr\left[ E_{2i-1}\land E_{2i}\,\Big\vert \bigwedge_{j=0}^{2i-2} E_j\right]
\nonumber\\&\le \prod_{i=1}^{\lfloor |J_k|/2\rfloor}\Pr\left[ E_{2i-1}\land E_{2i}\,\Big\vert \bigwedge_{j=0}^{2i-2} E_j\right].
\label{eq:pair_prod}
\end{align}

Note that $E_{\ell}$ implies $E_{\ell-1}\land E_{\ell-2}\land\ldots\land E_0$. So $\Pr\left[E_{2i-1}\land E_{2i}\mid \bigwedge_{i=1}^{2i-2} E_i\right]=\Pr\left[E_{2i-1}\land E_{2i}\mid E_{2i-2}\right]$.
Observe that event $E_{2i-2}$ is equivalent to the event that $X_{\sigma(k)}$ is co-rounded with $X_{j_{k,2i-1}}$. Also, observe that if $E_\ell$ occurs during step $s$ (for $\ell\in[|J_k|]$), then $s$ is the first step involving $X_{j_{k,\ell}}$, so the input value to $\simplify$ is $X_{j_{k,\ell}}^{(s-1)}=p_{j_{k,\ell}}$. These observations together with Lemma 2.3 imply that $\Pr\left[E_{2i-1}\land E_{2i}\mid E_{2i-2}\right]\le 1-\alpha_{j_{k,2i-1}}\alpha_{j_{k,2i}}$. Then (\ref{eq:pair_prod}) is bounded by $\delta_k$ as defined in the lemma.
\end{proof}

\subsubsection{Seldom co-rounded variables are near-independent}

The following lemmas show that if the probability of variables in $\{X_i\}_{i\in I}$ being co-rounded is low, then $\E[\prod_{i\in I} Y_i]\approx \prod_{i\in I} q_i$. For notational convenience, define $\delta_t:=0$.

\begin{lemma}\label{lemma:sandwich1} Let $I_k:=\{\sigma(k),\sigma(k+1),\ldots,\sigma(t)\}$.
Let $\mathcal{E}_k$ denote a set of events which consists of exactly one of $Z_i$ or $(\bar Z_i\land Y_{\sigma(i)}=y_i)$ for each $i=1\ldots k$, where $y_i$ is some attainable value of $Y_{\sigma(i)}$. Then, conditioned on a fixed permutation $\pi$, the following holds for all $k\in[t]$: 
\begin{equation} 
\prod_{i\in I_k} \max\left\{q_i - \delta_{\sigma^{-1}(i)},0\right\}
\le \E\bigg[\prod_{i\in I_k} Y_i \,\Big\rvert\, \mathcal{E}_{k-1}
\bigg]
\le \prod_{i\in I_k} \left(q_i + \delta_{\sigma^{-1}(i)}\right).
\label{eq:sandwich-ih}\end{equation}
\end{lemma}
\begin{proof}
Recall $\pi\in S_n$ is a fixed permutation; all probabilities and expectations in this proof are conditioned on $\pi$. This proof formalizes the idea that if $Z_k$ doesn't occur, then $X_{\sigma(k)}$ and $X_{\sigma(k+1)}$ remain independent. If it does occur, the effect on the expected value is limited by $\delta_k$.

After step $s$ of {\sc DepRound}, we may consider the remainder of the algorithm as simply a recursive call on vector $X^{(s)}$ (using the same permutation $\pi$). Let $D_k$ be the first such call where $X_{\sigma(k)}$ is one of the two left-most fractional variables to be co-rounded. Then there is at most one fractional variable to the left of $X_{\sigma(k)}$, say $X_{i_0}$, and all variables to the right still have their initial values from $P$.

The key observation is that while events in $\mathcal{E}_{k-1}$ do affect the identity and initial value of $X_{i_0}$, they do not further influence the outcome of $D_k$. The only way $D_k$ could be further influenced is if $\mathcal{E}_{k-1}$ contains the event $Y_{\sigma(j)}=y_j$, where $X_{i_0}$ is with some probability the variable $X_{\sigma(j)}$. However, for each $j=1\ldots k-1$, $\mathcal{E}_{k-1}$ either contains $\bar Z_j$ (which means $X_{\sigma(j)}$ was fixed earlier and cannot be $X_{i_0}$), or it lacks $Y_{\sigma(j)}=y_j$.

This means that all the properties of {\sc DepRound} shown so far (which hold for a fixed $\pi$) still hold for $D_k$ when conditioned on $\mathcal{E}_{k-1}$. Namely, we have $E[X_i\mid \mathcal{E}_{k-1}]=p_i$ (for all $X_i$ to the right of and including $X_{\sigma(k)}$), $\Pr[Z_k|\mathcal{E}_{k-1}]\le \delta_k$, and -- as we will claim by induction --  (\ref{eq:sandwich-ih}) for $I_k$.
The bounds derived below handle the problematic compound event $Z_k\land Y_{\sigma(k)}=y_k$ explicitly by assuming that when $Z_k$ occurs, $Y_{\sigma(k)}$ always attains its worst-case value (1 for the upper bound, or 0 for the lower bound).

As a base case, consider the singleton set $I_t=\{\sigma(t)\}$. As just described, we have $\E[X_{\sigma(t)}\mid\mathcal{E}_{k-1}]=p_{\sigma(t)}$ , so $\E[Y_{\sigma(t)}\mid\mathcal{E}_{k-1}]=q_{\sigma(t)}$, and (\ref{eq:sandwich-ih}) follows from $\delta_t=0$.
We now proceed by induction on $k$, counting backward from $t$. Let $W_k:=\prod_{i\in{I_k}} Y_i=\prod_{j=k}^t Y_{\sigma(j)}$. Let $\bar Z_k$ be the complement of event $Z_k$.  For some $k<t$, assume that (\ref{eq:sandwich-ih}) holds for $D_{k+1}$ with set $I_{k+1}$. Then, using the independence properties just mentioned, we have that $\E[W_k\mid\mathcal{E}_{k-1}]$ is
\begin{align*}
=&\E[Y_{\sigma(k)} W_{k+1}\mid\mathcal{E}_{k-1}] 
\\=& \E[Y_{\sigma(k)} W_{k+1}\mid\bar Z_k\land\mathcal{E}_{k-1}]\Pr[\bar Z_k\mid\mathcal{E}_{k-1}]
	+\E[Y_{\sigma(k)} W_{k+1}\mid Z_k\land\mathcal{E}_{k-1}]\Pr[Z_k\mid\mathcal{E}_{k-1}] 
\\\le& \E[Y_{\sigma(k)} W_{k+1}\mid\bar Z_k\land\mathcal{E}_{k-1}]\Pr[\bar Z_k\mid\mathcal{E}_{k-1}]
	+\E[W_{k+1}\mid Z_k\land\mathcal{E}_{k-1}]\delta_k 
\\ =&\textstyle\sum_{y_k}\left(y_k\E[W_{k+1}\mid Y_{\sigma(k)}=y_k\land \bar Z_k\land 
\mathcal{E}_{k-1}]
	\Pr[Y_{\sigma(k)}=y_k\mid\bar Z_k\land \mathcal{E}_{k-1}]
	\Pr[\bar Z_k\mid\mathcal{E}_{k-1}]\right)
	+ \E[W_{k+1}\mid\mathcal{E}_k']\delta_k 
\\ =&\textstyle\sum_{y_k}\left(y_k\E[W_{k+1}\mid \mathcal{E}_k'']
	\Pr[Y_{\sigma(k)}=y_k\mid\bar Z_k\land\mathcal{E}_{k-1}]
	\Pr[\bar Z_k\mid\mathcal{E}_{k-1}]\right)
	+ \E[W_{k+1}\mid\mathcal{E}_k']\delta_k 
\\ \le& \textstyle\prod_{i\in I_{k+1}}(q_i+\delta_{\sigma^{-1}(i)})
	\big(\sum_{y_k}(y_k\Pr[Y_{\sigma(k)}=y_k \land\bar Z_k\mid \mathcal{E}_{k-1}]
	) +\delta_k \big) 
\\ \le& \textstyle\prod_{i\in I_{k+1}}(q_i+\delta_{\sigma^{-1}(i)})
	\big(\sum_{y_k}(y_k\Pr[Y_{\sigma(k)}=y_k \mid \mathcal{E}_{k-1}]
	) +\delta_k \big) 
\\ =& \textstyle\prod_{i\in I_{k+1}}(q_i+\delta_{\sigma^{-1}(i)})
	\big(\E[Y_{\sigma(k)}\mid \mathcal{E}_{k-1}] +\delta_k\big) 
\\ =& \textstyle\prod_{i\in I_{k+1}}(q_i+\delta_{\sigma^{-1}(i)})
	\big(q_{\sigma(k)} + \delta_{\sigma^{-1}(\sigma(k))}\big) 
= \textstyle\prod_{i\in I_{k}}(q_i+\delta_{\sigma^{-1}(i)}).
\end{align*}
Similarly, we have that $\E[W_k\mid\mathcal{E}_{k-1}]$ is
\begin{align*}
=&\E[Y_{\sigma(k)} W_{k+1}\mid\mathcal{E}_{k-1}] 
\\ \ge& \E\left[Y_{\sigma(k)} W_{k+1}\mid \bar Z_k\land\mathcal{E}_{k-1}\right] 
	\Pr[\bar Z_k\mid\mathcal{E}_{k-1}] 
\\ =& \textstyle\sum_{y_k} y_k \E\left[W_{k+1}\mid Y_{\sigma(k)}=y_k\land\bar Z_k\land\mathcal{E}_{k-1}\right]
	\Pr[Y_{\sigma(k)}=y_k\mid \bar Z_k\land \mathcal{E}_{k-1}] 
	\Pr[\bar Z_k\mid\mathcal{E}_{k-1}] 
\\ =& \textstyle\sum_{y_k} y_k \E\left[W_{k+1}\mid \mathcal{E}_{k}'\right] 
	\Pr[Y_{\sigma(k)}=y_k\mid \bar Z_k\land\mathcal{E}_{k-1}]
	\Pr[\bar Z_k\mid\mathcal{E}_{k-1}] 
\\ \ge& \textstyle\prod_{i\in I_{k+1}}(q_i-\delta_{\sigma^{-1}(i)})
	\E[Y_{\sigma(k)} \mid \bar Z_k\land\mathcal{E}_{k-1}]\Pr[\bar Z_k\mid\mathcal{E}_{k-1}] 
\\ =& \textstyle\prod_{i\in I_{k+1}}(q_i-\delta_{\sigma^{-1}(i)})
	\big(\E[Y_{\sigma(k)}\mid\mathcal{E}_{k-1}]-\E[Y_{\sigma(k)}\mid  Z_k\land\mathcal{E}_{k-1}]\Pr[ Z_k\mid\mathcal{E}_{k-1}] \,\big)
\\ \ge& \textstyle\prod_{i\in I_{k+1}}(q_i-\delta_{\sigma^{-1}(i)})
	\big(q_{\sigma(k)}-\delta_k\,\big) 
=\textstyle\prod_{i\in I_{k}}(q_i-\delta_{\sigma^{-1}(i)}).
\end{align*}
But also $\E[W_{k}|\mathcal{E}_{k-1}]\ge 0\cdot\prod_{i\in I_{k+1}}(q_i-\delta_{\sigma^{-1}(i)})$ so we use the better of the two lower bounds.
\end{proof} 

\smallskip \noindent \textbf{Remark.} From now on, we will no longer take $\pi$ as fixed, and hence the $\delta_k$ (which are functions of the random permutation) will be viewed as random variables.

\smallskip
\begin{lemma} \label{lemma:sandwich}
If $\frac{a_{max}}{a_{min}}\le2$, we have 
\[ \prod_{i\in I} q_i \cdot\E\left[\prod_{k=1}^{t-1} \max\left\{1 - \frac{\delta_k}{q_{\sigma(k)}},0\right\}\right]
\le \E\left[\prod_{i\in I} Y_i \right]
\le \prod_{i\in I} q_i \cdot\E\left[\prod_{k=1}^{t-1} \left(1 + \frac{\delta_k}{q_{\sigma(k)}}\right)\right]
.\] 
\end{lemma}
\begin{proof}
Apply Lemma \ref{lemma:sandwich1} with $k=1$ and $\mathcal{E}=\emptyset$. Recall $\delta_t:=0$. Take the expectation over all permutations $\pi$, and then factor out the constant $\prod_{i\in I} q_i$.
\end{proof}
\subsubsection{Small subsets are spread out}

The following lemma gives a very useful combinatorial characterization of the distribution of $\{J_k\}$.
\begin{lemma}\label{lemma:bijection}
Let $g=(g_1,\dots,g_{t+1})$ be a sequence of nonnegative integers which sum to $n-t$, picked uniformly at random from all such possible sequences. Then the distribution of $g$ is equal to the distribution of $(|J_0|,\dots,|J_t|)$. Both distributions are symmetric.
\end{lemma}
\begin{proof}Consider the mapping $\Phi_I:S_n\to\{0,1\}^n$ from permutations on $[n]$ to binary strings of length $n$, in which we replace each index in $I$ with a 1, and the rest with a 0. Also define $\Theta$ to be the following standard combinatorial bijection from binary strings to arrangements of balls in boxes: given a binary string $s$, first add a 1 to the beginning and end of the string; then, viewing the space between each nearest pair of 1's as a `box', and the zeros between each pair as `balls' in that box, let $\Theta(s)$ be the sequence which counts the number of balls in each box, from left to right. For an arbitrary permutation $\pi$, and the corresponding sets $\{J_k\}$, we see that $|J_k|$ corresponds to the number of zeros between the $k$'th and the $(k+1)$'th 1 in $\Phi_I(\pi)$, and thus to the number of balls in the corresponding box in $(\Theta\circ\Phi_I)(\pi)$. Therefore we have that $(\Theta\circ\Phi_I)(\pi)=(|J_0|,\ldots,|J_t|)$.
\[I=\{2,3,8\}\quad\pi([n])=(\underbrace{12,11,5}_{J_0},\mathbf{3},\underbrace{4,1,9,6}_{J_1},\underset{J_2=\emptyset}{\mathbf{8},\mathbf{2}},\underbrace{7,10}_{J_3})
\;\implies
\begin{array}{rl}
	\Phi_I(\pi)=&000100001100 \\
	(\Theta\circ\Phi_I)(\pi)=&(3,4,0,2)
\end{array}\] 

Now recall that {\sc DepRound} chooses a uniformly random permutation $\pi\in S_n$. Notice that $\Phi_I(\pi)$ only maps to binary strings of length $n$ with exactly $|I|=t$ ones. Furthermore, for each such binary string, there are exactly $t!(n-t)!$ permutations which map to it. Thus, $\Phi_I(\pi)$ is uniformly distributed over all $\binom{n}{t}$ such binary strings. $\Theta$ provides an exact bijection between binary strings of length $n$ with $t$ ones, and sequences $g$ as defined in the lemma. This implies that $(\Theta\circ\Phi_I)(\pi)$ -- and thus $(|J_0|,\ldots,|J_t|)$ -- is uniformly distributed over all such possible sequences $g$. 

Furthermore, by definition of $g$, all permutations of a sequence $g$ would be equally likely. Therefore the distribution of $g$, and thus $(|J_0|,\ldots,|J_t|)$, is symmetric.
\end{proof}

\smallskip \noindent \textbf{Remark.} Note that the above distribution over balls-in-boxes is such that each possible arrangement is equally likely. This is not to be confused with distributions obtained by randomly and independently throwing the balls into the boxes. 

\begin{lemma}\label{lemma:exjc}
 Consider a subset $C\subseteq[t]$ of size $c$, and let $J_C:=\bigcup_{k\in C}J_k$. Then for any constant $0<x<1$,
\[\E\left[x^{|J_C|}\right]\le\left(\frac{t}{n(1-x)}\right)^c .\]
\end{lemma}
\begin{proof}
From Lemma \ref{lemma:bijection} the distribution of $(|J_0|,|J_1|,\ldots |J_t|)$ is symmetric. Thus, when considering the distribution of a function of the sizes $\{|J_k|\}_{k\in C}$, we may w.l.o.g.\ assume that $J_C=J_0\cup J_1\cup\ldots\cup J_{c-1}$. Notice since $\{J_k\}$ are all disjoint, we have $|J_C|=\sum_{k\in C}|J_k|$; also, $|J_C|\le|J|=n-t$.
\begin{equation}\label{eq:ejc}
\E\left[x^{|J_C|}\right]
=\sum_{m=0}^{n-t} \Pr[|J_C|=m] x^m
=\sum_{m=0}^{n-t} \Pr\left[\sum_{k=0}^{c-1}|J_k|=m\right] x^m.
\end{equation}
Now for a quick exercise in counting. From the previous proof, $\Phi_I(\pi)$ maps permutations uniformly to $n$-digit binary strings with $t$ 1's. For a given permutation $\pi$, we observe that $\sum_{k=0}^{c-1} |J_k|=m$ iff the binary string $\Phi_I(\pi)$ has exactly $m$ zeros before the $c$'th 1 (i.e., there are $m$ total balls in the first $c$ boxes). How many of the $\binom{n}{t}$ possible strings have this property? It is the number of ways to put $(c-1)$ 1's in the first $(m+c-1)$ digits, a $1$ in the $(m+c)$'th digit, and $(t-c)$ 1's in the remaining $(n-m-c)$ digits. Thus,
\begin{equation}\label{eq:prjkm}
\Pr\left[\sum_{k=1}^{c} |J_k|=m\right] 
= \frac{\binom{m+c-1}{c-1}\binom{n-m-c}{t-c}}{\binom{n}{t}}
\le \frac{\binom{m+c-1}{c-1}\binom{n-c}{t-c}}{\binom{n}{t}}
= \binom{m+c-1}{c-1}\cdot\frac{t^{\underline c}}{n^{\underline c}}
\le \binom{m+c-1}{c-1} \left(\frac t n\right)^c
\!\!\!,\end{equation}
where $n^{\underline c}:=n\cdot(n-1)\cdots(n-c+1)$ denotes the falling factorial. Now we combine (\ref{eq:ejc}) and (\ref{eq:prjkm}), and relax the bound by allowing $m$ to go up to infinity. The resulting series converges when $0<x<1$.
\begin{align*}
\E\left[x^{|J_C|}\right]
\le \left(\frac{t}{n}\right)^c\sum_{m=0}^{\infty}\binom{m+c-1}{c-1}x^m
=  \left(\frac{t}{n}\right)^c\frac{1}{(1-x)^c}
= \left(\frac{t}{n(1-x)}\right)^c.
\end{align*}
A quick proof of the series' convergence is to start with $\sum_{m=0}^{\infty}x^m=1/(1-x)$, and take the $(c-1)$'th derivative of both sides, with respect to $x$.
\end{proof}

\begin{lemma}\label{lemma:exprod}
Let $\delta_k:=\prod_{i=0}^{\lfloor |J_k|/2\rfloor}(1-\alpha_{j_{k,2i}}\alpha_{j_{k,2i+1}})$, and $\alpha:=\min_j\{\alpha_j\}$. Let $C\subseteq[t]$ of size $c$. If $\frac{a_{max}}{a_{min}}\le2$, 
\[ \E\left[ \prod_{k\in C} \frac{\delta_k}{q_{\sigma(k)}} \right] 
\le \left(\frac{16}{7}\cdot\frac{t}{n\alpha^3}\right)^c . \]
\end{lemma}

\begin{proof} First, recall by definition that $q_i$ is either $p_i$ or $1-p_i$, so $q_i\ge \min\{p_i,1-p_i\} =\alpha_i\ge \alpha$. Then
\begin{align} 
\E\hspace{-1pt}\left[ \prod_{k\in C} \frac{\delta_k}{q_{\sigma(k)}} \right] 
&\le \frac{1}{\alpha^c}\E\left[ \prod_{k\in C} \delta_k \right] 
=\frac{1}{\alpha^c} \E\left[ \prod_{k\in C} \prod_{i=1}^{\lfloor |J_k|/2\rfloor}(1-\alpha_{j_{k,2i}}\alpha_{j_{k,2i+1}})\right] 
\le \frac{1}{\alpha^c}\E\left[ \prod_{k\in C} (1-\alpha^2)^{(|J_k|-1)/2}\right] \label{eq:am-gm}
\\&= \frac{1}{\alpha^c}\E\left[(1-\alpha^2)^{(|J_C|-c)/2}\right]
\le \frac{1}{\alpha^c}\E\left[\left(1-{\textstyle\frac{\alpha^2}{2}}\right)^{|J_C|-c}\right] \label{step:sqrt-bound}
\\&\le \frac{1}{\alpha^c}\cdot \left(1-\frac{\alpha^2}{2}\right)^{-c}\left(\frac{t}{n(\alpha^2/2)}\right)^c
\le \left(1-\frac{(1/2)^2}{2}\right)^{-c} \left(\frac{2t}{n\alpha^3}\right)^c
=\left(\frac{16}{7}\cdot\frac{t}{n\alpha^3}\right)^c . \label{step:lemma-exjc}
\end{align}
In the first line we used $\lfloor x/2\rfloor\ge (x-1)/2$. In (\ref{step:sqrt-bound}) we used $\sqrt{1-x^2}\le\sqrt{1-x^2+x^4/4}=1-x^2/2$. In (\ref{step:lemma-exjc}) we applied Lemma \ref{lemma:exjc} and then used $\alpha\le1/2$.
\end{proof}
Now we can complete the bound given in Lemma \ref{lemma:sandwich}. The upper bound follows by expanding the binomial, bounding the expected value of each term, and then refactoring. The lower bound follows by the Weierstrass product inequality.
\begin{align}
\E\left[\prod_{k=1}^{t-1} \left(1 + \frac{\delta_k}{q_{\sigma(k)}}\right)\right] 
 &=1+\sum_{1\le i<t}\E\left[\frac{\delta_i}{q_{\sigma(i)}}\right]
	+\sum_{1\le i<j<t}\E\left[\frac{\delta_i\delta_j}{q_{\sigma(i)}q_{\sigma(j)}}\right]
	+\cdots+\E\left[\frac{\delta_1\cdots \delta_{t-1}}{q_{\sigma(1)}\cdots q_{\sigma(t-1)}}\right]
\nonumber\\ &\le1+\!\sum_{1\le i<t}\left(\frac{16t}{7 n\alpha^3}\right)
	+\!\!\sum_{1\le i<j<t} \left(\frac{16t}{7 n\alpha^3}\right)^2 
	\!\!\!+\cdots+\left(\frac{16t}{7 n\alpha^3}\right)^{\!t-1}
\!\!\!\!= \left(1+\frac{16t}{7 n\alpha^3}\right)^{\!t-1}.
\label{eq:exp-up}\end{align}
\begin{align}
\E\left[\prod_{k=1}^{t-1} \max\left\{1 - \frac{\delta_k}{q_{\sigma(k)}},0\right\} \right]
&= \E\left[\prod_{k=1}^{t-1} \left(1-\min\left\{\frac{\delta_k}{q_{\sigma(k)}},1\right\}\right)\right]
\ge \E\left[1-\sum_{k=1}^{t-1} \min\left\{\frac{\delta_k}{q_{\sigma(k)}},1\right\}\right]
\nonumber\\ &\ge 1-\sum_{k=1}^{t-1} \E\left[\frac{\delta_k}{q_{\sigma(k)}}\right]
\ge 1-\sum_{k=1}^{t-1} \frac{16t}{7n\alpha^3}
= 1-\frac{16t(t-1)}{7n\alpha^3}.
\label{eq:exp-down}\end{align}

Thus we are led to our main theorem on dependent rounding (which in turn is improved upon, with further work, in Section~\ref{sec:dep-round-special}): 

\begin{theorem}\label{thm:limited-dep-general}
Let $(X_1,\ldots,X_n)$ be the vector returned by running {\sc DepRound} with probabilities $(p_1,\ldots,\allowbreak p_n)$ and positive weights $(a_1,\ldots,a_n)$. Let $I^+$ and $I^-$ be disjoint subsets of $[n]$. Define $\alpha:=\min_i\{p_i,1-p_i\}$, $I:=I^+\cup I^-$, $t = |I|$, and $\lambda:=\displaystyle\prod_{i\in I^+}p_i\prod_{i\in I^-}(1-p_i)$. Then if $\displaystyle \max_{i,j}\left\{\frac{a_i}{a_j}\right\}\le 2$, we have
\[ \left(1-\frac{16t(t-1)}{7n\alpha^3}\right) \lambda 
\le \E\left[\prod_{i\in I^+} X_i \prod_{i \in I^-}(1-X_i) \right] 
\le\left(1+\frac{16t}{7n\alpha^3}\right)^{t-1} \lambda.\]
\end{theorem}
\begin{proof}
For all $i\in I^+$, set $b_i=1$; for all $i\in I^-$, set $b_i=0$. Then apply (\ref{eq:exp-up}) and (\ref{eq:exp-down}) to Lemma \ref{lemma:sandwich}. The theorem follows by recognizing that
\begin{equation*}
 \prod_{i\in I^+} X_i \prod_{i \in I^-}(1-X_i) = \prod_{i\in I} Y_i \qquad\text{and}\qquad
	\prod_{i\in I^+}p_i\prod_{i\in I^-}(1-p_i) = \prod_{i\in I}q_i
.\end{equation*}
\end{proof}
Note that $\left(1+\frac{16t}{7n\alpha^3}\right)^{t-1}\le\exp\left(\frac{16t^2}{7n\alpha^3}\right)$. Thus we see that Theorem \ref{thm:limited-dep-general} allows us to bound the dependence among groups of variables as large as $O(\sqrt{n})$ when $\alpha = \Theta(1)$.
%============ Section 2.4
\subsection{Improvements and Special Cases}
\label{sec:dep-round-special}
In this section we present several refinements of Theorem \ref{thm:limited-dep-general}. The proofs all follow the same outline as that of the main result; we describe only the places where they differ.  We reuse the same definitions unless stated otherwise.

In our $k$-median application we will have that all $p_i$ are uniform. In this case, if the maximum ratio of weights is sufficiently small, we can tighten the bound to show a weaker dependency on $\alpha$.
\begin{theorem}\label{cor:dep-uniform-p}
Suppose $p_1=p_2=\cdots=p_n=p$ and $\alpha = \min\{p, 1 - p\}$. Then if $\frac{a_{max}}{a_{min}}\le 1+\alpha$, we have  
\[ \left(1-\frac{8t(t-1)}{3n\alpha^2}\right) \lambda
\le \E\left[\prod_{i\in I^+} X_i \prod_{i \in I^-}(1-X_i) \right]
\le\left(1+\frac{8t}{3n\alpha^2}\right)^{t-1} \lambda  .\]
\end{theorem}

\begin{proof}
The improvement comes from strengthening the result of Lemma \ref{lemma:pair-prob}: Suppose $X_i$ is co-rounded with variable $X_j$ in step $s$ and then (if it survives) variable $X_k$ in step $s+1$, where $X_j^{(s-1)}=X_k^{(s)}=p$. Then we can show $X_i$ will be fixed in one of these two steps with probability at least $\alpha$.

First assume that during both steps Case III occurs. 
By requirements for Case III, we have $\frac{a_j}{a_i}(1-p)<X_i^{(s-1)}<1-\frac{a_j}{a_i}p$ and $\frac{a_k}{a_i}(1-p)<X_i^{(s)}<1-\frac{a_k}{a_i}p$. 
Suppose $X_j$ is fixed to 0 in step $s$. Then 
$$X_i^{(s)}=X_i^{(s-1)}+p\frac{a_j}{a_i}>\frac{a_j}{a_i}(1-p)+p\frac{a_j}{a_i}=\frac{a_j}{a_i}\ge\frac{1}{1+\alpha}=1-\frac{1}{1+\alpha}\cdot\alpha\ge 1-\frac{a_k}{a_i} p>X_i^{(s)}.$$ 
Else suppose $X_j$ is fixed to 1. Then
$$X_i^{(s)}=X_i^{(s-1)}-(1-p)\frac{a_j}{a_i}<1-\frac{a_j}{a_i}p-(1-p)\frac{a_j}{a_i}=1-\frac{a_j}{a_i}<1-\frac{1}{1-\alpha}
=\frac{1}{1-\alpha}\cdot\alpha<\frac{a_k}{a_i}(1-p)<X_i^{(s)}.$$

In either outcome, we have a contradiction. Therefore in at least one of the two steps, a case other than III must occur. As shown in the proof of Lemma \ref{lemma:pair-prob}, in the other 3 cases $X_i$ will be fixed with probability at least $\min\{p,1-p\}=\alpha$.

This stronger bound carries through the remaining lemmas in a straightforward way. Following the proof for Lemma \ref{lemma:zbound}, starting from (\ref{eq:pair_prod}), we get
$$\Pr[Z_k]
%\le\prod_{i=1}^{\lfloor |J_k|/2\rfloor}\Pr\left[ E_{2i-1}\land E_{2i}\,\Big\vert \bigwedge_{j=0}^{2i-2} E_j\right]
\le\prod_{i=1}^{\lfloor |J_k|/2\rfloor}\Pr\left[ E_{2i-1}\land E_{2i}\,\Big\vert E_{2i-2}\right]
\le \prod_{i=1}^{\lfloor |J_k|/2\rfloor} (1-\alpha)=(1-\alpha)^{\lfloor |J_k|/2\rfloor}.
$$
So Lemmas \ref{lemma:zbound}, \ref{lemma:sandwich1} and \ref{lemma:sandwich} will now hold with the new definition $\delta_k:=(1-\alpha)^{\lfloor |J_k|/2\rfloor}$. Then in Lemma \ref{lemma:exprod}, we get
\begin{align*} 
\E\hspace{-1pt}\left[ \prod_{k\in C} \frac{\delta_k}{q_{\sigma(k)}} \right] 
&\le\frac{1}{\alpha^c} \E\left[ \prod_{k\in C} (1-\alpha)^{\lfloor |J_k|/2\rfloor}\right] 
\le \frac{1}{\alpha^c}\E\left[ \prod_{k\in C} (1-\alpha)^{(|J_k|-1)/2}\right] 
\\&= \frac{1}{\alpha^c}\E\left[(1-\alpha)^{(|J_C|-c)/2}\right]
\le \frac{1}{\alpha^c}\E\left[\left(1-{\textstyle\frac{\alpha}{2}}\right)^{|J_C|-c}\right]
\\&\le \frac{1}{\alpha^c}\cdot \left(1-\frac{\alpha}{2}\right)^{-c}\left(\frac{t}{n(\alpha/2)}\right)^c
\le \left(1-\frac{(1/2)}{2}\right)^{-c} \left(\frac{2t}{n\alpha^2}\right)^c
=\left(\frac{8}{3}\cdot\frac{t}{n\alpha^2}\right)^c . \label{step:lemma-exjc}
\end{align*}
The theorem follows as before.
\end{proof}

In the unweighted case (where all $a_i=1$), we can similarly tighten the bound. We can also refine the bound to be in terms of a sort of average of the probabilities instead of just the most extreme. 
\begin{theorem}
\label{thm:limited-dep-unweighted}
Let $\X:=(X_1,\ldots,X_n)$ be the vector returned by running {\sc DepRound} with probabilities $(p_1,\ldots,p_n)$ and unit weights $(1,\ldots,1)$. Let $\alpha_i=\min\{p_i,1-p_i\}$. Let $I^+$ and $I^-$ be disjoint subsets of $[n]$.  Let $q_i=p_i$ for $i\in I^+$, and let $q_i=1-p_i$ for $i\in I^-$; let
$I = I^+ \cup I^-$. Define
\[
J = ([n] \setminus I), ~~~\lambda:=\displaystyle\prod_{i\in I^+}p_i\prod_{i\in I^-}(1-p_i), ~~~
\alphavg:=\frac{1}{|J|}\sum_{j\in J}\alpha_j, ~~~
\text{and}~~ \frac1{\qavg}:=\frac{1}{|I|}\sum_{i\in I} \frac{1}{q_i}
.\] 
Then,
\[ \left(1-\frac{t(t-1)}{n\qavg\alphavg}\right) \lambda
\le \E\left[\prod_{i\in I^+} X_i \prod_{i \in I^-}(1-X_i) \right]
\le\left(1+\frac{t}{n\qavg\alphavg}\right)^{t-1} \lambda  .\]
Furthermore, if $\sum_i p_i$ is an integer, then $\X$ has no fractional elements.
\end{theorem}

\begin{proof}
Uniform weights allow us to strengthen Lemma \ref{lemma:pair-prob} even further; in particular, we no longer need to consider pairs of steps. Suppose $X_i$ is co-rounded with $X_j$ during step $s$. Since all $a_i=1$, Cases II and III cannot occur. Thus, $X_i$ will be fixed with probability at least $\min\{X_j^{(s-1)},1-X_j^{(s-1)}\}=\alpha_j$ (as in proof of Lemma \ref{lemma:pair-prob}).

If we follow the proof of Lemma \ref{lemma:zbound}, but without splitting events into pairs, we can show
$$\Pr[Z_k]
\le\prod_{i=1}^{|J_k|}\Pr\left[ E_i\,\Big\vert E_{i-1}\right]
\le \prod_{i=1}^{|J_k|} (1-\alpha_{j_{k,i}})=\prod_{j\in J_k}(1-\alpha_j).
$$
So Lemmas \ref{lemma:zbound}, \ref{lemma:sandwich1} and \ref{lemma:sandwich} will now hold with the new definition $\delta_k:=\prod_{j\in J_k}(1-\alpha_j)$. Now (as in Lemma \ref{lemma:exprod}), we wish to upper bound $\E[\prod_{k\in C} \frac{\delta_k}{q_{\sigma(k)}}]$. Recall the expectation here is conditioned on the random permutation $\pi$. We may decompose $\pi$ into 3 independent components. First, recall $\Phi_I(\pi)=:\phi$ is the binary string corresponding to $\pi$ with $t$ 1's representing the locations of indices in $I$. Second, let $\pi_I\in S_t$ be the permutation representing the ordering of $I$ over the 1's in $\phi$. Third, let $\pi_J\in S_{n-t}$ be the permutation representing the ordering of $J$ over the 0's in $\phi$. Then $\pi$ is uniquely defined by the tuple ($\phi$, $\pi_I$, $\pi_J$) and vice versa. So we can think of $\pi$ as being generated by choosing each element of the tuple uniformly at random. 
Thus, the value of $q_{\sigma(k)}$ depends only on $\pi_I$; the sizes $\{|J_k|\}$ depend only on $\phi$; and the elements of $\{J_k\}$ (conditioned on a particular set of sizes) depend only on $\pi_J$. This shows that some of the variables are independent, so we may separate the terms. Here we are explicit over which random variable we take each expectation:
\begin{align} 
\E_\pi\hspace{-2pt}\bigg[ \prod_{k\in C} \frac{\delta_k}{q_{\sigma(k)}} \bigg] 
&= \E_{\pi_I}\hspace{-2pt}\bigg[\prod_{k\in C}\frac{1}{q_{\sigma(k)}}\bigg] 
\E_{\pi_J,\phi}\hspace{-2pt}\bigg[ \prod_{k\in C} \prod_{j\in J_k}(1-\alpha_j)\bigg] \nonumber
\\&= \E_{\pi_I}\hspace{-2pt}\bigg[\prod_{k\in C}\frac{1}{q_{\sigma(k)}}\bigg] 
\sum_{\phi}\Pr[\phi]\cdot\E_{\pi_J}\hspace{-2pt}\bigg[ \prod_{j\in J_C}(1-\alpha_j)\Big\vert \phi\bigg].\label{eq:exp_decompose} 
\end{align}
The following lemma is basically a restatement of Maclaurin's inequality for symmetric polynomials:
\begin{lemma}\label{lemma:maclaurin}
Given a vector of positive reals $\mathbf{x}=x_1,x_2,\ldots,x_n$, with average value $\bar x$, let $S\subseteq[n]$ be a subset chosen uniformly at random from all such subsets of size $k$. Then $\E_S[\prod_{i\in S} x_i]\le {\bar x}^k$.
\end{lemma}
The first expectation in (\ref{eq:exp_decompose}) is over a product of $c$ random terms from $\{1/q_j\}_{j\in I}$. The second expectation is a product over $|J_C|$ (which as a function of $\phi$ is fixed for each term) random terms from $\{1-\alpha_j\}_{j\in J}$. So both expectations may be bounded by Lemma \ref{lemma:maclaurin}:
\begin{align*} 
\E_\pi\hspace{-2pt}\bigg[ \prod_{k\in C} \frac{\delta_k}{q_{\sigma(k)}} \bigg] 
&\le \frac{1}{\hat q^c}
\sum_{\phi}\Pr[\phi]\cdot (1-\hat\alpha)^{|J_C(\phi)|}
=\frac{1}{\hat q^c}\E_\phi\big[(1-\hat\alpha)^{|J_C|}\big]
\le\frac{1}{\hat q^c}\left(\frac{t}{n\hat\alpha}\right)^c
=\left(\frac{t}{n\hat q\hat\alpha}\right)^c.
\end{align*}
The theorem follows as before.
\end{proof}

\subsubsection{An alternative lower bound}
All the lower bounds given thus far become negative for $t$ larger than $O(\sqrt{n})$. We now derive an alternative lower bound which remains positive even for larger values of $t$. We will do this for the uniform weight case for simplicity, but it may be adapted in a straightforward manner to the weighted case. 

\begin{theorem}\label{thm:alt-lower-bound}
Suppose $a_1=a_2=\cdots=a_n=1$. Let $d$ be an integer which satisfies $(1-\alpha)^d\le \alpha$ and $d\le (n-t)/t$. Then
\[ \left(1-\frac{td}{n-t}\right)^t \left(1 - \frac{ (1-\alpha)^d}{\alpha}\right)^{t-1} \lambda \le \E\left[\prod_{i\in I^+} X_i \prod_{i \in I^-}(1-X_i) \right]. \]
\end{theorem}
\begin{proof}
Start with the lower bound given by Lemma \ref{lemma:sandwich}. As shown in the proof of Theorem \ref{thm:limited-dep-unweighted}, if all $a_i=1$, we may use $\delta_k:= \prod_{i\in J_k}(1-\alpha_i)\le (1-\alpha)^{|J_k|}$. To lower bound the expression, we will focus on the event that sets $J_1,\ldots,J_{t-1}$ all have at least $d$ elements, and use the trivial bound of 0 if this event does not occur. This event is useful because it implies that $\{X_i\}_{i\in I}$ are all far away in $\pi(X)$.
\begin{align*}
\E\left[\prod_{k=1}^{t-1} \max\left\{1 - \frac{\delta_k}{q_{\sigma(k)}},0\right\} \right]
\ge \E\left[\max\left\{1 - \frac{ (1-\alpha)^{|J_k|}}{\alpha},0\right\}^{t-1} \right]
\ge \Pr\left[\bigwedge_{k=1}^{t-1} |J_k|\ge d\right] 
\left(1 - \frac{ (1-\alpha)^d}{\alpha}\right)^{t-1}.
\end{align*}
To calculate this probability, recall that in Lemma \ref{lemma:bijection}, we showed that the distribution of $(|J_i|,\ldots,|J_t|)$ is equivalent to the uniform distribution over unique arrangements of $n-t$ identical balls into $t+1$ boxes. Note there are $\binom{n}{t}$ such arrangements. How many of these arrangements have at least $d$ balls in the middle $t-1$ boxes? To count these arrangements, we suppose that there are already exactly $d$ balls in each of the middle $t-1$ boxes and then count how many ways there are to add the remaining $n-t-(t-1)d$ balls to $t+1$ boxes, which is $\binom{n-(t-1)d}{t}$. So,
\begin{align*}
\Pr\left[\bigwedge_{k=1}^{t-1} |J_k|\ge d\right] 
= \frac{\binom{n-(t-1)d}{t}}{\binom{n}{t}}
> \left(\frac{n-(t-1)d-t}{n-t}\right)^t
> \left(1-\frac{td}{n-t}\right)^t.
\end{align*}
\end{proof}

We show that if $\alpha=\Theta(1)$, and $n$ is sufficiently large, then Theorem \ref{thm:alt-lower-bound} gives a nontrivial bound for $t=O(n/\ln n)$ and a tight bound (close to $\lambda$) for some $t=O(\sqrt{n/\ln n})$.

First suppose $t\le \kappa \frac{n}{\ln n}$ and set $d=\lceil\frac{\ln n}{2\kappa}\rceil$, for some $\kappa>0$. Assume $\ln n> \frac{2\kappa}{\alpha}(\alpha+\ln(1/\alpha))>2\kappa$. Then we have $(1-\alpha)^d\le e^{-\alpha d}\le e^{-\alpha(\frac{\ln n}{2\kappa}-1)}< e^{-(\alpha+\ln(1/\alpha))+\alpha}=\alpha$ and $d\le\frac{\ln n}{2\kappa}=\frac{\ln n}{\kappa}-\frac{\ln n}{2\kappa}< \frac{\ln n}{\kappa}-1\le\frac{n}{t}-1=\frac{n-t}{t}$, so $d$ is valid. These two inequalities also imply that the bound is positive.

%%%%%%%%%%%%%%%%%%5

Now suppose for some $\epsilon\in(0,1]$ that $t\le \sqrt{\frac{\alpha\epsilon n}{4\ln n}}$ and set $d=\lceil \frac{\ln n}{\alpha}\rceil$. 
Assume $n>\max\{\frac{2\epsilon}{\alpha}\ln\frac{2\epsilon}{\alpha},\frac{e^{2\alpha}}{\alpha\epsilon}\}$, which implies $\frac{n}{\ln n}>\frac{(2\epsilon/\alpha)\ln(2\epsilon/\alpha)}{\ln(2\epsilon/\alpha)+\ln\ln(2\epsilon/\alpha)}>\frac{(2\epsilon/\alpha)\ln(2\epsilon/\alpha)}{2\ln(2\epsilon/\alpha)}=\frac{\epsilon}{\alpha}$. 
For simplicity of the argument, observe that
$n-t\ge n-\sqrt{\frac{\alpha\epsilon n}{4\ln n}}\ge n-\sqrt{\frac{n\cdot n}{4\cdot 1}}=\frac{n}{2}$.
Then we have 
$(1-\alpha)^d\le e^{-\alpha d}
\le e^{-\alpha(\frac{\ln n}{\alpha}-1)}
= \frac{e^\alpha}{n}
<\frac{\alpha\epsilon}{e^\alpha}\le\alpha$ 
and 
$d\le\frac{\ln n}{\alpha}
=\sqrt{\frac{\ln n}{n}\cdot\frac{n\ln n}{\alpha^2}}
\le\sqrt{\frac{n\ln n}{\alpha\epsilon}}\le\frac{n}{2t}\le\frac{(n-t)}{t}$, so $d$ is valid. Then
\begin{align*}
\left(1-\frac{td}{n-t}\right)^t 
&\ge \left(1-\frac{t \frac{\ln n}{\alpha}}{n/2}\right)^t 
%=\left(1-\frac{2t \ln n}{\alpha n}\right)^t 
\ge 1-\frac{2t^2 \ln n}{\alpha n} 
\ge 1-\frac{\alpha\epsilon n}{4\ln n}\cdot\frac{2 \ln n}{\alpha n}
= 1-\frac{\epsilon}{2},
\\
\left(1 - \frac{ (1-\alpha)^d}{\alpha}\right)^{t-1}
&\ge \left(1 - \frac{1}{\alpha}\cdot\frac{e^\alpha}{n}\right)^t
\ge 1 - t\frac{e^\alpha}{\alpha n}
\ge 1 - \sqrt{\frac{\alpha\epsilon n}{4\ln n}}\cdot\frac{e^\alpha}{\alpha n}
> 1-\sqrt{\frac{e^{2\alpha} \epsilon}{4\alpha n}}
> 1- \frac{\epsilon}{2}.
\end{align*}
This implies the bound is at least $(1-\epsilon)\lambda$.

\section{Improved bi-point rounding algorithm} 
\label{sec:bipoint-round}

\subsection{A lower bound on bi-point rounding factors}
For a given $k$-median instance $\I$, we can apply the JMS algorithm from \cite{jain2003greedy} to obtain a bi-point solution whose cost is at most $2 \cdot OPT_\I$. In Section 3, we address the step of rounding a bi-point solution to an integral solution. As a warmup, we begin with a concrete example, which will also demonstrate a lower bound on the approximation factor of this step.

We define a family of bi-point solutions and show that the optimal rounding factor, even when opening $k+o(k)$ facilities, approaches $\frac{1+\sqrt2}2\approx 1.207$ for large instances. When counting facilities we will use fractional values proportional to $k$, and assume that $k$ is sufficiently large so that the effect of rounding these to integer values is negligible. Then define the instance as follows.

Let $\F_1$ and $\F_2$ be facility sets of size $f_1k$ and $f_2k$, respectively, for some constants $f_1<1$ and $f_2>1$. Then it follows that $(a,b)=(\frac{f_2-1}{f_2-f_1},\frac{1-f_1}{f_2-f_1})$. Define the client set $\J$ as follows: for every pair of facilities $i_1\in\F_1$ and $i_2\in\F_2$, place a single client $j$ with $d(j,i_1)=\alpha$ and $d(j,i_2)=1-\alpha$, for some constant $\frac12<\alpha\le1$. Let all other distances be the maximal such values permitted by the triangle inequality. This means for every $i\in\F_1\setminus\{ i_1\}$ we have $d(j,i)=2-\alpha$, and for every $i\in\F_2\setminus\{ i_2\}$ we have $d(j,i)=1+\alpha$.

Because of the symmetry of the instance, any integer solution may be uniquely defined by the proportion of facilities opened in $\F_1$ and $\F_2$. Opening less than $k$ facilities can only hurt the solution, so assume we open exactly $k$. Let $S(x)$ be a solution that opens $xf_1k$ facilities in $\F_1$ and $k-xf_1k=(1-xf_1)k$ facilities in $\F_2$. Also, we always open at least 1 facility in $\F_1$, even if we have to borrow 1 from $\F_2$. For sufficiently large $k$, this does not affect the proportions. 

Since it doesn't matter which facilities we open within either set, suppose (for ease of analysis) we randomly open $xf_1k$ facilities in $\F_1$ and $(1-xf_1)k$ facilities in $\F_2$. What is the expected cost of a client $j$? The closest facility is $i_2(j)$ of distance $1-\alpha$, followed by $i_1(j)$ of distance $\alpha$. The third closest facility is any other facility in $\F_1$; these are all $2-\alpha$ away, and at least one will always be open. Thus we may calculate the expected distance as follows. Note that we open in proportion $x$ of $\F_1$ and $\frac{1-xf_1}{f_2}$ of $\F_2$, independently of one another.
\begin{align*}
E[COST(j)]&=\frac{1-xf_1}{f_2}(1-\alpha)+\left(1-\frac{1-xf_1}{f_2}\right)\left(x\alpha+(1-x)(2-\alpha)\right).
\end{align*}
The expression is quadratic in $x$ with a negative coefficient on $x^2$. Thus, it will be minimized at one of the two edge cases $x=0$ or $x=1$, and one of these two must yield the optimal solution. Summing over all clients, and observing that the total cost is actually deterministic, we get
\[ OPT=|\J|\min\left\{2-\alpha-\frac{1}{f_2},\alpha+\frac{(2\alpha-1)(f_1-1)}{f_2}  \right\}. \]
On the other hand, the cost of the bi-point solution itself is 
\[ a|\J|\alpha+b|\J|(1-\alpha)=|\J|\frac{(1-f_2)\alpha+(f_1-1)(1-\alpha)}{f_1-f_2}. \]
Now fixing $\alpha=\frac{1}{\sqrt{2}}$, $f_1=\frac17(4-\sqrt2)$ and $f_2=\frac27(3+\sqrt2)$, we get that the ratio of cost between the optimal integer solution and the bi-point solution is $\frac{1+\sqrt2}{2}$.

Finally, suppose we take any $S(x)$ and open $o(k)$ additional facilities in either or both sets. Then the respective proportions we open of $\F_1$ and $\F_2$ are $\frac{xf_1k+o(k)}{f_1k}=x+o(1)$ and $\frac{(1-xf_1)k+o(k)}{f_2k}=\frac{1-xf_1}{f_2}+o(1)$. Thus, for sufficiently large $k$, the increase to the proportions is negligible and we obtain the same cost ratio.

In this instance, the algorithm by Li and Svensson opens $(\F_1,\F_2)$ in proportions either $(a,b)$ or $(1,0)$, and does strictly worse than the optimal factor. The new algorithm considers a solution that opens no facilities in $\F_1$, which is crucial to obtaining an improved factor.

\subsection{Preliminaries}

 We refer to $\F_1, \F_2, a, b, D_1,$ and $D_2$ as defined in Definition \ref{def:bipoint}.

\begin{definition}[Stars] For a given bi-point solution $a\F_1 + b\F_2$, we associate each facility $i_2 \in \F_2$ to its closest facility $i_1 \in \F_1$ (breaking ties arbitrarily). For each $i \in \F_1$, the set of $i$ and its associated facilities in $\F_2$ is called a star. We refer to $i$ as the center of the star and other facilities in the star as leaves.  Also let $\S_i$ denote the set of leaves of the star with center $i$.
\end{definition}

Now we further partition the stars by their number of leaves. Let $\T_0$ be the set of stars with no leaves, $\T_1$ be the set of stars with one leaf, and $\T_2$ be the set of stars with at least 2 leaves. We call the stars in $\T_0, \T_1, \T_2$ as 0-stars, 1-stars, and 2-stars, respectively. Let $\C_0, \C_1, \C_2$ be the sets of centers of stars in $\T_0, \T_1, \T_2$, respectively. Let $\L_1, \L_2$ be the sets of leaves of stars in $\T_1, \T_2$, respectively. For a client $j$, let $i_1(j)$ and $i_2(j)$ denote the closest facilities to $j$ in $\F_1$ and $\F_2$ respectively.

We also use the following notations: $\Delta_F := |\F_2| - |\F_1|, r_D := D_2 / D_1, r_0 := |\C_0| / \Delta_F, r_1 := |\C_1| / \Delta_F, $ $r_2 := |\C_2| / \Delta_F$, and $s_0 := 1/(1+r_0)$.
Note that if $\Delta_F=0$, then $|\F_1|=|\F_2|=k$, and we may simply choose $\F_2$ as our solution, which has cost at most that of the bipoint solution. Thus we assume that $\Delta_F>0$.

In this section, we describe a set of randomized algorithms to round a bi-point solution into a pseudo-solution which opens at most $k+O(1)$ facilities. In order to keep the number of extra facilities bounded, we consider several different cases depending on certain properties of the bi-point solution. In the main case, we get a $1.3371 +\epsilon$ approximation, utilizing {\sc DepRound} to open only $O(\log(1/\eps))$ extra facilities. In the edge cases, we are able to use weaker, but simpler techniques to obtain the same bound.

\subsection{Main case: $s_0 \geq 5/6,  b \in [0.508, 3/4], r_D \in [19/40, 2/3],$ and $r_1 > 1$}
\label{sec:bipoint-main-case}

For each 1-star with center $i$ and leaf $i'$, we define the following ratio. (Note that $\Delta_F> 0$ implies $\L_2$ is nonempty.)
$$ g_i = \frac{d(i, i')}{\min_{j \in \L_2} d(i, j)}.$$

We partition the set $\T_1$ into sets $\T_{1A}$ of \emph{long} stars and $\T_{1B}$ of \emph{short} stars as follows.  We sort all the stars in $\T_1$ in decreasing order of $g_i$. Let $\T_{1A}$ be the set of the first $\lceil a \Delta_F \rceil$ stars of $\T_1$ and $\T_{1B} := \T_1 \setminus \T_{1A}$. Also let $\C_{1A}$ and $\C_{1B}$ be the sets of centers of stars in $\T_{1A}$ and $\T_{1B}$, respectively. Similarly,  let $\L_{1A}$ and $\L_{1B}$ be the corresponding sets of leaves. Note that $\T_{1A}$ is well-defined since $|\T_1| / \Delta_F = r_1 > 1$ implies $|\T_1| > \Delta_F$.

Next, we describe a rounding scheme called $\A(p_0,p_{1A},q_{1A},p_{1B},q_{1B},p_2,q_2)$ which is the main procedure of our algorithm. 
The purpose of $\A$ is to (for $X\in\{0,1_A,1_B,2\}$) randomly open roughly $p_X$ fraction of facilities in $\C_X$, and $q_X$ fraction of facilities in $\L_X$, while maintaining the important property that if any leaves of a star are closed, its center will be opened -- \emph{except} in some cases where we completely close all stars in $\T_{1A}$.

When $p_2 \neq 0$, we further partition $\T_2$ into ``large'' and ``small'' stars (as in \cite{li_svensson}). For a given parameter $\eta > 0$, we say that a star centered at $i \in \C_2$ is large if $|\S_i| \geq 1 / ( p_2 \eta )$ and small otherwise. Let $\beta = \min\{q_2,1-q_2\}$ and $c = \lceil\frac{16}{3\beta^2}\rceil$. Then, we group the small stars according to their sizes: For each $s = 1, \ldots, \lceil \log_{1+\beta} ( 1 / (p_2 \eta)) \rceil - 1 $, let
$ \G_s := \{i \in \C_2 : (1+\beta)^s \leq |\S_i| < (1+\beta)^{s+1} \}$. 

\subsubsection{Main algorithm}

Below we define Algorithm $\A$ and its subroutine {\sc Round2Stars}. 
 The main algorithm will simply run $\A$ with 9 different sets of parameters and return the solution with minimum connection cost. We refer to these calls of $\A$ as algorithms $\A_1, \cdots, \A_9$. See Table \ref{tab:algos} for a complete set of parameters. It is easy to see that all numbers in the table belong to $[0,1]$ as $b \geq a$, $0 \leq s_0 \leq 1$, and $r_2 \geq 0$.

\begin{algorithm}[h]
\caption{$\A(p_0,p_{1A},q_{1A},p_{1B},q_{1B},p_2,q_2)$}
\begin{algorithmic}[1]
\STATE Randomly open a subset of size $\lceil p_0 |\C_0| \rceil$ of $\C_0$.
\STATE Take a random permutation of $\T_{1A}$. Open the centers of the first $\lceil p_{1A} |\T_{1A}| \rceil$ stars and the leaves of the last $\lceil q_{1A} |\T_{1A}| \rceil$.
\STATE Take a random permutation of $\T_{1B}$. Open the centers of the first $\lceil p_{1B} |\T_{1B}| \rceil$ stars and the leaves of the last $\lceil q_{1B} |\T_{1B}| \rceil$.
\IF{$p_2 = 1$ \textbf{or} $p_2 = 0$}
\STATE	Open all or none of $\C_2$, respectively. Also open a random subset of size $\lceil q_2 |\L_2| \rceil$ of $\L_2$.
\ELSE
\STATE \sc Round2Stars $(p_2,q_2)$.
\ENDIF
\STATE \textbf{Return} the set of all opened facilities.
\end{algorithmic} 
\label{algo:A}
\end{algorithm}

\begin{algorithm}[h]
\caption{\sc Round2Stars$(p_2,q_2)$}
\begin{algorithmic}[1]
\STATE Open the centers of all large stars. Let $\C_2'$ be the set of these centers, and let $\L_2'$ be the set of their leaves. Randomly open a subset of size $\lceil q_2(|\L_2'| - |\C_2'|) \rceil$ of $\L_2'$.
\FOR{$s = 1, \ldots, \lceil \log_{1+\beta} ( 1 / (p_2 \eta)) \rceil - 1 $}
	\STATE Let $A, P$ be vectors with $A_i = |\S_i|-1$ and $P_i = q_2$ for $i \in \G_s$. 
	\STATE Let $X$ be the vector returned by {\sc DepRound} on $A$ and $P$.  	
	\STATE For all integer elements $X_i$, if $X_i = 1$, open all facilities in $\S_i$. Else, open the center $i$.
	\STATE Let $X_{i^*}$ be the fractional element (if any). Open the center of $\S_{i^*}$ and a random set of size $\lceil X_{i^*} |\S_{i^*}| \rceil$ of $\S_{i^*}$.
	\STATE Pick $\min\{c, |\G_s| \}$ centers of stars in $\G_s$ uniformly at random and open them if not already opened. \label{line:addc}
\ENDFOR

\end{algorithmic} 
\label{algo:round2stars}
\end{algorithm}

\begin{table}[H]
\begin{center}
\begin{tabular}{ c | c | c | c | c | c | c | c }
Algorithms & $p_0$ & $p_{1A}$& $q_{1A}$& $p_{1B}$& $q_{1B}$& $p_2$& $q_2$ \\
  \hline                        
$\A_1$ & $0$ & $0$& $1$& $0$& $1$& $as_0$& $1-as_0$ \\
$\A_2$ & $1$ & $0$& $1$& $0$& $1$& $1-bs_0$& $bs_0$ \\
$\A_3$ & $1$ & $0$& $1$& $1$& $0$& $1-bs_0$& $bs_0$ \\
$\A_4$ & $1$ & $1$& $0$& $0$& $1$& $1-bs_0$& $bs_0$ \\
$\A_5$ & $1$ & $1$& $0$& $1$& $0$& $1-bs_0$& $bs_0$ \\
$\A_6$ & $1$ & $1$& $1$& $1$& $0$& $1 - (b-a)s_0$& $(b-a)s_0$ \\
$\A_7$ & $1$ & $1$& $0$& $1$& $0$& $1$& $\frac{1}{2}bs_0$ \\
  $\A_8$ & $0$ & $0$& $0$& $0$& $1$& $0$& $1$ \\
  $\A_9$ & $a$ & $a$& $b$& $a$& $b$& $a$& $b$  \\
  \hline                        
\end{tabular}
\caption{The main algorithm makes 9 calls to $\A$, with the above parameters.} \label{tab:algos}
\end{center}
\end{table}

The algorithm itself runs in linear time. However, when we use Li and Svensson's algorithm to convert our pseudo-solution to a feasible one, it will take time $O(n^{O(C/\eps)})$ in total, where $C$ is the number of extra facilities we open. So it is important that $C$ is a (preferably small) constant. 
A few of these extra facilities come from handling basic rounding (e.g. $\lceil q_2|\L_2|\rceil$), however, the majority come from handling the positive correlation within the groups $\G_s$. Li and Svensson considered $O(1/\eta)$ groups of stars, each with uniform size, bounded the positive correlation by adding a few extra facilities per group, and showed that the total cost is only blown up by a factor of $(1+\eta)$. Property (A2) of {\sc DepRound} allows us to run it on a group with stars of varying sizes. This allows us to use a geometric grouping of stars, and thus open only $O(\log(1/\eta))$ extra facilities.  Property (A3) gives a bound on the positive correlation, so that we may compensate for it by adding $O(1/\beta^2)$ extra facilities per group. Thus, $\beta$ must be bounded away from zero, which strongly motivates our restriction of the domain of the main algorithm. 

Note that if we run $\A$ with parameters $p_0=p_{1A}=p_{1B}=p_2=a$, and $q_{1A}=q_{1B}=q_2=b$, the resulting algorithm is essentially the same as that given in \cite{li_svensson}. (The set of algorithms we use subsumes the need for this one.) The main difference in this case is that we need to open only $O(\log(1/\eps))$ extra facilities instead of $O(1/\eps)$.

\subsubsection{Bounding the number of opened facilities}
Since our main algorithm will return one of the solutions by $\A_1, \cdots, \A_9$, we need to show that none of these will open too many facilities. Algorithm \ref{algo:round2stars} essentially partitions all stars into a constant number of groups.  Consider the \emph{budget} of each group, which is the expected number of facilities opened in that group if we independently open each facility in $\C_X$ with probability $p_X$ and each in $\L_X$ with $q_X$. We want to show that for each group, the number of facilities opened is always within an additive constant of that group's budget. The trickiest groups are the groups of small stars $\{\G_s\}$.

\begin{lemma} \label{lem:groupG} For each group $\G_s$, let $\C(s)$ and $\L(s)$ be the set of centers and leaves of stars in $\G_s$ respectively. Then {\sc Round2Stars} always opens at most $p_2|\C(s)| + q_2|\L(s)| + c + 2$ facilities in $\C(s)\cup\L(s)$.
\end{lemma}
\begin{proof}
By property (A2) of {\sc DepRound}, we have with probability 1 that
$$\sum_{i \in \C(s)} X_i (|\S_i|-1) = \sum_{i \in \C(s)} q_2 (|\S_i|-1).$$
The number of facilities opened in lines 5 and 6 is at most
\begin{align*}
	&\sum_{i \in \C(s), i \neq i^*} \big( X_i |\S_i| + (1-X_i) \big) + 1 + \lceil X_{i^*} |\S_{i^*}| \rceil \\
	\leq& \sum_{i \in \C(s), i \neq i^*}  X_i(|\S_i|-1) + (|\C(s)|-1) + 2 + X_{i^*}(|\S_{i^*}|-1) + X_{i^*}\\
	=& \sum_{i \in \C(s)}  X_i(|\S_i|-1) + (|\C(s)|-1) + 2 + X_{i^*}\\
	=& \sum_{i \in \C(s)}  q_2 (|\S_i|-1)  + |\C(s)| + 1 + X_{i^*}\\
	\le&  \sum_{i \in \C(s)} q_2 (|\S_i|-1) + |\C(s)| + 2 \\
	=&  \sum_{i \in \C(s)} (q_2(|\S_i|-1)+1)  + 2
		= \sum_{i \in \C(s)} (q_2|\S_i|+p_2)  + 2 
		= p_2|\C(s)| + q_2|\L(s)| + 2,
\end{align*}
where in the penultimate step we have used that $p_2+q_2=1$ whenever {\sc Round2Stars} is called. (This follows from $\A_1\ldots \A_9$, except $\A_7$ where {\sc Round2Stars} would never be called.) The lemma follows because we open at most $c$ additional facilities in line 7.
\end{proof}

Note that the number of groups of small stars is at most $\log_{1+\beta}(1/(p_2\eta))$, and we open at most $c+2 = \lceil 16/(3\beta^2) \rceil +2$ additional facilities in each group. It is straightforward to see that the other groups ($\T_{1A}$, $\T_{1B}$, and large stars) only open a constant number of extra facilities, and so our total budget is violated by only a constant amount. The following claim shows that $\beta$ and $p_2$ are strictly greater than zero (i.e., $c$ and the number of groups are  upper-bounded by real constants.) All proofs of the remaining claims in this section are in Appendix \ref{apdx:nfacilities}.
\begin{claim} \label{claim:bound_beta}
When {\sc Round2Stars} is called, we have $\beta > 1/75$ and $p_2 \geq 5/24$.
\end{claim}
Since we open basically $O(\frac{1}{\beta^3}\log(\frac{1}{\eta}))$ extra facilities, these small lower bounds lead to poor constants. Significant improvement may be made by further splitting the cases, and carefully choosing the set of algorithms used in each. However, in order to avoid further complicating the algorithm and its analysis, we do not attempt to optimize these values here.
\begin{lemma} \label{lem:singleA} For any given set of parameters $\{p_0,p_{1A},q_{1A},p_{1B},q_{1B},p_2,q_2 \}$ in Table \ref{tab:algos}, $\A$ will open at most $E + O(\log(1/\eta))$ facilities with probability $1$, where
$$ E: = p_0|\C_0| + p_{1A}|\C_{1A}| + q_{1A}|\C_{1A}| + p_{1B}|\C_{1B}| +  q_{1B}|\C_{1B}|  + p_2|\C_2| + q_2|\L_2| . $$ 
\end{lemma}
The $O(\log(1/\eta))$ term comes as a result of us opening $O(\log(1/\eta))$ small groups $\G_s$. 
The parameters in $\A_1, \cdots, \A_8$ are carefully chosen so that the total budget $E \approx k$ in each case. This gives us the following result.

\begin{lemma} Algorithms $\A_1, \cdots, \A_9$ will always open at most $k + O(\log(1/\eta))$ facilities. \label{lem:kfacilities}
\end{lemma}

\subsubsection{Cost analysis} \label{sec:cost_analysis}

We now derive bounds for the expected connection cost of a single client. 
For each client $j\in \J$, let $i_1(j)$ and $i_2(j)$ be the client's closest facilities in $\F_1$ and $\F_2$, and let $d_1(j)$ and $d_2(j)$ be their respective distances from $j$. Also let $i_3(j)$ be the center of the star containing $i_2(j)$. (Where obvious, we omit the parameter $j$.)
We will obtain several different upper bounds, depending on the class of the star in which $i_1(j)$ and $i_2(j)$ lie. Full derivations of these bounds are in Appendix~\ref{sec:final-cost-proofs}. 
A key characteristic of Algorithm \ref{algo:A} is that for any star in class $Y\in\{1_A,1_B,2\}$, as long as $p_Y+q_Y\ge 1$, it will always open either the star's center or all of the star's leaves.  By definition of stars, we know $i_3$ is not too far away. We will slightly abuse notation and let $i$ and $\bar i$ represent the events that facility $i$ is opened or closed, respectively. By considering these probabilities, we obtain the following two bounds, similar to the one used in \cite{li_svensson}. 

\begin{lemma}\label{lemma:123}
Let $j$ be a client. Suppose we are running one of algorithms $\A_1$ to $\A_7$, \emph{OR} we are running $\A_8$ \emph{and} $i_2(j)\not\in\L_{1A}$. Then the expected connection cost of $j$ after running Algorithm \ref{algo:A} is bounded above by both
$c_{213}(j) := d_2+\Pr[\bar i_2](d_1-d_2)+2\Pr[\bar i_1\bar i_2]d_2$ and 
$c_{123}(j) := d_1+\Pr[\bar i_1](d_2-d_1)+\Pr[\bar i_1\bar i_2](d_1+d_2)$.
\end{lemma} 
In $\A_8$, $p_{1A}=q_{1A}=0$, meaning all stars in $\T_{1A}$ have both center and leaf closed, so if $i_2\in\L_{1A}$ the previous bound does not hold. In this case, let $i_4$ be the closest leaf of a 2-star to $i_3$. Recall the definition of $g_i$; this gives us information on the distance to $i_4$. Let $g:=\min_{i\in\C_{1A}}g_i$ be the minimum value over all stars in $\T_{1A}$. Then we may bound the cost to $i_4$ (or its center, in the worst case) as follows:

\begin{lemma}\label{lemma:145}
Let $j$ be a client such that $i_2(j)\in\L_{1A}$. Then the expected connection cost of $j$, when running $A_8$, is bounded above by
$c_{145}(j):=d_1+\Pr[\bar i_1]\left(2d_2+\frac1g (d_1+d_2)\right)+\Pr[\bar i_1 \bar i_4]\frac1g(d_1+d_2)$.
\end{lemma}

These two lemmas provide a valid bound for all clients. However, the bound in Lemma \ref{lemma:145} may be very poor if $g$ is small. To balance this, we provide another bound which does well for small $g$.

\begin{lemma}\label{lemma:120}
Let $j$ be a client such that $i_1(j)\in \C_{1B}$ and $i_2(j)\in \L_2$. Then in all algorithms, the expected cost of $j$ is bounded above by both of the following:
\begin{align*}
c_{210}(j)&:=d_2+\Pr[\bar i_2](d_1-d_2)+\Pr[\bar i_1\bar i_2]g(d_1+d_2),\\
c_{120}(j)&:=d_1+\Pr[\bar i_1](d_2-d_1)+\Pr[\bar i_1\bar i_2](d_1-d_2+g(d_1+d_2)).
\end{align*}
\end{lemma}
(Note: as we observe in the proof of the above, the coefficient $(d_1-d_2+g(d_1+d_2))$ is nonnegative.) 

The following lemma relates the probabilities in the above bounds to the parameters of the algorithm. In particular, we take advantage of properties (A1) and (A3) of {\sc DepRound} as described in Section \ref{sec:depround}.

\begin{lemma}\label{lemma:probbound}
Let $i_1$ and $i_2$ be any two facilities in $\F_1$ and $\F_2$, respectively. Let $X,Y\in\{0,1_A,1_B,2\}$ be the classes such that $i_1\in\C_X$ and $i_2\in\L_Y$. Then for any $\A(p_0,p_{1A},q_{1A},p_{1B},q_{1B},p_2,q_2)$ in Table \ref{tab:algos}, the following are true:
\begin{align}
\Pr[\bar i_1]&\le 1-p_X\label{bound1:pi1},\\ 
\Pr[\bar i_2]&\le (1+\eta)(1-q_Y)\label{bound1:pi2},\\
\Pr[\bar i_1 \bar i_2]&\le (1+\eta)(1-p_X)(1-q_Y)\label{bound1:pi1i2}.
\end{align}
\end{lemma}
\begin{proof} Consider $i_1$. Suppose $i_1\in\C_X$. If $X\in\{0,1_A,1_B\}$, we have $\Pr[i_1]\ge p_X$ (by lines 1, 2, and 3 of Algorithm \ref{algo:A}). Otherwise $X=2$. 
If $p_2 = 0$ or $p_2 = 1$, line 5 of Algorithm \ref{algo:A} is executed and $\Pr[i_1] = p_2$ exactly.
Else, we run {\sc Round2Stars}. If $i_1$ is part of a large star, then it is always opened so $\Pr[\bar i_1]=0$.  Else, $i_1$ is in a small star, and we have $\Pr[\bar{i_1}] \le \Pr[X_{i_1}=1]\le\E[X_{i_1}] = q_2 = 1 - p_2$. This holds because $\bar i_1$ only occurs when $X_{i_1}=1$.
In all cases (\ref{bound1:pi1}) holds.

Consider $i_2$. Suppose $i_2\in\L_Y$. If $Y\in\{1_A,1_B\}$, we have $\Pr[i_2]\ge q_Y$. Otherwise $Y=2$. Again, if line 5 of Algorithm \ref{algo:A} is executed, $\Pr[i_2] \geq q_2$. 
Else, we run {\sc Round2Stars}. Consider the case that $i_2\in\L'_2$ is part of a large star. Recall that large stars have at least $1/(p_2\eta) = 1/((1-q_2)\eta)$ leaves. Then
\begin{align*}
\Pr[i_2]&\ge \frac{q_2(|\L'_2|-|\C'_2|)}{|\L'_2|}
\ge q_2-\frac{|\C'_2|}{|\L'_2|}
\ge q_2-(1-q_2)\eta=1-(1-q_2)(1+\eta).
\end{align*}
Otherwise, $i_2$ is part of some small star, with center $i_3$. If $X_{i_3}$ is 1 or 0, by line 5, $\Pr[i_2]=X_{i_3}$. If $0<X_{i_3}<1$, then by line 6, $\Pr[i_2]\ge X_{i_3}$.
So in any case, we have $\Pr[i_2|X_{i_3}=x]\ge x$. Note that each indicator returned by {\sc DepRound} can only take finitely many values in $[0,1]$. Letting $\U$ be the set of these values, we have
\begin{align*}\Pr[i_2]
	&=\sum_{x\in\U} \Pr[i_2|X_{i_3}=x]\Pr[X_{i_3}=x]\ge\sum_{x\in\U} x\Pr[X_{i_3}=x]
	= \E[X_{i_3}]=q_2.\end{align*}
In all cases (\ref{bound1:pi2}) holds.

Now consider both $i_1$ and $i_2$. There are many cases to consider, but most of them are easy. If $i_1$ and $i_2$ belong to stars of different classes, then they are opened independently, so $\Pr[\bar i_1 \bar i_2]=\Pr[\bar i_1]\Pr[\bar i_2] \leq (1+\eta)(1-p_X)(1-q_Y)$. For the remaining cases, $i_1\in\C_X$ and $i_2\in \L_X$ for the same class $X\in\{1_A,1_B,2\}$. There is a special case where $X=1_A$, and we are running $\A_8$. In this case, $p_{1A}=q_{1A}=0$ so $\Pr[\bar i_1 \bar i_2]=1=(1-p_{1A})(1-q_{1A})$. Otherwise, if $X\in\{1_A,1_B\}$, then at least one of $q_X$ and $p_X$ is 1, so all centers or leaves are opened, so $\Pr[\bar i_1 \bar i_2]=0$.

The remaining case is when $X=2$. Notice that line 5 of Algorithm \ref{algo:A} is called when either $p_2 = 0$ or $p_2 = 1$. The only time $p_2 = 0$ is $\A_8$, in which $q_2=1$, so all the leaves are opened and $\Pr[\bar i_1 \bar i_2] = 0$. If $p_2 = 1$, then $i_1$ is always opened and $\Pr[\bar i_1\bar i_2] = 0$. 
Otherwise $\T_2$ is divided into one group of large stars, and many groups $\G_s$ of small stars. Again, if $i_1$ and $i_2$ are in different groups, they are rounded independently. If they are both in a large star, then $i_1$ will always be opened and $\Pr[\bar i_1 \bar i_2]=0$. If they are both in the same small star, then the center-or-leaves property of our algorithm implies they will never both be closed, so $\Pr[\bar i_1 \bar i_2]=0$. 

In the only remaining case, we have that {\sc Round2Stars} is run (and $p_2+q_2=1$), and $i_1$ and $i_2$ lie in separate stars within the same group $\G_s$. Let $\mathcal{E}$ be the event that ``$i_1$ is among the $c$ random facilities chosen to be opened in line \ref{line:addc} of {\sc Round2Stars}''. We first show
\begin{align*}
	\Pr[X_{i_1} = 1 \wedge \bar{i_2} ] &\leq \sum_{x_{i_1} \in \U} x_{i_1} \Pr[X_{i_1} = x_{i_1} \wedge \bar{i_2}] \\
		&= \sum_{x_{i_1} \in \U} x_{i_1} \sum_{x_{i_3} \in \U} \Pr[X_{i_1} = x_{i_1} \wedge \bar{i_2} \wedge X_{i_3} = x_{i_3}]   \\
		&= \sum_{x_{i_1} \in \U} x_{i_1} \sum_{x_{i_3} \in \U} \Pr[\bar{i_2} | X_{i_1} = x_{i_1} \wedge X_{i_3} = x_{i_3}] \Pr[X_{i_1} = x_{i_1} \wedge X_{i_3} = x_{i_3}]  \\
		&\leq \sum_{x_{i_1} \in \U} x_{i_1} \sum_{x_{i_3} \in \U} (1 - x_{i_3}) \Pr[X_{i_1} = x_{i_1} \wedge X_{i_3} = x_{i_3}]  \\
		&= \E[X_{i_1}(1-X_{i_3})].
\end{align*}
If $|\G_s|\le c$, then all facilities in $\G_s$ will be opened and $\Pr[\bar{i_1} \bar{i_2}]=0$. Otherwise, we can bound $\Pr[\bar{i_1} \bar{i_2}]$ as follows. Conditioned on $\bar{\mathcal{E}}$, $i_1$ is closed iff $X_{i_1} = 1$. Thus,
\begin{align*}
	\Pr[\bar{i_1} \bar{i_2}] &= \Pr[\mathcal{E}] \Pr[\bar{i_1} \bar{i_2}|\mathcal{E}]  + (1-\Pr[\mathcal{E}])\Pr[\bar{i_1} \bar{i_2}|\bar{\mathcal{E}}]  \\
		&= \frac{c}{|\G_s|}\cdot 0  + \left(1-\frac{c}{|\G_s|} \right) \Pr[\bar{i_1} \bar{i_2}|\bar{\mathcal{E}}] \\
		&=\left(1-\frac{c}{|\G_s|} \right) \Pr[X_{i_1} = 1 \wedge \bar{i_2} | \bar{\mathcal{E}} ] \\
		&=\left(1-\frac{c}{|\G_s|} \right) \Pr[X_{i_1} = 1 \wedge \bar{i_2} ] \\
		&\leq \left(1-\frac{c}{|\G_s|} \right)\E[X_{i_1}(1-X_{i_3})] \\
		&\leq \left(1-\frac{c}{|\G_s|} \right) \left(1+\frac{16}{3|\G_s|\beta^2}\right) (1-p_2)(1-q_2),
\end{align*}
where we have applied Theorem \ref{cor:dep-uniform-p} from Section \ref{sec:depround}. There are $t=2$ variables of interest, $n=|\G_s|$ total variables, and $\alpha = \min\{q_2, 1 - q_2\} = \beta$.

We want to choose $c$ such that $\left(1-\frac{c}{|\G_s|} \right) \left(1+\frac{16}{3|\G_s|\beta^2}\right) \leq 1 + \eta$, or equivalently,  
\begin{align*}
	c \geq \frac{16/(3\beta^2) - \eta |\G_s|}{1 + 16/(3|\G_s|\beta^2)}.
\end{align*}
Therefore, our choice of $c = \lceil 16/(3\beta^2) \rceil$ implies that (\ref{bound1:pi1i2}) holds true in all cases.

\end{proof}

\subsubsection{ The nonlinear factor-revealing program  }
\newcommand{\sd}{D^Z}
Now we will construct a nonlinear program which bounds the ratio between the total connection cost and the cost of the bi-point solution. We first introduce some necessary notation. Partition the clients into classes according to the types of stars in which $i_1(j)$ and $i_2(j)$ lie:
\begin{align*}
\J^{(X,Y)}&:= \{j\in\J\mid i_1(j)\in \C_X \land\, i_2(j)\in \L_Y \} 
&\forall X\in\{0,1_A,1_B,2\}, Y\in\{1_A,1_B,2\}.
\end{align*}
Furthermore, since we have multiple cost bounds available, we want to use the one which will be smallest for each client. Simply put, we want to try connecting the client to the closest facility first. To this end, we define subclasses for clients who are closer to either $i_1(j)$ or $i_2(j)$, respectively:
\begin{align}
\J^{P(X,Y)}& := \{j\in\J^{(X,Y)}\mid d_2(j)\le d_1(j)\} \label{def:class-first} \\
\J^{N(X,Y)}& := \{j\in\J^{(X,Y)}\mid d_1(j) < d_2(j)\}.
\end{align}
For $(X,Y)=(1_B,2)$, we define the subclasses slightly differently. This takes into account whether each client is closer to $i_0(j)$ or $i_3(j)$:
\begin{align}
\J^{P(1_B,2)}&:= \{j\in\J^{(1_B,2)}\mid d_2\le d_1\land d_1+2d_2\le d_1+g(d_1+d_2)\}\\
\J^{P'(1_B,2)}&:= \{j\in\J^{(1_B,2)}\mid d_2\le d_1\land d_1+g(d_1+d_2)< d_1+2d_2\}\\
\J^{N(1_B,2)}&:= \{j\in\J^{(1_B,2)}\mid d_1< d_2\land d_1+2d_2\le d_1+g(d_1+d_2)\}\\
\J^{N'(1_B,2)}&:= \{j\in\J^{(1_B,2)}\mid d_1< d_2\land d_1+g(d_1+d_2)< d_1+2d_2\}. \label{def:class-last}
\end{align}
Define the following set of classes, observing $\{\J^Z\}_{Z\in\Z}$ fully partitions the set of clients.
\[ \Z=
\{P'(1_B,2),N'(1_B,2)\}\cup
\bigcup_{\substack{W\in\{P,N\}\\X\in\{0,1_A,1_B,2\}\\Y\in\{1_A,1_B,2\}}} 
\{W(X,Y)\}.\]
For each client class $Z\in\Z$, let $\sd_1:=\sum_{j\in\J^Z}d_1(j)$ and $\sd_2:=\sum_{j\in\J^Z}d_2(j)$, be the total cost contribution to $D_1$ or $D_2$, respectively, from clients in class $J^Z$. Then define the following:
\begin{align*}
C_{213}^Z&:=\sd_2+(1-q_Y)(\sd_1-\sd_2)+2(1-p_X)(1-q_Y)\sd_2\\
C_{123}^Z&:=\sd_1+(1-p_X)(\sd_2-\sd_1)+(1-p_X)(1-q_Y)(\sd_1+\sd_2)\\
C_{210}^Z&:=\sd_2+(1-q_Y)(\sd_1-\sd_2)+(1-p_X)(1-q_Y)g(\sd_1+\sd_2)\\
C_{120}^Z&:=\sd_1+(1-p_X)(\sd_2-\sd_1)+(1-p_X)(1-q_Y)g(\sd_1-\sd_2+g(\sd_1+\sd_2))\\
C_{145}^Z&:=\sd_1+(1-p_X)\left(2\sd_2+\frac1g(\sd_1+\sd_2)\right)
		+(1-p_X)(1-q_2)\frac1g(\sd_1+\sd_2).
\end{align*}
Finally, given an algorithm $\A_i=A(p_0,p_{1A},q_{1A},p_{1B},q_{1B},p_2,q_2)$, define
\begin{align}
COST_1(\A_i):=&C_{210}^{P'(1_B,2)}+C_{120}^{N'(1_B,2)}
+\sum_{\substack{X\in\{0,1_A,1_B,2\}\\Y\in\{1_A,1_B,2\}}} \left(C_{213}^{P(X,Y)}+C_{123}^{N(X,Y)}\right),\\
COST_2(\A_i):=&C_{210}^{P'(1_B,2)}+C_{120}^{N'(1_B,2)}
+\sum_{\substack{X\in\{0,1_A,1_B,2\}\\Y\in\{1_B,2\}}} \left(C_{213}^{P(X,Y)}+C_{123}^{N(X,Y)}\right)
+\sum_{X\in\{0,1_A,1_B,2\}}C_{145}^{(X,1_A)}
.\end{align}

\begin{lemma}\label{lemma:algcost}
For algorithms $\A_1,\ldots,\A_7$ and $\A_9$, the total expected cost is bounded above by $(1+\eta)COST_1(\A_i)$. The  expected cost of $\A_8$ is bounded above by $(1+\eta)COST_2(\A_8)$.
\end{lemma}
\begin{proof}
Sum the bounds from Lemmas \ref{lemma:123}, \ref{lemma:145}, and \ref{lemma:120} over each corresponding client class, and apply the bounds from Lemma \ref{lemma:probbound}. To apply those upper bounds, we need that the coefficients of $\Pr[\bar i_1]$, $\Pr[\bar i_1 \bar i_2]$ (or similar terms) are nonnegative. This follows by definition of the class being summed over. (For example, for class $P(X,Y)$, we have $d_2\le d_1$, so $d_1-d_2\ge 0$.) By linearity of expectation, we get the total expected cost of the algorithm.
\end{proof}

\begin{align}
\text{\textbf{Our NLP:} max  }\quad & X \\
\text{s.t  } \quad
& X \le COST_1(\A_i) &&\forall i\in\{1,2,3,4,5,6,7,9\} \label{constr:cost17}\\
& X \le COST_2(\A_8)\label{constr:cost8}\\
& D_2^Z\le D_1^Z && \forall Z=P(X,Y)\text{ or }Z=P'(1_B,2)\label{constr:P}\\
& D_2^Z\ge D_1^Z && \forall Z=N(X,Y)\text{ or }Z=N'(1_B,2)\label{constr:N}\\
& (2-g)D_2^{W(1_B,2)} \le gD_1^{W(1_B,2)}&&\forall W\in\{P,N\} \label{constr:1bN}\\
& (2-g)D_2^{W(1_B,2)} \ge gD_1^{W(1_B,2)}&&\forall W\in\{P',N'\} \label{constr:1bP}\\
& \sum_{Z\in\Z} D_1^Z = \frac{1}{1-b+br_D}\label{constr:sumd1}\\
& \sum_{Z\in\Z} D_2^Z = \frac{r_D}{1-b+br_D}\label{constr:sumd2}\\
& 0.508 \le b \le 3/4\\
& 19/40\le r_D \le 2/3\\
&5/6\le s_0\le 1\nonumber\\
& g\ge 0\nonumber\\
& D_1^Z,D_2^Z\ge 0 && \forall Z\in\Z \nonumber
\end{align}

\begin{lemma}
Given a bi-point solution with cost $aD_1+bD_2$ as input, with $s_0\ge5/6,b\in[0.508,3/4],r_D\in[19/40,2/3],$ and $r_1>1$, the best solution returned by $\A_1,\ldots, \A_9$ has expected cost $E[COST]\le X^*\cdot(1+\eta)(aD_1+bD_2)$, where $X^*$ is the solution to the above nonlinear program. Furthermore, $X^* \in  [1.3370, ~1.3371]$.
\end{lemma}
\begin{proof}
Given a bi-point instance $a\F_1+b\F_2$, first normalize all the distances by dividing by $aD_1+bD_2$. This does not change the solution or the ratio of approximation obtained. Let $X$ be the cost of the solution given by Algorithm \ref{algo:A}. Because of the normalization, $X$ is also the bi-point rounding factor. Constraints (\ref{constr:cost17}) and (\ref{constr:cost8}) must hold because we take the best cost of all algorithms. Lemma \ref{lemma:algcost} shows that $X$ may be a factor $(1+\eta)$ larger. Constraints (\ref{constr:P}), (\ref{constr:N}), (\ref{constr:1bN}), and (\ref{constr:1bP}) must hold by definition of each client class (see (\ref{def:class-first}) through (\ref{def:class-last})). Constraints (\ref{constr:sumd1}) and (\ref{constr:sumd2}) enforce that the corresponding distance contributions from each client class sum to $D_1$ and $D_2$ (normalized). 

We observe that for a fixed set of values of $b$, $r_D$, $s_0$, and $g$, the program becomes linear. As described in Appendix \ref{apdx:interval}, we exploit this with computer-assisted methods (rigorous interval-arithmetic) and prove that $1.3370 \le X^* \le 1.3371$.
\end{proof}

\subsection{Algorithms for edge cases}
We have several border cases which we handle in a different, generally simpler, manner.  The algorithms and proofs are given in the appendix.
\begin{lemma} There is a $(1+\eta) \cdot 1.3371  $-approximation algorithm for rounding the bi-point solution and opens at most $k+O(\log(1/\eta))$ facilities when either $b \leq 0.508, \ b \geq 3/4, \ r_D \leq 19/40,$ or $r_D \geq 2/3$. \label{lem:lisven}
\end{lemma}

\begin{lemma} There is a $(1+\eta) \cdot 1.3371  $-approximation algorithm for rounding the bi-point solution which opens at most $k+O(\log(1/\eta))$ facilities when $s_0 \leq 5/6,  b \in [0.508,3/4],$ and $ r_D \in [19/40, 2/3]$.
 \label{lem:s0}
\end{lemma}

\begin{lemma} There is a $(1+\eta) \cdot 1.3371  $-approximation algorithm for rounding the bi-point solution which opens at most $k+O(\log(1/\eta))$ facilities when $s_0 \geq 5/6,  b \in [0.508,3/4], r_D \in [19/40, 2/3],$ and $ r_1 \leq 1$. \label{lem:r1}
\end{lemma}

The result is summarized in the following theorem.
\begin{theorem}
\label{thm:final-thm}
There is a $(1+\eta) \cdot 1.3371$-approximation algorithm for rounding the bi-point solution which opens at most $k+O(\log(1/\eta))$ facilties. 
\end{theorem}

\subsection{Dichotomy result}
In the last subsections, we introduced a $(2.675 + \epsilon)$-approximation algorithm for the $k$-median problem which runs in $O\left(n^{O((1/\epsilon)\log(1/\epsilon))}\right)$ time. Now we show that by using a simple scaling technique and careful analysis, we can either improve the runtime by getting rid of the $\log(1/\epsilon)$ factor in the power of $n$, or we can improve the approximation ratio. Our result is summarized in the following theorem. Recall from the last subsection that, when {\sc Round2Stars($p_2,q_2$)} is called, $\beta := \min\{q_2,1-q_2\}$ is strictly bounded away from zero.

\begin{theorem}
For any parameter $\epsilon > 0$ small enough, there exist algorithms $A_\epsilon$ and $B_\epsilon$ such that, for any instance $\I$ of the $k$-median problem, either $A_\epsilon$ is fast or $B_\epsilon$ is more accurate:
\begin{itemize}
	\item $A_\epsilon$ is a randomized $(2.675 + \epsilon)$-approximation algorithm which produces a solution to $\I$ with constant probability and runs in $O(n^{O(1/\epsilon)})$ time, or
	\item $B_\epsilon$ is a $(2+\epsilon)$-approximation algorithm for $\I$ which runs in $O(n^{O(\poly(1/\epsilon))})$ time.
\end{itemize}

\end{theorem}

We say that a star $\S_i$ with $i \in \C_2$ is small if $2 \leq |\S_i| \leq \frac{c_0}{\eta} $ for some constant $c_0 > 0$. Otherwise, $|\S_i| > \frac{c_0}{\eta}$ and we call it a large star. Again, let $\C_2'$ and $\L_2'$ denote sets of centers and leaves of large stars. Also let $\C_2''$ and $\L_2''$ be sets of centers and leaves of small stars. 

First, observe that for large stars, we can reuse the following trick: move a little mass from the leaves to open the center. In other words, we will open $C_2'$ and a subset of size $\lceil q_2(|\L_2'| - |\C_2'|) \rceil$ of $\L_2'$. For $i_2 \in \L_2'$, it is not hard to show that $\Pr[i_2] \geq q - p\eta$ (i.e. the loss is negligible). We open 1 extra facility in this class. Recall that $\A$ opens at most 4 extra facilities in $\T_0 \cup \T_1 \cup \C_2' \cup \L_2'$.  The question is can we also reduce the number of extra opened facilities which are part of small stars (previously, this number was $O(\log(1/\eta))$)? We consider the following cases.
 
\begin{itemize}	
	\item Case 1: $|\C_2''| > f(1/\eta)$ for some  function $f = O(\poly(1/\eta))$ to be determined. In this case, we have a lot of small stars. We scale down the probability of opening the leaves by $(1-\eta)$ and open/close the centers/leaves independently. That is, for each center $i \in \C_2''$, we randomly open $\S_i$ and close center $i$ with probability $(1-\eta)q_2$. (With the remaining probability, we close $\S_i$ and open center $i$.) We show that, with constant probability, the algorithm returns a feasible solution whose cost is only blown up by a small factor of $(1+\eta)$. 
	
	\item Case 2: $|\L_2'| + |\L_2''| \leq g(1/\eta)$ for some function $g = O(\poly(1/\eta))$ to be determined. In this case, the number of leaves should be small enough so that we can simply open all the leaves in $\L_2$. The number of extra opened facilities is $O(g(1/\eta))$. However, we achieve a pseudo solution with no loss in connection cost compared to the bipoint solution.
	
	\item Case 3: Neither Case 1 nor Case 2 holds (i.e. $|\C_2''| \leq f(1/\eta) $ and $|\L_2'| + |\L_2''| \geq g(1/\eta)$). Note that, by definition of small stars, 
	$$|\L_2''| \leq \frac{c_0}{\eta}|\C_2''| \leq \frac{c_0 f(1/\eta)}{\eta}.$$
This implies that
$$ |\L_2'| \geq g(1/\eta) - |\L_2''| \geq g(1/\eta) - \frac{c_0 f(1/\eta)}{\eta}.$$
Intuitively, the number of centers and leaves of small stars are upper-bounded by some constant. On the other hand, we have a lower-bound on the number of leaves of large stars. If we have enough leaves in $\L_2'$, we can scale down the probability to open each facility in $\L_2'$ so that all centers in $\C_2''$ can be opened without violation.
\end{itemize}

See Appendix \ref{sec:dichotomy} for details of these cases.

\section{Discussion}
We conclude with a specific discussion followed by more general speculation.

In Section~\ref{sec:bipoint-round}, we considered a selection of
counterbalancing algorithms which were chosen (with numerical aid) to be a
minimal such set which obtains the bi-point rounding factor 1.3371.
However, this can be improved, if only slightly, at the cost of adding
more nonlinear variables to the factor-revealing program. We currently
split the 1-stars into two groups based on their size-to-distance ratio
$g_i$, and a threshold $g$. We assume this ratio may be unbounded on either
side of the threshold, yet the analysis is only tight when all $g_i$ are
exactly $g$. We could exploit this by splitting the 1-stars into 3 or more
classes, with multiple thresholds, and adding more sets of parameters to
take advantage of the division. Testing this with three classes, we get a
new factor in the interval $[1.332,1.3371)$. So we know there is a little
more gain to be had, but it adds more complexity to the algorithm and
analysis.

Also, consider that in our algorithm we have fixed the parameter $r_{1A}$
so that it is exactly large enough to close and open all the big leaves ($\A_8$). This is
a strategic choice, and makes the algorithms simple. However, it is
possible that there is a better choice, as a function of some other
variables in the instance. It is also possible to fix $g$ instead and let
$r_{1A}$ be a variable in the program, but this creates more cases. A
purely analytical analysis would be greatly helpful toward choosing the
appropriate parameter.

A rough lower bound on the potential improvement from these ideas is
$\frac{3+\sqrt{5}}{4}\approx 1.309$, as this is the ratio we get on an instance
with no 1-stars (or 0-stars) at all, by opening roots and leaves with
proportions $(a,b)$ and $(1,b/2)$.

Recent years have seen significant progress on hard-capacitated problems, e.g., for vertex-cover and its variants \cite{DBLP:journals/siamcomp/ChuzhoyN06,DBLP:journals/jcss/GandhiHKKS06,DBLP:conf/soda/CheungGW14}. However, progress on the different variants of capacitated problems has been slower: see, e.g., \cite{an-etal:cap-k-center,jain_vazirani,DBLP:journals/mp/LeviSS12} and the references therein. We suggest speculatively that the (probabilistic) analog of \cite{srin:level-sets} for bipartite graphs -- the work of \cite{DBLP:journals/jacm/GandhiKPS06} -- may help with ensuring that the capacity constraints are met with probability one, while ensuring other desired negative-correlation and near-independence properties. 

\section{Acknowledgements}
We thank Joachim Spoerhase for helping us to discover an inconsistency in an earlier version \cite{ByrkaPRST15} of this paper. We also thank Marcin Mucha for showing us the work of Mahdian \& Yan \cite{mahdian}.

\bibliographystyle{abbrv}
\bibliography{./kmed_268}

\appendix

\medskip
\begin{center}
\textbf{\Large Appendix}
\end{center}

%\section{Proofs: JMS'($\gamma$)}
%\label{apdx_jms}
%\input{apdx_jms_prime.tex}

\section{Proofs for Section 3: {\sc DepRound}}
\label{apdx:depround}
\begin{proof}(Lemma \ref{lemma:b-props}) As example, we prove the properties hold in case I:
\begin{itemize}
\item[\textbf{(B0)}]We either set $\gamma_1=0$, and $\gamma_2=\beta_2+\beta_1\frac{a_1}{a_2}=\frac{1}{a_2}(a_2\beta_2+a_1\beta_1)\le\frac{\min\{a_1,a_2\}}{a_2}\le 1$, or we set $\gamma_2=0$ and $\gamma_1=\beta_1+\beta_2\frac{a_2}{a_1}=\frac{1}{a_1}(a_1\beta_1+a_2\beta_2)\le\frac{\min\{a_1,a_2\}}{a_1}\le1$.
\item[\textbf{(B1)}] $\E[\gamma_1]=\!\frac{a_2\beta_2}{a_1\beta_1+a_2\beta_2}\cdot0+\frac{a_1 \beta_1}{a_1\beta_1+a_2\beta_2} (\beta_1+\beta_2\frac{a_2}{a_1})=\beta_1,$
and $\E[\gamma_2]=\!\frac{a_2\beta_2}{a_1\beta_1+a_2\beta_2}(\beta_2+\beta_1\frac{a_1}{a_2})+\frac{a_1\beta_1}{a_1\beta_1+a_2\beta_2}\cdot0=\beta_2$.
\item[\textbf{(B2)}] If we set $\gamma_1=0$, then  $a_1\gamma_1+a_2\gamma_2=0+a_2(\beta_2+\beta_1\frac{a_1}{a_2})=a_1\beta_1+a_2\beta_2$. If we set $\gamma_2=0$, then $a_1\gamma_1+a_2\gamma_2=a_1(\beta_1+\beta_2\frac{a_2}{a_1})+0=a_1\beta_1+a_2\beta_2$.
\item[\textbf{(B3)}] The first part is trivial as $\E[\gamma_1\gamma_2]=0\le\beta_1\beta_2$. For the second part:
\begin{align*}\E[(1-\gamma_1)(1-\gamma_2)]&=\frac{a_2\beta_2}{a_1\beta_1+a_2\beta_2}\cdot1\cdot(1-\beta_2-\beta_1\frac{a_1}{a_2})
+\frac{a_1\beta_1}{a_1\beta_1+a_2\beta_2}\cdot(1-\beta_1-\beta_2\frac{a_2}{a_1})\cdot1
\\&=\frac{a_2\beta_2(1-\beta_2)-a_1\beta_1\beta_2+a_1\beta_1(1-\beta_1)-a_2\beta_1\beta_2}{a1\beta_1+a_2\beta_2}
\\&=(1-\beta_1-\beta_2)\le(1-\beta_1)(1-\beta_2).
\end{align*}
\end{itemize}
\end{proof}

\section{Proofs: Bounding the number of opened facilities }
\label{apdx:nfacilities}

\begin{proof} (Claim \ref{claim:bound_beta})
{\sc Round2Stars} is only called during $\A_1, \ldots, \A_6$. (In $A_7$ and $A_8$, line 5 is called instead.) Consider possible values of $p_2$ and $q_2$ in Table \ref{tab:algos}. Recall that $s_0 \in [5/6,1]$, $b \in [0.508,3/4]$ and $a+b = 1$. The minimum of $\beta = \min\{q_2, 1-q_2\}$ is attained in $\A_6$ when $b = 0.508$ and $s_0 = 5/6$; here $q_2 = (b-a)s_0=(0.508-0.492)\cdot 5/6=1/75$. Also, the minimum of $p_2$ is attained at $p_2 = as_0$, $a = 1/4$, and $s_0 = 5/6$.
\end{proof}

\begin{proof} (Lemma \ref{lem:singleA})
We consider Algorithm \ref{algo:A}. It is easy to see that
\begin{itemize}
	\item In line 1, we open at most $p_0|\C_0|+1$ facilities,
	\item In line 2, we open at most $p_{1A}|\C_{1A}| + q_{1A}|\C_{1A}|+2$ facilities,
	\item In line 3, we open at most $p_{1B}|\C_{1B}| +  q_{1B}|\C_{1B}|+2$ facilities,
	\item If line 5 is executed then we open at most $p_2|\C_2| + q_2|\L_2| + 1$ facilities,
	\item Otherwise, {\sc Round2Stars} is called:
	\begin{itemize}
		\item[$\circ$] In line 1, the number of opened facilities is
\begin{align*}
	|\C_2'| + \lceil q_2(|\L_2'|-|\C_2'|) \rceil  &\leq 1 + |\C_2'| + q_2(|\L_2'|-|\C_2'|)
  				= p_2|\C_2'| + q_2|\L_2'| + 1,
\end{align*}
where the equality follows due to the fact that $1 - q_2 = p_2$ whenever {\sc Round2Stars} is called.
		\item[$\circ$] By Lemma \ref{lem:groupG} and Claim \ref{claim:bound_beta}, the number of facilities opened by the for loop (lines $3\ldots 7$) is at most
	\begin{align*}
		\sum_{s=1}^{\lceil \log_{1+\beta} ( 1 / (p_2 \eta)) \rceil - 1}\!\!\!\!\!\!\!\! (p_2|\C(s)| + q_2|\L(s)| + c + 2) 
	= \sum_{s=1}^{\lceil \log_{1+\beta} ( 1 / (p_2 \eta)) \rceil - 1}\!\!\!\!\!\!\!\! (p_2|\C(s)| + q_2|\L(s)|) + O(\log (1/\eta)) .
	\end{align*}

	\end{itemize}
\end{itemize}
The lemma follows by taking the sum of opened facilities in each case.

\end{proof}

\begin{proof} (Lemma \ref{lem:kfacilities}) Since $b \in [1/2,3/4]$ and $s_0 \leq 1$, $p_2$ is bounded away from $0$ in $\A_1, \ldots, \A_9$. Note that {\sc Round2Stars} is not called in $\A_7$ and $\A_8$; at most $E+1$ facilities can be opened in these two algorithms. By Lemma \ref{lem:singleA}, it suffices to show that $E \leq k + 1$. The proof is straightforward. We substitute the parameters in Table \ref{tab:algos} and $s_0 = \frac{1}{1+|\C_0|/\Delta_F} = 1 + \frac{|\C_0| }{ |\C_2| - |\L_2| }$ to compute $E$ in each case. We use simple facts such as $|\C_{1A}| = |\L_{1A}|$, $|\C_{1B}| = |\L_{1B}|$, and $|\C_{1A}|+|\C_{1B}| = |\C_1|$ to further simplify the expression. Also recall that $|\C_{1A}| = \lceil a\Delta_F \rceil$, and thus $a\Delta_F \leq |\C_{1A}| \leq a\Delta_F+1$.
By definition, we have $\Delta_F=|\L_2|-|\C_0|-|\C_2|$ and $2|\C_2| \leq |\L_2|$.
\begin{itemize}
\item For $\A_1$, we have
\begin{align*}
	E &= |\C_{1A}| + |\C_{1B}| + as_0|\C_2| + (1-as_0) |\L_2 |   \\
		&= 	|\C_1| + a|\C_0| + a|\C_2| + b|\L_2|    \\
		&= a(|\C_0| + |\C_1| + |\C_2|) + b(|\C_1| + |\L_2|) = a|\F_1| + b|\F_2| = k.
\end{align*} 
\item For $\A_2, \A_3, \A_4$, and $\A_5$,  substituting the parameters gives the same $E$:
\begin{align*} 
	E &= |\C_0| +  |\C_1| +  (1-bs_0)|\C_2| +bs_0|\L_2| \\ 
	&=|\C_0| - b |\C_0| + |\C_1| + |\C_2| - b |\C_2| + b |\L_2|\\ &= a|\F_1| + b|\F_2| = k.
\end{align*}

\item For $\A_6$, we have
\begin{align*} 
	E &= |\C_0| + |\C_{1A}| + |\C_{1A}| + |\C_{1B}| + (1-(b-a)s_0)|\C_2| + (b-a)s_0|\L_2| \\
	&= |\C_0| + |\C_{1A}| + |\C_1| + (1-(b-a)s_0)|\C_2| + (b-a)s_0|\L_2| \\
	&\leq 1 + |\C_0| + a\Delta_F + |\C_1| + (1-(b-a)s_0)|\C_2| + (b-a)s_0|\L_2| \\
	&=1 + |\C_0| - b |\C_0| + |\C_1| + |\C_2| - b |\C_2| + b |\L_2| \\
	&= 1 + a|\F_1| + b|\F_2| = k + 1.
\end{align*}

\item For $\A_7$, we have
\begin{align*} 
	E &= |\C_0| + |\C_1| + |\C_2| + \frac{b}{2(1/s_0)}|\L_2| \\
	  &\leq |\C_0| + |\C_1| + |\C_2| + \frac{b}{1/s_0 + r_2}|\L_2| \\
	  &=|\C_0| - b |\C_0| + |\C_1| + |\C_2| - b |\C_2| + b |\L_2|\\ &= a|\F_1| + b|\F_2| = k.
\end{align*}

\item For $\A_8$, we have
\begin{align*} 
	 E &= |\C_{1B}| + |\L_2| \\
	  &=  |\C_1| - |\C_{1A}| + |\L_2| \\
	  &\leq  |\C_1| - a\Delta_F + |\L_2| \\
	  &\leq  |\C_1| - \Delta_F + b\Delta_F + |\L_2| \\
	  &=  |\F_1| + b\Delta_F = k.
\end{align*}

\item For $\A_9$ (this is exactly Li-Svensson algorithm), we have
\begin{align*} 
	 E &= a|\C_0| + a|\C_1| + b|\L_1| + a|\C_2| + b|\L_2| \\
	   &= a|\F_1| + b|\F_2| = k.
\end{align*}

\end{itemize}

\end{proof}

\section{Proofs: Bounding client connection cost}
\label{sec:final-cost-proofs}
\begin{proof}(Lemma \ref{lemma:123})
For all these clients, we know that at least one of $i_2$ or $i_3$ will always be open. First consider the case that $i_1\ne i_3$. 
 Then the facilities are as shown below. 
Observe by the construction of stars, $i_2$ cannot be closer to $i_1$ than $i_3$ (otherwise $i_1$ would be its center). Thus $d(i_2,i_3)\le d(i_2,i_1)\le d_1+d_2$. It follows by the triangle inequality that $d(j,i_3)\le d(j,i_2)+d(i_2,i_3)\le d_1+2d_2$.

\begin{center}
\includegraphics{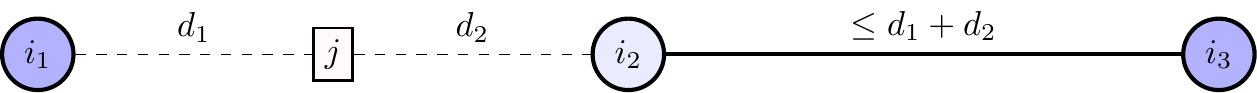}
\end{center}
Now let us connect $j$ to $i_2$ if open. Else, connect to $i_1$ if open. Else, connect to $i_3$. The actual facility which $j$ connects to can only be closer than any of these. Thus, this yields the following upper bound for the expected connection cost of $j$. 
\begin{align}
c_{213}(j)&:= Pr[i_2]d_2+Pr[i_1\bar i_2]d_1+Pr[\bar i_1\bar i_2](d_1+2d_2)\nonumber\\
&= Pr[i_2]d_2+Pr[\bar i_2]d_1+2Pr[\bar i_1\bar i_2]d_2\nonumber\\
&= d_2+Pr[\bar i_2](d_1-d_2)+2Pr[\bar i_1\bar i_2]d_2,\label{cost:213}
\end{align}
where the subscript corresponds to the order in which we try connecting to facilities. Alternatively, we may connect $j$ first to $i_1$ if open, else $i_2$ if open, else $i_3$. This gives the equally valid bound
\begin{align}
c_{123}(j)&:= Pr[i_1]d_1+Pr[\bar i_1 i_2]d_2+Pr[\bar i_1\bar i_2](d_1+2d_2)\nonumber\\
&= Pr[i_1]d_1+Pr[\bar i_1]d_2+Pr[\bar i_1\bar i_2](d_1+d_2)\nonumber\\
&= d_1+Pr[\bar i_1](d_2-d_1)+Pr[\bar i_1\bar i_2](d_1+d_2).\label{cost:123}
\end{align}

Now consider the remaining case that $i_1=i_3$. In this case, at least one of $i_1$ or $i_2$ will always be open. Again, depending on which facility we attempt to connect to first, we can obtain either of two bounds:
\begin{align*}
c_{21}(j)&:= Pr[i_2]d_2+Pr[\bar i_2]d_1\le c_{213}(j)\\
c_{12}(j)&:= Pr[i_1]d_1+Pr[\bar i_1]d_2\le c_{123}(j).
\end{align*}
Thus (\ref{cost:213}) and (\ref{cost:123}) are valid bounds in both cases.
\end{proof}

\begin{proof}(Lemma \ref{lemma:145})
For these clients it is possible that $i_1$, $i_2$ and $i_3$ are all closed. Let $i_4$ be the closest facility in $\L_2$ to $i_3$, and let $i_5$ be the center of $i_4$. Then by definition, we have
\[g_{i_3}=\frac{d(i_3,i_2)}{d(i_3,i_4)}.\]
This yields the following bound on $d(i_3,i_4)$ (and thus $d(i_4,i_5)$):
\[ d(i_4,i_5)\le d(i_3,i_4)=\frac{1}{g_{i_3}} d(i_2,i_3) 
\le \frac{1}{g}(d_1+d_2).\]

\begin{center}\includegraphics{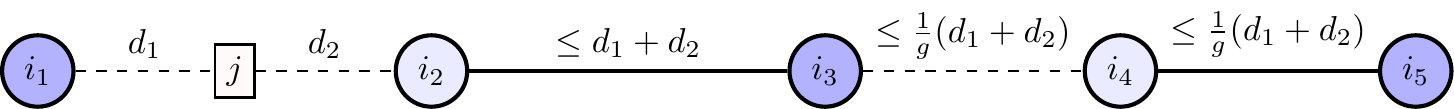}\end{center}

Now we know that if $i_4$ is closed, then $i_5$ must be open. We also know that $i_2$ and $i_3$ will always be closed. In the case that $i_1\ne i_3$ (which is shown above), we will try connecting, in order, to $i_1$, $i_4$, and $i_5$, connecting to the first one which is open. This yields the following bound:
\begin{align}
c_{145}(j):=& Pr[i_1]d_1+Pr[\bar i_1 i_4]\left(d_1+2d_2+\frac1g (d_1+d_2)\right)
+Pr[\bar i_1 \bar i_4]\left(d_1+2d_2+\frac2g(d_1+d_2)\right)\nonumber\\
=&d_1+Pr[\bar i_1]\left(2d_2+\frac1g (d_1+d_2)\right)+Pr[\bar i_1 \bar i_4]\frac1g(d_1+d_2).\label{cost:145}
\end{align}
For the case that $i_1=i_3$, we have the below situation:
\begin{center}
\includegraphics{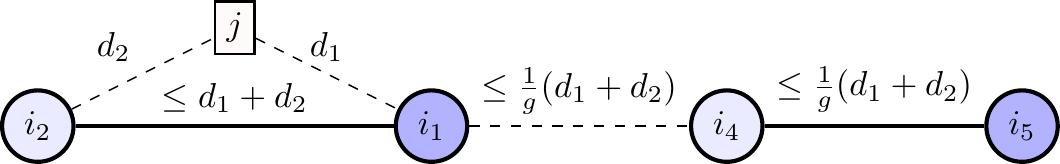}
\end{center}
Here $i_1$ and $i_2$ are always closed, so we try connecting first to $i_4$, then to $i_5$, giving the following bound:
\begin{align*}
c_{45}(j):=&d_1+\frac1g(d_1+d_2)+Pr[\bar i_4]\frac1g(d_1+d_2).
\end{align*}
In this case $Pr[\bar i_1]=1$. Also since $i_1\in\C_{1A}$ and $i_4\in\L_2$ are in different types of stars, they are rounded independently, so $Pr[\bar i_1\bar i_4]=Pr[\bar i_1]Pr[\bar i_4]=Pr[\bar i_4]$. Thus, $c_{45}(j)\le c_{145}(j)$, and the claim still holds.
\end{proof}

\begin{proof}(Lemma \ref{lemma:120})
In this case $i_1\in\C_{1B}$. Let $i_0$ be the leaf attached to $i_1$. Again, we know that if $i_0$ is closed, $i_1$ will be open. Recall that by definition $g_i = \frac{d(i, i')}{\min_{j \in \L_2} d(i, j)}$, where $i$ and $i'$ are the center and leaf, respectively of a 1-star. Applying this to $i_1$ and $i_0$, we have
\[ d(i_1,i_0) = g_{i_1}\min_{i\in\L_2}d(i_1,i)\le g
\cdot d(i_1,i_2) \le g(d_1+d_2). \]
\begin{center}
\includegraphics{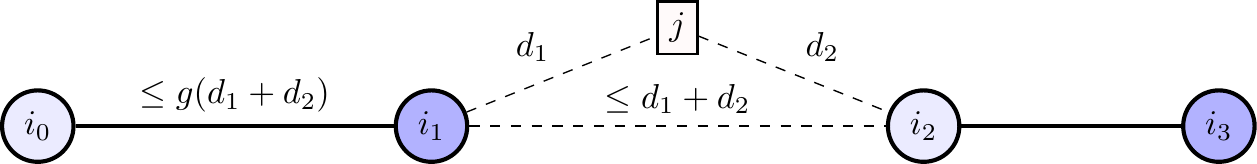}
\end{center}
Now we may try connecting in order $i_2$, $i_1$, $i_0$, or alternatively, in order $i_1$, $i_2$, $i_0$, yielding the following bounds:
\begin{align}
c_{210}(j)&:= Pr[i_2]d_2+Pr[i_1\bar i_2]d_1+Pr[\bar i_1\bar i_2](d_1+g(d_1+d_2))\nonumber\\
&=d_2+Pr[\bar i_2](d_1-d_2)+Pr[\bar i_1\bar i_2]g(d_1+d_2)\label{cost:210}\\
c_{120}(j)&:= Pr[i_1]d_1+Pr[\bar i_1 i_2]d_2+Pr[\bar i_1\bar i_2](d_1+g(d_1+d_2))\nonumber\\
&=d_1+Pr[\bar i_1](d_2-d_1)+Pr[\bar i_1\bar i_2](d_1-d_2+g(d_1+d_2)).\label{cost:120}
\end{align}
Note that by definition of $i_2(j)$, we can say $d_2=d(j,i_2)\le d(j,i_0)\le d_1+g(d_1+d_2)$, which implies $(d_1-d_2+g(d_1+d_2))\ge0$, a fact that will be used later.
\end{proof}

\section{Interval relaxation: Bounding the NLP}
\label{apdx:interval}

The non-linear program does not appear to admit a simple method of solving. We may find a local maximum, but as the system is not concave, we have no guarantee of global optimality. However, we observe that by fixing a small number of variables, the remaining system becomes linear and may be solved exactly. Using this fact together with an interval arithmetic approach \cite{zwick}, we systematically prove an upper bound of $1.3371$ on the system.

Let $V=\{X,1\}\bigcup_{Z\in\Z}\{D_1^Z,D_2^Z\}$. Then we may express each constraint $C_j$ in the following form, where $f_{x,j}$ is a function of several variables:
\[C_j:=\quad\sum_{x\in V} f_{x,j}(b,r_D,g,s_0)x \ge 0.\]

For each $(x,j)$, define the constant $c_{x,j}:=\max_{(b,r_D,g,s_0)\in I}f_{x,j}(b,r_D,g,s_0)$ to be the maximum value of each function over the interval $I$. And let 
\[C'_j:=\quad \sum_{x\in V} c_{x,j} x \ge 0.\]
Since all variables in $V$ are nonnegative, we may relax the program by replacing each constraint $C_j$ in the original program with constraint $C'_j$. The new, relaxed program is linear and may thus be solved efficiently. The solution is an upper bound on the value of the original program over interval $I$. The relaxed bound may be rather loose, since each term is maximized independent of the others. However, for sufficiently small intervals, the relaxation can approximate the original program to any desired precision. 
(For intervals in which $g$ may be 0 or infinitely large, we may get terms of the form $\frac10$ in the relaxed equations. For these cases, we remove any algorithm which has these terms from the program.)

Our implementation starts with several large intervals and calculates an upper bound with the above method. Any bound which is larger than the specified goal is divided into 16 subintervals (dividing in half for each variable), and the program is run recursively on the new intervals.

To speed up the search, we made several modifications to the program as stated. First, we only used the $P/N$ class division for a few of the client classes, using the bounds $c_{213}$ and $c_{210}$ for the remaining classes. This is a valid relaxation of the system which significantly reduces the number of variables, but does not appear to make the solution any worse. Second, we added a constraint that the cost must be less than that the simple formula given by Li and Svensson (relaxed over the interval); this is equivalent to the bound we describe in the LP for the cost of $A_9$, but gives a tighter relaxation in this simpler form. Third, when relaxing (\ref{constr:cost17}) and (\ref{constr:cost8}), we grouped terms of the form $D_1^{P(X,Y)}- D_2^{P(X,Y)}$ before relaxing the coefficients. (We know this value is positive by definition of the class.) This gives a tighter relaxation in many cases.

Using this approach we obtained an upper bound of 1.3371. The calculation was implemented in Mathematica. It examined around 8 million intervals, and took around 7 hours on an Intel Core i7 2.9GHz machine.

\subsection{Tight example}
The following is a solution that obtains $1.3370$. (There are many possible solutions, as there is some degree of freedom among some of the D-type variables.) This shows that our above bound on the nonlinear program is tight.
\begin{align*}
& &D_1^{P(1_B,1_B)}&=0.287\\
X&=1.3370&D_1^{P(1_A,2)}&=0.221\\
b&=0.645&D_1^{P(1_B,2)}&=0.847\\
r_D&=0.497&D_1^{P(2,1_A)}&=0.125\\
g&=0.646&D_2^{P(1_A,2)}&=0.221\\
s_0&=1&D_2^{P(1_B,2)}&=0.404\\
&&D_2^{P(2,1_A)}&=0.111
\end{align*}

\section{Proofs: Algorithms for edge cases}
We will need these two facts:
\begin{claim} $r_2 \leq 1/s_0 .$
\label{cl:r2}
\end{claim}
\begin{proof}
By definition of 2-stars, we have $|\C_2| \leq |\L_2| / 2$ which implies $|\C_2| \leq (|\C_1|+|\L_2|) - (|\C_0|+|\C_1|+|\C_2|) + |\C_0| = \Delta_F + |\C_0|$. Thus, $r_2 \leq 1+r_0 = 1/s_0$.
\end{proof}

\begin{claim} $\frac{|\L_2|}{\Delta_F} \leq \frac{2}{s_0} .$
\label{cl:r22}
\end{claim}
\begin{proof}
Since $|\L_2|/2 \leq  |\L_2| - |\C_2|  = \Delta_F + |\C_0|  $, we have $ \frac{|\L_2|}{2\Delta_F}  \leq 1 + r_0 = 1/s_0$.

\end{proof}

\begin{proof} (Lemma \ref{lem:lisven}.)
Basically, we just return the better solutions between $\F_1$ and the one produced by $\A' = \A(a,a,b,a,b,a,b)$. As mentioned before, $\A'$ and Li-Svensson's rounding algorithm are essentially the same, and output a solution with the same upper-bound on the expected cost:
$$ (1+\eta)(ad_1 + b(1+2a)d_2),$$
The main difference is that $\A'$ only opens at most $O(\log(1/\eta))$ (instead of $O(1/\eta)$) extra facilities.

As in \cite{li_svensson}, we need to be careful when $a$ or $b$ is close to $0$ because the number of extra  facilities opened by $\A'$ is roughly $\frac{16}{3 \beta^2}\log_{1+\beta}( 1/( a \cdot \eta) )$, where $\beta = \min\{a,b\}$.  We consider two corner cases:
\begin{itemize}
	\item If $0 \leq b \leq 1/4$, we can just return $\F_1$ as our solution. Note that $|\F_1|\leq k$ and the approximation ratio is $\frac{d_1}{ad_1+bd_2} \leq \frac{1}{a} = \frac{1}{1-b} \leq 4/3$. 
	
	\item If $b \geq 5/6$, we can use the knapsack algorithm described in \cite{li_svensson}, which only opens at most $k+2$ facilities, to get a $4/3$-approximation algorithm. Note that the approximation ratio of this algorithm is bounded by $1 + 2a \leq 4/3$ as $a \leq 1/6$. 
	
	\item We claim that the above algorithm gives an approximation ratio of $1.337$ for all remaining cases. Note that the cost of this algorithm is at most $(1+\eta)\min\{d_1, ad_1+b(1+2a)d_2 \}$. Thus, it suffices to bound the ratio
	\begin{align*}
	f := \frac{\min\{d_1, ad_1+b(1+2a)d_2 \}}{ad_1+bd_2} &= \min  \left\{ \frac{1}{1-b+br_D}, \frac{1-b + b(1+2(1-b))r_D}{1-b+br_D}  \right\} .
\end{align*}
	Note that the right-hand-side is a function of $b$ and $r_D$. When $b \in [1/4, 0.508]$, $b \in [3/4, 5/6]$, $r_D \leq 19/40$, or $r_D \geq 2/3$, that function is at most $1.337$ by elementary calculus. 
	
Observe that
\begin{align*}
	\frac{\partial}{\partial r_D} \left(\frac{1}{1-b+br_D} \right) = -\frac{b}{(1-b+br_D)^2 } \leq 0,
\end{align*}
and
\begin{align*}
	\frac{\partial}{\partial r_D} \left(\frac{1 - 2 b^2 r_D + b (-1 + 3 r_D)}{1-b+br_D} \right) = \frac{2(-1 + b)^2 b}{(1 + b (-1 + r_D))^2} \geq 0.
\end{align*}
It means that, for a fixed value of $b$, the former is a decreasing function of $r_D$ and the latter is an increasing function of $r_D$. Therefore,
\begin{itemize}
	\item Case $b \in [1/4, 0.508]$ or $b \in [3/4, 5/6]$: The maximum of $f$ will be achieved at some point such that
\begin{align*}
	\frac{1}{1-b+br_D} &= \frac{1 - 2 b^2 r_D + b (-1 + 3 r_D)}{1-b+br_D},
\end{align*}
or equivalently,
$$ r_D = \frac{1}{3-2b}.$$
Then, in this case,
	\begin{align*}
		 f \leq \max_{b \in [1/4, 0.508] \cup [3/4, 5/6], r_D = \frac{1}{3-2b}} \frac{1}{1-b+b r_D} = 1.33681.
	\end{align*}

	\item Case $b \in [0.508, 3/4]$ and $r_D \leq 19/40$: Note that $r_D \leq \frac{1}{3-2b}$ which implies that $\frac{1}{1-b+br_D} \geq \frac{1 - 2 b^2 r_D + b (-1 + 3 r_D)}{1-b+br_D}$. Since the RHS is increasing in $r_D$, we have
	$$ f \leq \max_{b \in [0.508, 3/4], r_D = 19/40} \frac{1 - 2 b^2 r_D + b (-1 + 3 r_D)}{1-b+br_D} = 1.33294.$$
	
	\item Case $b \in [0.508, 3/4]$ and $r_D \geq 2/3$: Note that $r_D \geq \frac{1}{3-2b}$ which implies that $\frac{1}{1-b+br_D} \leq \frac{1 - 2 b^2 r_D + b (-1 + 3 r_D)}{1-b+br_D}$. Since the LHS is decreasing in $r_D$, we have
		$$ f \leq \max_{b \in [0.508, 3/4], r_D = 2/3} \frac{1}{1-b+br_D} = \frac{4}{3} \leq 1.33334. $$

\end{itemize}

\end{itemize}
\end{proof}

\begin{proof} (Lemma \ref{lem:s0}.)
We show that the better solution returned from a set of 3 algorithms will be within a factor 1.3371 of the optimal solution. The purpose of this case is to bound $s_0$ away from 0, so that $p_2$ and $q_2$ in our main case are bounded away from zero. The first algorithm is a knapsack algorithm in which we open all facilities in $\L_1$ and $\C_2$. After that we almost greedily choose some of the 2-stars, close their centers, and open all their leaves. This algorithm does very well if $s_0$ is small. In the second algorithm, we open $\F_1$ and some additional facilities in $\L_2$ which maximize the saving. In particular, we use the following algorithms:
\begin{itemize}
	\item Algorithm 1: Open all facilities in $\L_1$ and $\C_2$. For each client $j$, if $i_2(j) \in \L_1$, connect $j$ to $i_2(j)$. Otherwise $i_2(j)$ is a leaf of a 2-star, let $i_3$ be the center of this star and connect $j$ to $i_3$. Thus, the total connection cost of the current solution is upper-bounded by
	$$ \sum_{j:i_2(j) \in \L_1} d_2(j) + \sum_{j:i_2(j) \in \L_2} (d_1(j) + 2d_2(j)) = D_2 + \sum_{j:i_2(j) \in \L_2} (d_1(j) + d_2(j)).$$
	Now, for each $i \in \C_2$, if we close facility $i$ and open all of its leaves, the total cost will be reduced by $\sum_{j \in \delta(\S_i)}(d_1(j) + d_2(j))$, where $\delta(\S_i)$ is the set of clients $j$ having $i_2(j) \in \S_i$, and we also open additional $|\S_i|-1$ facilities. This motivates us to solve the following knapsack LP, just as in \cite{li_svensson}:
	
  \begin{alignat*}{2}
    \text{maximize }   &  \sum_{i \in \C_2} x_i \left(\sum_{j \in \delta(\S_i)} (d_1(j) + d_2(j))  \right) \\
    \text{subject to } &  \sum_{i \in \C_2} x_i(|\S_i|-1) \leq k - |\L_1| - |\C_2|      \\
                             & 0 \leq x_i \leq 1 \ \forall i \in \C_2
  \end{alignat*}
Note that a basic solution of the above LP only has at most 1 fractional value. Thus, we can easily obtain it by a greedy method. Let us call this fractional value $x_{i^*}$.  Now, for all $i \in \C_2$, if $x_i = 0$, we keep $i$ opened. If $x_i = 1$, we close $i$ and open all of its leaves. We also open $i^*$ and a subset of size $\lceil x_{i^*} |\S_{i^*}| \rceil$ of $\S_{i^*}$ uniformly at random. It is easy to see that the expected saved cost is at least the optimal value of the LP by doing so.

	\item Algorithm 2: Open all facilities in $\F_1$. Define the saving of a facility $i \in \L_2$ be $\sum_{j\in \delta(\S_i)} (d_1(j)-d_2(j))_+$. Sort all the facilities in $\L_2$ non-increasing by its saving. Open the first $\lceil  \frac{bs_0}{2} |\L_2| \rceil$ facilities in this order.

\end{itemize}
Analysis: 
\begin{itemize}
	\item The two algorithms only open $k+2$ facilities. In the first algorithm, we claim that at most $k+2$ facilities will be opened. The first constraint of the LP guarantee that a fractional solution $x$ will open at most $k$ facilities. The two extra facilities come from the fact that we open $i^*$ and take the ceiling of $x_{i^*} |\S_{i^*}|$. In the second algorithm, we open at most
	\begin{align*}
		|\F_1| + \frac{bs_0}{2} |\L_2| + 1 &\leq  |\F_1| + \frac{b|\L_2|}{2} \times \frac{2\Delta_F}{|\L_2|}  + 1 \\
			&=  |\F_1| + b\Delta_F + 1 = k + 1,
	\end{align*}
where the first inequality is due to $s_0 \leq \frac{2\Delta_F}{|\L_2|}$, by Claim \ref{cl:r22}.

	\item Now, we bound the cost of the first algorithm. Let $q$ be the maximum value such that the solution $x_i = q$ for all $i \in C_2$ is feasible to the knapsack LP. We solve for $q$ by requiring 
	\begin{align*}
		&\sum_{i \in \C_2}q(|\S_i|-1) \leq k - |\L_1| - |\C_2| \\
		&q(|\L_2| - |\C_2|) \leq k - |\L_1| - |\C_2| \\
		&q \leq \frac{k - |\L_1| - |\C_2|}{|\L_2| - |\C_2|} = \frac{|\F_1| + b\Delta_F - |\L_1| - |\C_2|}{|\L_2| - |\C_2|} = \frac{b\Delta_F + |\C_0|}{|\L_2| - |\C_2|}.
	\end{align*}
	Thus, we can set $q := \frac{b\Delta_F + |\C_0|}{|\L_2| - |\C_2|}$. Note that $|\L_2| - |\C_2| = |\C_0| + \Delta_F$, we have
	$$ q = \frac{b\Delta_F + |\C_0|}{ |\C_0| + \Delta_F} = \frac{b + r_0}{1 + r_0} = 1 - as_0.$$
	Since $x = q$ is a feasible solution, the saved cost is at least 
$$ \sum_{i \in \C_2} q \left(\sum_{j \in \delta(\S_i)} (d_1(j) + d_2(j))  \right) = (1 - as_0)\sum_{j:i_2(j) \in \L_2}(d_1(j)+d_2(j)). $$	
Therefore, we can upper-bound the cost by
\begin{align*}
 & D_2 + \sum_{j:i_2(j) \in \L_2} (d_1(j) + d_2(j)) -   (1 - as_0)\sum_{j:i_2(j) \in \L_2}(d_1(j)+d_2(j)) \\
 & = D_2 + as_0\sum_{j:i_2(j) \in \L_2}(d_1(j)+d_2(j)).
\end{align*}
\item Recall that the sets $\delta(\S_i)$ are pairwise disjoint. By a simple average argument, the cost of the second algorithm is upper-bounded by
   	\begin{align*}
   		D_1 - \frac{bs_0}{2}\sum_{i \in \L_2}\sum_{j\in\delta(\S_i)} (d_1(j)-d_2(j))_+ \leq D_1 - \frac{bs_0}{2}\left(\sum_{j:i_2(j)\in\L_2}(d_1(j)-d_2(j)) \right)_+.
   	\end{align*}	
\end{itemize}
We run these two algorithms along with $\A' = \A(a,a,b,a,b,a,b)$, and use the best solution of the three. (Again, note that $\beta = \min\{a,b\} \geq 1/4$ and $a \geq 1/4$ in this case.) We can easily formulate an NLP to derive the approximation ratio as discussed in subsection \ref{sec:cost_analysis}. Using an interval search as before over the interval $0\le s_0 \leq 5/6,b\in[0.508,3/4], r_D\in[19/40,2/3]$, we get an upper-bound of $1.3371$ on the factor-revealing NLP. Note that the value of $g$ is irrelevant to any of these algorithms, so our intervals are over only $b$, $r_D$ and $s_0$. This interval search runs in seconds and examines about 6400 intervals.

\end{proof}

\begin{proof} (Lemma \ref{lem:r1})
We will run the set of 10 algorithms shown in Table \ref{tab:r1}. Obviously, when {\sc Round2Stars} is called (only algorithms $\A_1', A_2',\A_7',\A_8'$), we have $\beta = \min\{q_2,1-q_2\} \geq 5/24$ and $p_2 \geq 5/24$, which are achieved at $p_2 = as_0$, $a=1/4$, $s_0 = 5/6$. Thus, it is easy to check that $\A_1'$, $A_2'$, and $A_4',\ldots, \A_8'$ only open $k+O(\log(1/\eta))$ facilities. For $\A_9'$ and $\A_{10}'$, using the same argument as in Lemma \ref{lem:kfacilities} and Lemma \ref{lem:singleA}, we open at most
\begin{align*}
     |\F_1| + b|\C_1| + O(\log(1/\eta)) &\leq     |\F_1| + (b/r_1)|\C_1| 	+ O(\log(1/\eta))\\
	&= |\F_1| + b\Delta_F + O(\log(1/\eta)) = k + O(\log(1/\eta)),
\end{align*}
where the first inequality is due to the fact that $r_1 \leq 1$ in the interval of interest.

Recall that $s_0 = \frac{1}{1+r_0}=\frac{1}{1+|\C_0|/\Delta_F} = \frac{|\L_2|-|\C_2|-|\C_0|}{|\L_2|-|\C_2|} = 1-\frac{|\C_0|}{|\L_2|-|\C_2|}$. For $\A_3'$, we open at most
\begin{align*}
	|\C_1| + |\C_2| + \frac{1-as_0}{2}|\L_2| + O(\log(1/\eta)) &=  |\C_1| + |\C_2| + \frac{1-a + a\frac{|\C_0|}{|\L_2|-|\C_2|}}{2}|\L_2| + O(\log(1/\eta)) \\
		&= |\C_1| + |\C_2| + \frac{b}{2}|\L_2| + a|\C_0| \frac{|\L_2|}{2(|\L_2|-|\C_2|)} + O(\log(1/\eta)) \\
		&\leq |\C_1| + |\C_2| + \frac{b}{2}|\L_2| + a|\C_0| + O(\log(1/\eta)) \\
		&= |\C_1| + a|\C_2| + b|\C_2| + b|\L_2|-\frac{b}{2}|\L_2| + a|\C_0| + O(\log(1/\eta)) \\
		&\leq |\C_1| + a|\C_2|  + b|\L_2| + a|\C_0| + O(\log(1/\eta)) = k + O(\log(1/\eta)).
\end{align*} 
The approximation ratio will be bounded by an NLP as in our main case. It is simpler, as we need not consider the distinction between $\T_{1A}$ and $T_{1B}$, or the value of $g$. We do an interval search and get an upper-bound of  1.337 when $b \in [0.508,3/4], r_D \in [19/40,2/3]$, and $ s_0 \in [5/6,1]$.
 
\begin{table}[h]
\begin{center}
\begin{tabular}{ c | c | c | c | c | c | c | c }
Algorithms & $p_0$ & $p_{1A}$& $q_{1A}$& $p_{1B}$& $q_{1B}$& $p_2$& $q_2$ \\
  \hline                        
$\A_1'$ & $0$ & $0$& $1$& $0$& $1$& $as_0$& $1-as_0$ \\
$\A_2'$ & $0$ & $1$& $0$& $1$& $0$& $as_0$& $1-as_0$ \\
$\A_3'$ & $0$ & $0$& $1$& $0$& $1$& $1$& $\frac{1-as_0}{2}$ \\
$\A_4'$ & $0$ & $1$& $0$& $1$& $0$& $1$& $\frac{1-as_0}{2}$ \\

$\A_5'$ & $1$ & $0$& $1$& $0$& $1$& $1$& $\frac{bs_0}{2}$ \\
$\A_6'$ & $1$ & $1$& $0$& $1$& $0$& $1$& $\frac{bs_0}{2}$ \\
$\A_7'$ & $1$ & $0$& $1$& $0$& $1$& $1 - bs_0$& $bs_0$ \\
$\A_8'$ & $1$ & $1$& $0$& $1$& $0$& $1 - bs_0$& $bs_0$ \\

$\A_9'$ & $1$ & $1$& $b$& $1$& $b$& $1$& $0$ \\
$\A_{10}'$ & $1$ & $b$& $1$& $b$& $1$& $1$& $0$ \\

  \hline                        
\end{tabular}
\caption{Calls of $\A$ when $r_1 \leq 1$.}
\label{tab:r1}
\end{center}
\end{table}
				  
\end{proof}

\section{Details: Dichotomy result}
\label{sec:dichotomy}
\subsection{Case 1}
For each $i \in \C_2''$, let $X_i$ be an indicator of the event that we open $\S_i$ and close $i$. (If $X_i = 0$, we close $S_i$ and open $i$.) The idea is to set each $X_i = 1$ independently with probability $(1-\eta)q_2$.

\begin{lemma} With probability at least $1 - \exp\left( -\frac{\eta^3(1 - \eta)\beta f(1/\eta) }{ 3c_0 } \right) $, the algorithm opens at most $p_2|\C_2''| + q_2|\L_2''|$ facilities which are part of small stars. \label{lem:prob_violation_small_stars}
\end{lemma}
\begin{proof}
Recall that small stars have at most $c_0 / \eta$ leaves. The number of opened facilities which are part of small stars is 
$$	X = \sum_{i \in \C_2''} (X_i|\S_i| + (1-X_i)) = |\C_2''| + \sum_{i \in \C_2''} X_i(|\S_i| - 1) = |\C_2''| + \frac{c_0}{\eta} \sum_{i \in \C_2''} Y_i,
$$
where $Y_i = \frac{X_i(|\S_i| - 1)}{c_0/\eta}$. Note that $Y_i$'s take random values in $[0,1]$. The expected value of $Y = \sum_{i \in \C_2''} Y_i$ is
\begin{align*}
	\mu &:= \E\left[\sum_{i \in \C_2''} Y_i \right]	 \\
		&= \sum_{i \in \C_2''} \frac{\E[X_i](|\S_i| - 1)}{c_0/\eta} \\
		&= \frac{1 - \eta}{c_0/\eta} \sum_{i \in \C_2''} q_2(|\S_i| - 1) \\
		&= \frac{\eta(1-\eta)}{c_0}q_2(|\L_2''| - |\C_2''|) \\
		&\geq \frac{\eta(1-\eta)}{c_0}q_2 |\C_2''| \geq \frac{\eta(1-\eta)}{c_0}\beta f(1/\eta), 
\end{align*}
since each small star has at least 2 leaves. Using Chernoff's bound, we have
\begin{align*}
	\Pr\left[X > p_2|\C_2''| + q_2|\L_2''| \right] &= \Pr\left[|\C_2''| + \frac{c_0}{\eta} Y > |\C_2''| + q_2(|\L_2''| - |\C_2''|)\right] \\
	&= \Pr\left[ Y > \frac{\eta}{c_0} q_2(|\L_2''| - |\C_2''|)\right] \\
	&= \Pr\left[ Y > \frac{\mu}{1 - \eta} \right] \\
	&\leq \Pr\left[ Y > (1+\eta)\mu \right] \\
	&\leq \exp\left( -\frac{\eta^2}{3}\mu \right)\\
	&\leq \exp\left( -\frac{\eta^3(1 - \eta)\beta f(1/\eta) }{ 3c_0 } \right).
\end{align*}
\end{proof}
To bound the expected connection cost, we need the following lemma.

\begin{lemma} 
Let $i_1$ and $i_2$ be any facilities in $\F_1$ and $\F_2$, respectively. Let $X,Y\in\{0,1_A,1_B,2\}$ be the classes such that $i_1\in\C_X$ and $i_2\in\L_Y$. Then the following are true:
\begin{align}
Pr[\bar i_1]&\le 1-p_X\label{bound:pi1},\\ 
Pr[\bar i_2]&\le \left(1 + \frac{1 - \beta}{\beta} \eta \right)(1-q_Y)\label{bound:pi2},\\
Pr[\bar i_1 \bar i_2]&\le \left(1+\frac{1 - \beta}{\beta} \eta \right)(1-p_X)(1-q_Y).\label{bound:pi1i2}
\end{align}
\label{lem:probbound2}
\end{lemma}
\begin{proof}
The proof is quite similar to that of lemma \ref{lemma:probbound}. 
\begin{itemize}
	\item Proof of (\ref{bound:pi1}): We only need to check the case $i_1 \in \C_2''$. It is clear that 
	$$\Pr[\bar{i_1}] = (1-\eta)q_2 = 1 - p_2 - \eta q_2 \leq 1 - p_2.$$
	\item Proof of (\ref{bound:pi2}): We only need to check the case $i_2 \in \L_2''$. It is clear that 
	$$\Pr[\bar{i_2}] = 1 - (1-\eta)q_2 \leq \left(1 + \frac{1 - \beta}{\beta} \eta\right)(1 - q_2), $$
since $q_2 \leq 1 - \beta$.
	\item Proof of (\ref{bound:pi1i2}): $i_1$ and $i_2$ are always opened independently
	$$ \Pr[\bar{i_1}\bar{i_2}] = \Pr[\bar{i_1}]\Pr[\bar{i_2}] \leq \left(1+\frac{1 - \beta}{\beta} \eta \right)(1-p_X)(1-q_Y).$$

\end{itemize}
\end{proof}

\begin{corollary} The expected connection cost of the solution returned by the algorithm is at most $1.337\cdot\left(1 + \frac{1 - \beta}{\beta} \eta\right)$ times the cost of the bipoint solution.
\end{corollary}

\begin{theorem}
There exists a choice of $f = O(\poly(1/\eta))$ so that, when $|\C_2''| > f(1/\eta)$, the algorithm returns a solution opening at most 4 additional facilities and having connection cost at most $1.3371\cdot\left(1 + \frac{1 - \beta}{\beta} \eta\right)(1+\eta)$ times the cost of the bipoint solution with constant positive probability which is a function of $\eta$. 
\end{theorem}
\begin{proof}
Let $f(1/\eta) = \frac{3c_0}{\eta^3(1-\eta)\beta}\ln \eta^{-2}$. Also let $\mathcal{E}_1$ be the event ``the connection cost is at most  $1.3371\cdot\left(1 + \frac{1 - \beta}{\beta} \eta\right)(1+\eta)$ times the cost of the bipoint solution'' and $\mathcal{E}_2$  be the event ``the algorithm opens at most 4 extra facilities''. By Markov bound,
$$ \Pr[\mathcal{E}_1] \geq 1 - \frac{1}{1+\eta}.$$
Note that at most 4 additional facilities in $\C_2'\cup \L_2'\cup \T_0 \cup \T_{1A} \cup \T_{1B}$ could be opened. By the choice of $f$ and lemma \ref{lem:prob_violation_small_stars},
$$ \Pr[\mathcal{E}_2] \geq 1 - \eta^2.$$
Thus,
\begin{align*}
	\Pr[\mathcal{E}_1 \wedge \mathcal{E}_2] &= \Pr[\mathcal{E}_2] - \Pr[\bar{\mathcal{E}}_1 \wedge \mathcal{E}_2] \\
		&\geq \Pr[\mathcal{E}_2] - \Pr[\bar{\mathcal{E}}_1 ] \\
		&\geq (1 - \eta^2) - \frac{1}{1+\eta} = 1 - \frac{\eta^3 + \eta^2 + 1}{\eta + 1},
\end{align*}
which is strictly greater than zero when $\eta$ is small enough.
\end{proof}

\subsection{Case 2}
Assume $|\L_2'| + |\L_2''| \leq g(1/\eta)$ for some $g = O(\poly(1/\eta))$ to be determined, we simply open all the facilities in $\F_2$. Indeed the number of extra opened facilities is $O(\poly(1/\eta))$; however, the solution has cost equal to $\D_2 < a\D_1 + b\D_2$.

\subsection{Case 3}
If neither Case 1 nor Case 2 holds, we have $|\C_2''| \leq f(1/\eta)$ and $|\L_2'| + |\L_2''| \geq g(1/\eta)$. Since $|\L_2''| \leq \frac{c_0}{\eta}|\C_2''|$, we can bound the number of leaves of large stars
$$ |\L_2'| \geq g(1/\eta) - |\L_2''| \geq g(1/\eta) - \frac{c_0 f(1/\eta)}{\eta}.$$
The ``budget'' to open facilities in the class of large stars is 
$$ p_2|\L_2'| + q_2|\C_2'| = |\C_2'| + q_2(|\L_2'| - |\C_2'|). $$ 
Suppose that we want to open each leaf in $\L_2'$ with probability $q_2(1 - c_1\eta)$ for some constant $c_1 \geq 1$ and open all the centers in $\C_2'$. Then the remaining budget which can be used to open other facilities in $\C_2''$ is 
\begin{align*}
	R &= |\C_2'| + q_2(|\L_2'| - |\C_2'|) - (|\C_2'| + q_2(1 - c_1\eta)|\L_2'| )  \\
	&= c_1\eta q_2 |\L_2'| - q_2|\C_2'|. 
\end{align*}
To open all centers in $\C_2''$, we need to open at most $q_2|\C_2''| \leq q_2 f(1/\eta)$ additional facilities in $\C_2''$, apart from the usual $p_2|\C_2''|$ ones. Thus, it suffices to require that $R \geq q_2 f(1/\eta)$.
\begin{itemize}
	\item If $|\C_2'| \geq \frac{f(1/\eta)}{-1 + c_0c_1}$, recall that $|\L_2'| \geq \frac{c_0}{\eta}|\C_2'|$, we have
	\begin{align*}
		R &\geq c_0 c_1 q_2|\C_2'| - q_2|\C_2'| \\
			&\geq (c_0c_1 - 1)q_2\frac{f(1/\eta)}{-1 + c_0c_1} \\
			&= q_2 f(1/\eta).
	\end{align*}
	\item Else $|\C_2'| < \frac{f(1/\eta)}{-1 + c_0c_1}$, recall that $|\L_2'| \geq g(1/\eta) - \frac{c_0 f(1/\eta)}{\eta}$, we can lower-bound $R$ as follows.
	\begin{align*}
		R &\geq c_1\eta q_2 |\L_2'| - q_2 \frac{f(1/\eta)}{-1 + c_0c_1} \\
		 &\geq c_1\eta q_2 \left( g(1/\eta) - \frac{c_0 f(1/\eta)}{\eta} \right) - q_2 \frac{f(1/\eta)}{-1 + c_0c_1}. 
	\end{align*}
A simple calculation shows that we can choose
	$$g(1/\eta) = \frac{c_0^2 c_1}{\eta(c_0 c_1 - 1)} f(1/\eta),$$
then
	\begin{align*}
		R &\geq c_1\eta q_2 \left(\frac{c_0^2 c_1}{\eta(c_0 c_1 - 1)} f(1/\eta) -  \frac{c_0 f(1/\eta)}{\eta}\right) - q_2 \frac{f(1/\eta)}{-1 + c_0c_1} \\
			&= \frac{c_0^2c_1^2}{c_0c_1-1} q_2f(1/\eta) - c_0c_1q_2f(1/\eta) - \frac{1}{c_0c_1-1}q_2f(1/\eta) \\
			&= \left(\frac{c_0^2c_1^2}{c_0c_1-1} - c_0c_1 - \frac{1}{c_0c_1-1} \right)q_2f(1/\eta) \\
			&= \left(\frac{c_0^2c_1^2 - c_0c_1(c_0c_1-1) - 1}{c_0c_1-1} \right)q_2f(1/\eta) \\
			&= q_2f(1/\eta).
	\end{align*}
\end{itemize}

\begin{theorem} For any polynomial function $f(1/\eta)$, let 
$$ g(1/\eta) = \frac{2}{\eta} f(1/\eta). $$
If $|\C_2''| \leq f(1/\eta)$ and $|\L_2'|+|\L_2''| \geq g(1/\eta)$, then there is an algorithm which returns a solution opening at most $4$ additional facilities and having expected connection cost at most $1.337\cdot (1+\frac{1-\beta}{\beta}\eta)(1+\eta)$ times the cost of the bipoint solution. 
\end{theorem}
\begin{proof}
We simply set $c_0 = 2$ and $c_1 = 1$. As discussed above, there are no extra opened facilities in $\C_2'' \cup \L_2''$. Lemma \ref{lem:probbound2} also holds in this case and can be used to bound the expected connection cost.  
\end{proof}
This implies an approximation algorithm which runs in $O(n^{O(4/\eps)})$ time for $k$-median.

\section{A simple approach to the budgeted MAX-SAT problem}
\label{sec:budgeted-MAX-SAT}
We present the algorithm and give a brief outline of the proof of its approximation guarantee. 
Consider an arbitrary CNF-SAT formula $\phi$, weight function $w$, and budget constraint $\sum_j X_j \leq k$ as in the introduction. 
Given some $\epsilon > 0$ (where we assume w.l.o.g.\ that $\epsilon$ is small, say $\epsilon \leq 0.1$), we aim to approximate this maximization problem to within $(1 - 1/e - \epsilon)$. Motivated by \cite{DBLP:journals/siamdm/GoemansW94}, consider a LP relaxation with a variable $y_j$ to indicate whether $x_j$ is True, and a variable $z_i$ to indicate whether clause $i$ is satisfied. Let $P(i)$ be the set of variables that appear positively in clause $i$, and $N(i)$ be the set of variables that appear negated in clause $i$. The LP relaxation is: maximize $\sum_i w_i z_i$ subject to: 
(i) $\sum_j y_j \leq k$; (ii) $\forall i$, $(\sum_{j \in P(i)} y_j) + (\sum_{j \in N(i)} (1 - y_j)) \geq z_i$; and (iii) $0 \leq y_j, z_i \leq 1$.

Let $\{y^*, z^*\}$ denote an optimal solution to the LP relaxation. 
Note that if we did not have the budget constraints (i), then this MAX SAT-problem can be approximated to within $3/4$ \cite{DBLP:journals/siamdm/GoemansW94}. Also, as pointed out in \cite{DBLP:journals/siamdm/GoemansW94}, if we just do standard (independent) randomized rounding on the $y^*_j$ values, we will get an $(1 - 1/e)$--approximation (again, if we did not have the constraints (i)); this is achieved when the typical clause $i$ has a ``large" number $t$ of literals, each $j \in P(i)$ has $y^*_j = 1/t$, and each $j \in N(i)$ has $y^*_j = 1 - 1/t$. We get a much-better-than-$(1 - 1/e)$--approximation in cases that deviate significantly from this. However, randomized rounding will not preserve (i) with high probability. 
Our problem has a simple solution. If $k \leq 1/\epsilon^3$, say, then find an optimal solution by brute force in $O(n^{1/\epsilon^3})$ time. Else, scale all the $y^*_j$ by $(1 - \epsilon)$ to give us some margin in the budget, and then do independent rounding of the $y^*_j$. In other words, set each variable $x_j$ to True independently with probability $(1-\epsilon)y^*_j$.

\begin{claim} When $k \geq 1/\epsilon^3$, the algorithm produces a feasible solution to Budgeted MAX SAT problem (i.e., the number of variables that are set to True is at most $k$) with probability at least $1 - \exp\left(\frac{1-1/\epsilon}{3}\right)$.
\end{claim}
\begin{proof}
Let $X$ be the number of variables that are set True. Then,
$$ \E[X] = (1-\epsilon)\sum_j y_j^* \leq (1-\epsilon)k. $$
By Chernoff's bound,
\begin{align*}
	\Pr[X>k] &\leq \Pr[X > (1+\epsilon)(1-\epsilon)k] \\
		&\leq \exp\left(-\frac{\epsilon^2}{3}(1-\epsilon)k \right) \leq \exp\left(-\frac{\epsilon^2(1-\epsilon)}{3\epsilon^3} \right) = \exp\left(\frac{1 - 1/\epsilon }{3 } \right).
\end{align*}
\end{proof}

\begin{claim} The expected number of satisfied clauses is at least $(1-1/e -\epsilon)OPT$, where $OPT$ is the optimal number of satisfied clauses to the budgeted MAX SAT problem.
\end{claim}
\begin{proof} (Sketch)
The proof is almost identical to the one in chapter 5.4 of \cite{book:ws}, except that we replace $y^*$ by $(1-\epsilon)y^*$ and $z^*$ by $(1-\epsilon)z^*$. In particular, let $\ell_i$ be the size of clause $i$, then
\begin{align*}
	\Pr[\text{clause } C_i \text{ is not satisfied}] &= \prod_{j \in P(i)}(1-(1-\epsilon)y_j^*) \prod_{j \in N(i)}((1-\epsilon)y_j^*) \\
	&\leq \left[ \frac{1}{\ell_i}\left(\sum_{j \in P(i)}(1-(1-\epsilon)y_j^*) + \sum_{j\in N(i)}(1-\epsilon)y_j^* \right) \right]^{\ell_i} \\
	&= \left[1 - \frac{1}{\ell_i}\left(\sum_{j \in P(i)}(1-\epsilon)y_j^* + \sum_{j\in N(i)}(1-(1-\epsilon)y_j^*) \right) \right]^{\ell_i} \\
	&\leq \left(1 - \frac{(1-\epsilon)z_i^*}{\ell_i} \right)^{\ell_i},
\end{align*}
since, by LP constraint,
\begin{align*}
\sum_{j \in P(i)}(1-\epsilon)y_j^* + \sum_{j\in N(i)}(1-(1-\epsilon)y_j^*) &\geq  \sum_{j \in P(i)}(1-\epsilon)y_j^* + (1-\epsilon)\sum_{j\in N(i)}(1-y_j^*) \\
	&\geq (1-\epsilon)z_i^*.
\end{align*}
Then, by concavity, we have
\begin{align*}
	\Pr[\text{clause } C_i \text{ is satisfied}] &\geq 1 - \left(1 - \frac{(1-\epsilon)z_i^*}{\ell_i} \right)^{\ell_i} \\
		&\geq \left[1-\left(1-\frac{1}{\ell_i} \right)^{\ell_i} \right](1-\epsilon)z_i^*.
\end{align*}
The claim follows by computing the expected number of satisfied clauses and observing that the approximation ratio is at least
\begin{align*}
	\min_{k \geq 1}\left[ 1 - \left(1 - \frac{1}{k} \right)^k \right](1 - \epsilon) &\geq (1-1/e)(1-\epsilon) \\
		&\geq 1 - 1/e - \epsilon.
\end{align*}

\end{proof}

\section{JMS with scaling}
\label{sec:jms_scaling}
In the conference version \cite{ByrkaPRST15} of this paper we have claimed that a scaled version of the primal-dual JMS algorithm
is a (1,1.953)-approximation algorithm for UFL. We called this scaled version JMS' and derived a factor revealing LP to analyze it.
This LP was an adaptation of the factor revealing LP from~\cite{jain2003greedy}. Unfortunately we overlooked a subtlety in the 
interpretation of $r_{j,i}$ variables present in this LP. This would not have had any effect on the numbers resulting from a standard analysis of JMS,
but it turned out to be crucial for analyzing the scaled version JMS'. Below we formally describe the algorithm, its faulty analysis,
and an instance showing that JMS' is in fact not better than JMS and in particular it is not a (1,1.953)-approximation algorithm as we wrongly claimed in \cite{ByrkaPRST15}.

\subsection{The JMS algorithm}
\label{jms}
In this section, we review the JMS algorithm by \cite{jms}. The more detailed description of the algorithm and its analysis can be found in \cite{jms}. For completeness, we briefly describe the algorithm here. Each client has a budget, initially equal to zero. The final value of the budget represents the amount of money which client pays in the resulting solution. The budget of each client increases up to the moment of the first connection of the client. From this time the value of the budget will not be changed. The client can only be reconnected to a closer, newly open, facility and spend the difference in connection distance for the (part of the) opening cost of the newly open facility.
\begin{enumerate}
 \item We use a notion of time. The algorithm starts at time 0. At this time, each client is defined to be unconnected ($U := \J$), all facilities are unopened, and budget $\alpha_j$ is set to 0 for every $j$. At every moment, each client $j$ offers some money from its budget as a contribution to open an unopened facility $i$. The amount of this offer is computed as follows: If $j$ is unconnected, the offer is equal to 
$\max(\alpha_j - d_{ij} , 0)$ (i.e., if the budget of $j$ is more than the cost that it has to pay to get connected to $i$, it offers to pay
this extra amount to $i$); If $j$ is already connected to some other facility $i'$ , then its offer to facility $i$ is equal to $\max\{d_{i'j} - d_{ij}, 0\}$ (i.e., the amount that $j$ offers to pay to $i$ is equal to the amount $j$ would save by switching its facility from $i'$ to $i$).

\item While $U \neq \emptyset$, increase the time, and simultaneously, for every client $j \in U$, increase the parameter
$\alpha_j$ at the same rate, until one of the following events occurs (if two events occur at the same time,
we process them in an arbitrary order).
\begin{itemize}
 \item For some unopened facility $i$, the total offer that it receives from clients is equal to the cost of
opening $i$. In this case, we open facility $i$, and for every client $j$ (connected or unconnected)
which has a nonzero offer to $i$, we connect $j$ to $i$ and remove $j$ from $U$. The amount that $j$ had offered to $i$ is now
called the contribution of $j$ toward $i$, and $j$ is no longer allowed to decrease this contribution.
\item For some unconnected client $j$, and some open facility $i$, $\alpha_j = d_{ij}$ . In this case, connect client $j$
to facility $i$ and remove $j$ from $U$. 
\end{itemize}
\end{enumerate}

Suppose the UFL instance to be solved is being designed by an adversary.
Consider an optimal solution to the instance being solved.
The optimal solution opens a certain set of facilities and the clients get partitioned into groups served
by a single facility in the optimal solution. In the analysis of the JMS algorithm we model
the behavior of the algorithm on a single such group of clients served by a single facility in the optimal solution.
If we manage to bound the cost incurred by these clients in the computed solution in terms
of the cost these clients have in the optimal solution, and if we do so for every possible such group of clients,
we obtain an estimate on the approximation ratio of the JMS algorithm.

The approximation factor of the JMS algorithm is upper bounded by the supremum of set $\{b_k : k \in \mathbb{N} \setminus \{0\}\}$, where $b_k$ is defined as the objective function of the factor revealing LP (\ref{jms:max}) - (\ref{jms:non-negative}). Moreover, we will prove that, $b_k$ is bounded from above by 1.61, for each value of $k$.

Consider a group of $k$ clients that in the optimal solution use a single facility $f$ (with a slight misuse of the notation, we call this facility $f$ and we refer to its cost also by $f$). Variables from the LP may be interpreted as follows: variable $\alpha_j$ is the dual budget of client $j$ (at the end of the algorithm), which is the time when $j$ is connected for the first time. Variable $r_{j,i}$ is the connection cost of a client $j$ \textbf{just before} client $i$ is being connected for the first time (time $\alpha_i - \epsilon$). In the case when client $j$ is not connected in time $\alpha_i - \epsilon$ the value of $r_{j,i}$ is equal to $\alpha_j$, which in this case equals $\alpha_i$. Variable $d_j = d(j, f)$ is the connection cost of $j$ in the optimal solution and $f$ (the opening cost of) is the facility used in the optimal solution. The analysis compares the cost of an algorithm on the considered set of clients to the cost of the optimal solution.

\begin{eqnarray}
  \label{jms:max}
  b_k = \max~\frac{\sum_{i = 1}^{k} \alpha_i}{f + \sum_{i = 1}^{k} d_i} && \\
  \label{jms:budgets_order}
  \alpha_i \leq \alpha_{i+1} &&  ~~~~~~\forall_{i < k}\\
  \label{jms:decrease_connection}
  r_{j,i} \geq r_{j,i+1} &&  ~~~~~~\forall_{j \leq i < k}\\
  \label{jms:triangle_ineq}
  \alpha_i \leq r_{j, i} + d_i + d_j &&  ~~~~~~\forall_{j < i \leq k}\\
  \label{jms:budget_connection}
  r_{i,i} \leq \alpha_i &&  ~~~~~~\forall_{i \leq k}\\
  \label{jms:f_lowerbound}
  \sum_{j = 1}^{i-1} \max\{r_{j,i} - d_j, 0\} + \sum_{j = i}^{k} \max\{\alpha_i - d_j, 0\} \leq f&&  ~~~~~~\forall_{i \leq k}~~~~~\\
  \label{jms:non-negative}
  \alpha_i, r_{j,i}, d_i, f \geq 0 &&  ~~~~~~\forall_{j \leq i \leq k}
\end{eqnarray}

\subsection{The JMS' algorithm}
We present a new algorithm as a variant of the JMS algorithm for UFL, and hence place our emphasis on differences.

A reasonable modification to the JMS algorithm would be to apply scaling to instances before running the algorithm. One could, for instance, alter the JMS algorithm
by feeding it an instance with distances scaled up by a certain fixed factor. This seems to be a reasonable solution but we find it difficult to analyze in general,
since the standard factor-revealing LP method does not capture some directions of scaling (i.e., we cannot benefit from the connections of some clients remaining ``far" from open facilities, because they may get closer centers later by simply switching to facilities opened later). To overcome this problem, we give an algorithm
in which clients are less eager to contribute to facility opening when they switch to closer facilities. 

One may think that a client $j$ is supposed to pay a certain tax on the amount saved by switching to a closer facility, and only offers the rest as a contribution toward opening the new facility. This additional tax
in a combination with scaled contribution of all, not yet connected clients, results in an algorithm for which we are able to prove an improved bound on the approximation ratio.

Each yet unconnected client scale the real value of its contribution by a factor of $\gamma$. This algorithm, denoted JMS'($\gamma$), is formally described below.

Define $U$ to be the set of yet-unconnected clients; $U := \J$ initially. Let $\gamma \geq 1$ be parameters of the JMS' algorithm. In \textbf{bold} we mark the differences in the algorithm as compared to the standard JMS algorithm.

Algorithm:
\begin{enumerate}
 \item We use a notion of time. The algorithm starts at time 0. At this time, each client is defined to be unconnected ($U := \J$), all facilities are unopened, and budget $\alpha_j$ is set to 0 for every $j$. At every moment, each client $j$ offers some money from its budget as a contribution to open an unopened facility $i$. The amount of this offer is computed as follows: If $j$ is unconnected, the offer is equal to 
{\boldmath $\gamma \cdot$}$\max(\alpha_j - d_{ij} , 0)$ (i.e., if the budget of $j$ is more than the cost that it has to pay to get connected to $i$, it offers to pay {\boldmath $\gamma$ \bf times}
this extra amount to $i$); If $j$ is already connected to some other facility $i'$ , then its offer to facility $i$ is equal to $\max\{d_{i'j} - d_{ij}, 0\}$ (i.e., the amount that $j$ offers to pay to $i$ is equal to the amount $j$ would save by switching its facility from $i'$ to $i$.

\item While $U \neq \emptyset$, increase the time, and simultaneously, for every client $j \in U$, increase the parameter
$\alpha_j$ at the same rate, until one of the following events occurs (if two events occur at the same time,
we process them in an arbitrary order).
\begin{itemize}
 \item For some unopened facility $i$, the total offer that it receives from clients is equal to the cost of
opening $i$. In this case, we open facility $i$, and for every client $j$ (connected or unconnected)
which has a nonzero offer to $i$, we connect $j$ to $i$ and remove $j$ from $U$. The amount that $j$ had offered to $i$ is now
called the contribution of $j$ toward $i$, and $j$ is no longer allowed to decrease this contribution.
\item For some unconnected client $j$, and some open facility $i$, $\alpha_j = d_{ij}$ . In this case, connect client $j$
to facility $i$ and remove $j$ from $U$. 
\end{itemize}
\end{enumerate}

\subsection{Our incorrect analysis}
When analyzing JMS' in \cite{ByrkaPRST15}, we also used a factor revealing LP. In this LP, we interpret the variables $r_{j,i}$ slightly differently than in Section \ref{jms}. In \cite{ByrkaPRST15} $r_{j,i}$ was the connection cost of client $j$ \textbf{precisely} at time $\alpha_i$ (when client $i$ is being connected for the first time). We analyzed this algorithm by the following factor-revealing LP.

\begin{eqnarray}
  \label{lb:max}
  z_k = \max~\frac{\sum_{i = 1}^{k} \gamma \alpha_i + (1 - \gamma) r_{i,i} -f}{\sum_{i = 1}^{k} d_i} && \\
  \label{lb:budgets_order}
  \alpha_i \leq \alpha_{i+1} &&  ~~~~~~\forall_{i < k}\\
  \label{lb:decrease_connection}
  r_{j,i} \geq r_{j,i+1} &&  ~~~~~~\forall_{j \leq i < k}\\
  \label{lb:triangle_ineq}
  \alpha_i \leq r_{j, i} + d_i + d_j &&  ~~~~~~\forall_{j < i \leq k}\\
  \label{lb:budget_connection}
  r_{i,i} \leq \alpha_i &&  ~~~~~~\forall_{i \leq k}\\
  \label{lb:f_lowerbound}
  \sum_{j = 1}^{i-1} \max\{r_{j,i} - d_j, 0\} + \gamma \sum_{j = i}^{k} \max\{\alpha_i - d_j, 0\} \leq f&&  ~~~~~~\forall_{i \leq k}~~~~~\\
  \label{lb:non-negative}
  \alpha_i, r_{j,i}, d_i, f \geq 0 &&  ~~~~~~\forall_{j \leq i \leq k}
\end{eqnarray}

\paragraph{Issue with the analysis} It can happen that client $i$ contributes towards the opening of facility $f'$ to which $j$ has distance $r_{j,j} = r_{j,i}$. Then the inequality $\alpha_i \leq d_i + d_j + r_{j,i}$  need not hold. Note that this issue follows from our new interpretation of the variables $r_{j,i}$. If we are using the interpretation of Section \ref{jms} then the constraints are satisfied but the term $r_{i,i}$ in the objective function does not necessarily represent the connection cost at time $\alpha_i$ but is equal to $\alpha_i$, which might be strictly larger than this connection cost.

\paragraph{Counterexample for the analysis}
Consider the feasible solution to the factor-revealing LP for any $k$ shown in the following figure.
\begin{figure}
  \centering
  \includegraphics{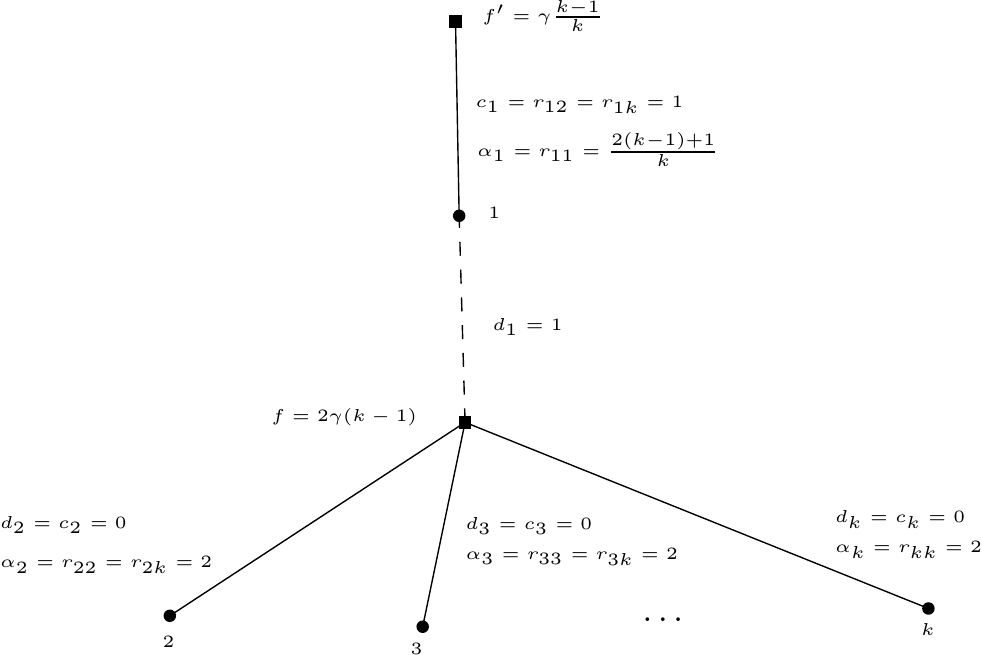}
  \caption{Single star counterexample}
\end{figure}

We interpret the variables as described in Section \ref{jms}. The distance in which client $i$ was connected first is denoted by $c_i$. One can observe that in the above example client $1$ was connected for the first time at distance $c_1 = 1$ with facility $f'$, but $\alpha_1 = r_{1,1} = \frac{2(k-1)+1}{k}$ as client $1$ contributes $\gamma(\frac{2(k-1)+1}{k} - 1)$ towards the opening of facility $f'$. All the other clients (from $2$ to $k$) were connected in distance $c_i = 0$ with facility $f$, but $\alpha_i = r_{i,i} = 2$ as each client $i$ contributes $2\gamma$ towards the opening of facility $f$.

One can observe that in the example $r_{i,i} > c_i$ which occurs with negative coefficient in the objective function, which implies that the cost of our algorithm in this example was underestimated. The real cost of our algorithm for this example is precisely expressed by 
$$\frac{\sum_{i = 1}^{k}( \gamma \alpha_i + (1 - \gamma) c_i) -f}{\sum_{i = 1}^{k} d_i} = 
\gamma\frac{2(k-1)-1}{k} + 1 - \gamma +2\gamma(k-1) - 2\gamma(k-1) \overset{k \rightarrow \infty}{\longrightarrow} 1 + \gamma.$$

\paragraph{Counterexample for the algorithm}
The instance from the previous paragraph can be extended by combining $2k$ copies of this instance that ''share`` the problematic facility $f'$. Thereby we achieve that all facilities have cost $2\gamma(k-1)$. In the solution returned by JMS'($\gamma$) clients $1^{(1)} \cdots 1^{(2k)}$ open together facility $f'$ and all others clients $2^{(l)}, \dots, k^{(l)}$ from each star open facility $f^{(l)}$, for each $l = 1 \dots 2k$. Algorithm JMS'($\gamma$) returns a solution with all facilities open and in the optimal solution all facilities except $f'$ are open. This proves that JMS'($\gamma$) cannot achieve bi-factor better than $(1, 1 + \gamma)$, if we insist on the facility factor to be one. Using the fact that $\gamma \geq 1$ we show that JMS'($\gamma$) does not perform better than JMS.

\begin{figure}
  \centering
  \includegraphics{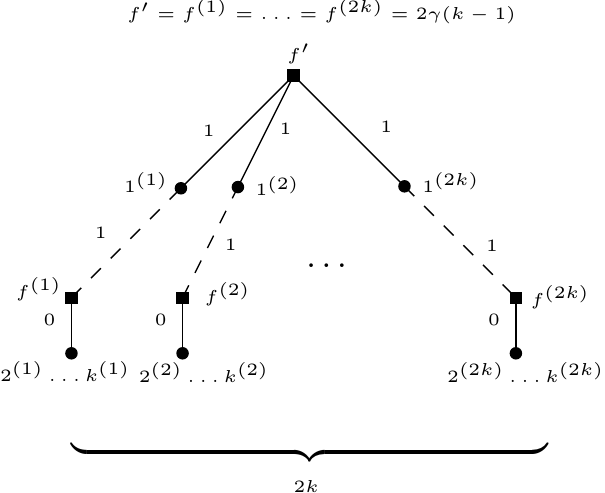}
  \caption{Complete counterexample instance with uniform facility opening cost}
\end{figure}

\end{document}